\documentclass[%
  format=acmsmall,%
  review=false,%
  anonymous=false,%
  timestamp=false,%
  screen=true
]{acmart}

\usepackage{graphicx}
\usepackage{grffile}
\usepackage{longtable}
\usepackage{wrapfig}
\usepackage{rotating}
\usepackage[normalem]{ulem}
\usepackage{amsmath}
\usepackage{textcomp}
\usepackage{amssymb}
\usepackage{capt-of}
\usepackage{hyperref}

\usepackage[scaled]{beramono}
\usepackage[T1]{fontenc}
\usepackage[utf8]{inputenc}
\usepackage{mathtools}
\usepackage{amsopn}
\usepackage{enumitem}
\usepackage{xcolor}
\usepackage{colortbl}
\usepackage{semantic}
\setpremisesend{0.5em}

\usepackage{framed}
\usepackage{amsthm}
\usepackage{relsize}
\usepackage{pdflscape}
\usepackage{todonotes}
\usepackage{tabu}
\usepackage[capitalize]{cleveref}
\crefformat{evcanonii}{(#2#1#3)}
\crefformat{evcanoniii}{(#2#1#3)}
\crefname{equation}{}{}
\crefname{enumi}{}{}
\crefformat{enumi}{(#2#1#3)}
\crefname{enumii}{Case}{Cases}
\makeatletter
\newcommand{\eqlabel}[1]{\refstepcounter{equation}\label{#1}\textup{\tagform@{\theequation}}}
\makeatother
\newcounter{ProofCase}
\crefname{ProofCase}{Case}{Cases}

\newcommand{\pcase}[1]{Case~\refstepcounter{ProofCase}\label{#1}\arabic{ProofCase}}

\newcommand{\Rthree}{(R3)}

\definecolor{light-gray}{gray}{0.8}
\newcommand{\new}[1]{#1}

\newcommand{\fv}[1]{\textit{fv}\hspace{1pt}(#1)}

\newcommand{\ftv}[1]{\textit{ftv}\hspace{1pt}(#1)}
\newcommand{\fdv}[1]{\textit{fdv}\hspace{1pt}(#1)}
\newcommand{\dom}[1]{\textit{dom}\hspace{1pt}(#1)}
\newcommand{\cev}[1]{\textit{canonEv}\hspace{1pt}(#1)}

\newcommand{\TC}{\textit{TC}}
\newcommand{\QQ}{\mathcal{Q}}

\newcommand{\ty}{\colon{}}
\newcommand{\dict}[1]{#1\texttt{.Dict}}
\newcommand{\dictapp}[2]{\texttt{@\{}#1~\texttt{as}~#2\texttt{\}}}
\newcommand{\tev}[1][]{t_{\textit{ev}#1}}
\newcommand{\doubleplus}{+\kern-1.3ex+\kern0.8ex}

\newcommand{\decelab}[5]{#1\,;\,#2\vdash{}#3\ty{}#4\,\new{\rightsquigarrow{}#5}}
\newcommand{\decelabs}[5]{#1\,;\,#2\vdash{}\hspace*{-1ex}\raisebox{1.3ex}{\scriptsize\textit{spec}}\hspace*{1.2ex}#3\ty{}#4\,\new{\rightsquigarrow{}#5}}
\newcommand{\dec}[4]{#1\,;\,#2\vdash{}#3\ty{}#4}
\newcommand{\dectop}[4]{#1\,;\,#2\vdash{}\hspace*{-1.5ex}\raisebox{1.7ex}{\scriptsize\textit{top}}\,#3\ty{}#4}
\newcommand{\dectopelab}[5]{\dectop{#1}{#2}{#3}{#4}\,\new{\rightsquigarrow{}#5}}
\newcommand{\elab}[3]{#2\rightsquigarrow{}\hspace*{-2.3ex}\raisebox{1.3ex}{\scriptsize{}${#1}$}\hspace*{1.2ex}#3}
\newcommand{\elabqg}[3]{#1\,;\,#2\hspace{1.0ex}\rightsquigarrow{}\hspace*{-3.6ex}\raisebox{1.3ex}{\scriptsize{}$Q\,;\,\Gamma$}\hspace*{0.9ex}#3}
\newcommand{\elabqqg}[3]{#1\,;\,#2\hspace{1.0ex}\rightsquigarrow{}\hspace*{-3.6ex}\raisebox{1.3ex}{\scriptsize{}$\QQ\,;\,\Gamma$}\hspace*{0.9ex}#3}
\newcommand{\entails}[2]{#1\,\Vdash\,#2}
\newcommand{\entailselab}[3]{#1\,\Vdash\,#2\,\new{;\,#3}}
\newcommand{\nentails}[2]{#1\,\nVdash\,#2}
\newcommand{\sftc}[3]{#1\vdash{}#2\ty{}#3}
\newcommand{\tred}[2]{#1 \approx #2}
\newcommand{\taueq}[3]{#2 \simeq\hspace*{-1.8ex}\raisebox{1.3ex}{\scriptsize{}${#1}$}\hspace*{1.2ex} #3}

\newcommand{\mg}[4]{#1\,;\,#2 \geq #3\,;\,#4}
\newcommand{\eqq}{:\!:=}

\newcommand{\introforall}{\forall\textsc{i}}
\newcommand{\elimforall}{\forall\textsc{e}}
\newcommand{\introct}{\Rightarrow\!\!\textsc{i}}
\newcommand{\elimct}{\Rightarrow\!\!\textsc{e}}
\newcommand{\betared}{$\beta$-\textsc{Reduction}}
\newcommand{\etared}{$\eta$-\textsc{Reduction}}

\let\phi\varphi

\setlist[itemize]{noitemsep}

\newtheorem{definition}{Definition}

\newtheorem{example}{Example}[section]

\setcopyright{none}
\acmConference{}{}{}
\acmBooktitle{}
\acmYear{}
\acmVolume{}
\acmNumber{}
\acmMonth{}
\acmISBN{}
\acmDOI{}
\acmPrice{}
\acmISBN{}
\acmBadgeR{}
\startPage{1}
\authorsaddresses{}
\renewcommand\footnotetextcopyrightpermission[1]{} %
\newcommand{\authorr}[2]{\author{#1}
\affiliation[obeypunctuation=true]{%
  \institution{imec-DistriNet, KU Leuven},
  \country{Belgium}
}
\email{#2@cs.kuleuven.be}}

\allowdisplaybreaks

\makeatletter
\@ifundefined{lhs2tex.lhs2tex.sty.read}%
  {\@namedef{lhs2tex.lhs2tex.sty.read}{}%
   \newcommand\SkipToFmtEnd{}%
   \newcommand\EndFmtInput{}%
   \long\def\SkipToFmtEnd#1\EndFmtInput{}%
  }\SkipToFmtEnd

\newcommand\ReadOnlyOnce[1]{\@ifundefined{#1}{\@namedef{#1}{}}\SkipToFmtEnd}
\usepackage{amstext}
\usepackage{amssymb}
\DeclareFontFamily{OT1}{cmtex}{}
\DeclareFontShape{OT1}{cmtex}{m}{n}
  {<5><6><7><8>cmtex8
   <9>cmtex9
   <10><10.95><12><14.4><17.28><20.74><24.88>cmtex10}{}
\DeclareFontShape{OT1}{cmtex}{m}{it}
  {<-> ssub * cmtt/m/it}{}
\newcommand{\texfamily}{\fontfamily{cmtex}\selectfont}
\DeclareFontShape{OT1}{cmtt}{bx}{n}
  {<5><6><7><8>cmtt8
   <9>cmbtt9
   <10><10.95><12><14.4><17.28><20.74><24.88>cmbtt10}{}
\DeclareFontShape{OT1}{cmtex}{bx}{n}
  {<-> ssub * cmtt/bx/n}{}

\newcommand{\Conid}[1]{\mathit{#1}}
\newcommand{\Varid}[1]{\mathit{#1}}
\newcommand{\anonymous}{\kern0.06em \vbox{\hrule\@width.5em}}

\renewcommand{\geq}{\geqslant}
\usepackage{polytable}

\@ifundefined{mathindent}%
  {\newdimen\mathindent\mathindent\leftmargini}%
  {}%

\def\resethooks{%
  \global\let\SaveRestoreHook\empty
  \global\let\ColumnHook\empty}
\newcommand*{\savecolumns}[1][default]%
  {\g@addto@macro\SaveRestoreHook{\savecolumns[#1]}}
\newcommand*{\restorecolumns}[1][default]%
  {\g@addto@macro\SaveRestoreHook{\restorecolumns[#1]}}
\newcommand*{\aligncolumn}[2]%
  {\g@addto@macro\ColumnHook{\column{#1}{#2}}}

\resethooks

\newcommand{\onelinecommentchars}{\quad-{}- }
\newcommand{\commentbeginchars}{\enskip\{-}
\newcommand{\commentendchars}{-\}\enskip}

\newcommand{\visiblecomments}{%
  \let\onelinecomment=\onelinecommentchars
  \let\commentbegin=\commentbeginchars
  \let\commentend=\commentendchars}

\newcommand{\invisiblecomments}{%
  \let\onelinecomment=\empty
  \let\commentbegin=\empty
  \let\commentend=\empty}

\visiblecomments

\newlength{\blanklineskip}
\newcommand{\hsindent}[1]{\quad}%
\let\hspre\empty
\let\hspost\empty

\EndFmtInput
\makeatother
\ReadOnlyOnce{polycode.fmt}%
\makeatletter

\newcommand{\hsnewpar}[1]%
  {{\parskip=0pt\parindent=0pt\par\vskip #1\noindent}}

\newcommand{\hscodestyle}{}

\newcommand{\sethscode}[1]%
  {\expandafter\let\expandafter\hscode\csname #1\endcsname
   \expandafter\let\expandafter\endhscode\csname end#1\endcsname}

  {\par\noindent
   \advance\leftskip\mathindent
   \hscodestyle
   \let\\=\@normalcr
   \let\hspre\(\let\hspost\)%
   \pboxed}%
  {\endpboxed\)%
   \par\noindent
   \ignorespacesafterend}

  {\hsnewpar\abovedisplayskip
   \advance\leftskip\mathindent
   \hscodestyle
   \let\hspre\(\let\hspost\)%
   \pboxed}%
  {\endpboxed%
   \hsnewpar\belowdisplayskip
   \ignorespacesafterend}

  {\hsnewpar\abovedisplayskip
   \advance\leftskip\mathindent
   \hscodestyle
   \let\\=\@normalcr
   \(\pboxed}%
  {\endpboxed\)%
   \hsnewpar\belowdisplayskip
   \ignorespacesafterend}

\newcommand{\plainhs}{\sethscode{plainhscode}}

\plainhs

  {\hsnewpar\abovedisplayskip
   \advance\leftskip\mathindent
   \hscodestyle
   \let\\=\@normalcr
   \(\parray}%
  {\endparray\)%
   \hsnewpar\belowdisplayskip
   \ignorespacesafterend}

  {\parray}{\endparray}

  {\(\parray}{\endparray\)}

\def\codeframewidth{\arrayrulewidth}
\RequirePackage{calc}

  {\parskip=\abovedisplayskip\par\noindent
   \hscodestyle
   \arrayrulewidth=\codeframewidth
   \tabular{@{}|p{\linewidth-2\arraycolsep-2\arrayrulewidth-2pt}|@{}}%
   \hline\framedhslinecorrect\\{-1.5ex}%
   \let\endoflinesave=\\
   \let\\=\@normalcr
   \(\pboxed}%
  {\endpboxed\)%
   \framedhslinecorrect\endoflinesave{.5ex}\hline
   \endtabular
   \parskip=\belowdisplayskip\par\noindent
   \ignorespacesafterend}

\newcommand{\framedhslinecorrect}[2]%
  {#1[#2]}

  {\(\def\column##1##2{}%
   \let\>\undefined\let\<\undefined\let\\\undefined
   \newcommand\>[1][]{}\newcommand\<[1][]{}\newcommand\\[1][]{}%
   \def\fromto##1##2##3{##3}%
   }{\) }%

  {\let\orighscode=\hscode
   \let\origendhscode=\endhscode
   \def\endhscode{\def\hscode{\endgroup\def\@currenvir{hscode}\\}\begingroup}
   \orighscode\def\hscode{\endgroup\def\@currenvir{hscode}}}%
  {\origendhscode
   \global\let\hscode=\orighscode
   \global\let\endhscode=\origendhscode}%

\makeatother
\EndFmtInput
\ReadOnlyOnce{forall.fmt}%
\makeatletter

\let\HaskellResetHook\empty
\newcommand*{\AtHaskellReset}[1]{%
  \g@addto@macro\HaskellResetHook{#1}}
\newcommand*{\HaskellReset}{\HaskellResetHook}

\newcommand\hsforall{\global\let\hsdot=\hsperiodonce}
\newcommand*\hsperiodonce[2]{#2\global\let\hsdot=\hscompose}
\newcommand*\hscompose[2]{#1}

\AtHaskellReset{\global\let\hsdot=\hscompose}

\HaskellReset

\makeatother
\EndFmtInput
\ReadOnlyOnce{exists.fmt}%
\makeatletter

\newcommand\hsexists{\global\let\hsdot=\hsperiodonce}

\AtHaskellReset{\global\let\hsdot=\hscompose}

\HaskellReset

\makeatother
\EndFmtInput

\begin{document}

\title{Coherent Explicit Dictionary Application for Haskell}
\subtitle{Formalisation and Coherence Proof}
\authorr{Thomas Winant}{thomas.winant}
\authorr{Dominique Devriese}{dominique.devriese}

\maketitle
\thispagestyle{plain}

\section{Introduction}%
\label{sec:introduction}

This report accompanies \cite{ThePaper} and describes the full formalisation and coherence proof of the system described in \cite{ThePaper}.
We refer to \cite{ThePaper} for a more thorough introduction, motivation, and discussion.

The formalisation of Haskell we present here is based on~\cite{OutsideIn,AssociatedTypeSynonyms}.
Besides type-checking rules, we also present \emph{elaboration} rules, which detail the translation from the source language to the target language, System~F\@.
In practice, GHC's target language is GHC Core, but as our formalisation does not include local assumptions, kinding, type families, etc., we restrict ourselves to System~F.
During elaboration things implicit in the source language, e.g.\ the passing of dictionaries (dictionary translation), and type abstraction and application, are made explicit in the target language.
We use the term \emph{evidence} for dictionary instances in the target language.

We first present the simple target language, then the source language along with the elaboration from the latter to the former.

Both languages have the following syntax in common:
\[
\begin{tabu} to\linewidth{lX}
  x, y, f & \textit{Variables}                     \\
  a, b    & \textit{Type variables (skolems)}      \\
  \TC     & \textit{Type classes}
\end{tabu}
\]
A type class (\TC) does not include its type arguments, e.g., \ensuremath{\Conid{Eq}} is a $\TC$, \ensuremath{\Conid{Eq}\;\Varid{a}} not.
For simplicity, we only consider type classes with exactly one argument.

\section{Target Language}%
\label{sec:target-lang}

\begin{figure}[tb]
\begin{framed}
\[
\begin{array}{lrlr}
  \upsilon       & \eqq & a \mid \upsilon_1 \rightarrow \upsilon_2 \mid \forall\,a.\,\upsilon \mid \dict{\TC}~\upsilon         & \textit{Type}                             \\
  t       & \eqq & x \mid \lambda{}(x\ty{}\upsilon).\,t \mid t_1~t_2 \mid \Lambda{}a.\,t \mid t~\upsilon & \textit{Term}                             \\
  \tev    & \eqq & d \mid \tev~~\overline{\upsilon}~~\overline{\tev}         & \textit{Evidence Term}                    \\
  \Gamma_\upsilon     & \eqq & \epsilon \mid (x\ty{}\upsilon), \Gamma_\upsilon \mid a, \Gamma_\upsilon                     & \textit{Typing environment}               \\
  \eta       & \eqq & [\cdot] \mid [d \mapsto \tev, \eta]                             & \textit{Dictionary evidence substitution} \\
                                                                                                               \\
  d       &      &                                                 & \textit{Dictionary evidence variables}    \\
  \fdv{\cdot} &      &                                                 & \textit{Free dictionary variables}        \\
\end{array}
\]
\end{framed}
\caption{Syntax of the target language}%
\label{fig:syntax-target-lang}
\end{figure}

\cref{fig:syntax-target-lang} shows the mostly standard syntax of the target language, System~F.
The type $\dict{\TC}~\upsilon$ is the type of the dictionary record corresponding to a type class $\TC~\upsilon$.
Note that dictionary evidence variables ($d$) are also variables, we simply use a different letter for clarity.
An evidence term $\tev$ should be seen as a regular term of a certain format: either a dictionary variable $d$, e.g., \ensuremath{\Varid{eqInt}} with type \ensuremath{\Conid{Eq}\;\Conid{Int}}, or an application of types and other evidence terms to an evidence term, e.g., \ensuremath{\Varid{eqList}\;\Conid{Int}\;\Varid{eqInt}} where \ensuremath{\Varid{eqList}} has type \ensuremath{\forall \Varid{a}\hsforall .\;\Conid{Eq}\;\Varid{a}\Rightarrow \Conid{Eq}\;[\mskip1.5mu \Varid{a}\mskip1.5mu]}.
Dictionary evidence substitution ($\eta$) will be used in the proof, see \cref{sec:coherence-proof}, to substitute dictionary evidence with evidence terms.
An empty substitution is written as $[\cdot]$.

\begin{figure}[tb]
\begin{framed}
\raggedright\boxed{\sftc{\Gamma_\upsilon}{t}{\upsilon}}
\vspace{-1em}
\[
\inference{(x\ty{}\upsilon) \in \Gamma_\upsilon}{\sftc{\Gamma_\upsilon}{x}{\upsilon}}[\textsc{Var}]
\]
\\
\[
\inference{\sftc{(x\ty{}\upsilon), \Gamma_\upsilon}{t}{\upsilon'}}
{\sftc{\Gamma_\upsilon}{\lambda{}(x\ty{}\upsilon).\,t}{\upsilon \rightarrow \upsilon'}}[\textsc{Abs}]
\hspace{4em}
\inference{\sftc{\Gamma_\upsilon}{t_1}{\upsilon \rightarrow \upsilon'} & \sftc{\Gamma_\upsilon}{t_2}{\upsilon}}
{\sftc{\Gamma_\upsilon}{t_1~t_2}{\upsilon'}}[\textsc{App}]
\]
\\[0.5em]
\[
\inference{\sftc{a, \Gamma_\upsilon}{t}{\upsilon}}
{\sftc{\Gamma_\upsilon}{\Lambda{}a.\,t}{\forall{}a.\,\upsilon}}[\textsc{TyAbs}]
\hspace{4em}
\inference{\sftc{\Gamma_\upsilon}{t}{\forall{}a.\,\upsilon'}}
{\sftc{\Gamma_\upsilon}{t~\upsilon}{[a \mapsto \upsilon]\upsilon'}}[\textsc{TyApp}]
\]
\end{framed}
\caption{Declarative typing rules of the target language}%
\label{fig:typing-rules-target-lang}
\end{figure}

\cref{fig:typing-rules-target-lang} shows the standard declarative typing rules for our target language, System F.

\section{Source Language Syntax}%
\label{sec:source-lang}

\begin{figure}[tb]
\begin{framed}
\[
\begin{array}{lrlr}
  \tau       & \eqq & a \mid \tau_1 \rightarrow \tau_2 \mid \dict{\TC}~\tau                      & \textit{Monotype}               \\
  \rho       & \eqq & C \Rightarrow \rho \mid \tau                                         & \textit{Qualfied type}          \\
  \sigma       & \eqq & \forall{}a.\,\sigma \mid \rho                                      & \textit{Type scheme}            \\
  C       & \eqq & \TC~\tau                                             & \textit{Type class Constraint}  \\
  Q       & \eqq & \epsilon \mid Q_1 \wedge Q_2 \mid \new{\tev \ty} C                  & \textit{Constraints}            \\
  A       & \eqq & \forall\,\overline{a}.~\overline{C} \Rightarrow C                 & \textit{Axiom scheme}           \\
  \QQ     & \eqq & \epsilon \mid \QQ_1 \wedge \QQ_2 \mid \new{\tev \ty} A              & \textit{Top-level axiom scheme} \\
  e       & \eqq & x \mid \lambda{}x.\,e \mid e_1~e_2 \mid e_1~\dictapp{e_2}{\TC~a} & \textit{Term}                   \\
  \Gamma       & \eqq & \epsilon \mid (x\ty{}\sigma), \Gamma \mid a, \Gamma                           & \textit{Typing environment}     \\
  \theta       & \eqq & [\cdot] \mid [a \mapsto \tau, \theta]                                  & \textit{Type substitution}      \\
                                                                                                       \\
  \ftv{\cdot} &      &                                                   & \textit{Free type variables}    \\

\end{array}
\]
\end{framed}
\caption{Syntax of the source language}%
\label{fig:syntax-source-lang}
\end{figure}

\cref{fig:syntax-source-lang} shows the syntax of the source language, based on~\cite{OutsideIn,AssociatedTypeSynonyms}.
We omit parts that are not important with regards to explicit dictionary application: \ensuremath{\mathbf{case}} expressions, \ensuremath{\mathbf{let}} bindings, equality constraints, etc.
We define $Q$ as a tree of conjunctions with $\epsilon$ as the empty leaf.
The structure and order of the tree is irrelevant and any conjunction $Q \wedge \epsilon$ can be treated as $Q$.
When convenient, we use the notation $\overline{Q}$, or $\overline{\tev \ty C}$ to refer to a list of constraints obtained by flattening the tree.
The two formats are interchangeable.
The same is true for $\QQ$.
The top-level axiom scheme ($\QQ$) contains instance declarations, for example:
\begin{align*}
  &(\textit{\$fEqMaybe} \ty \forall{}a.\,\textit{Eq}~a \Rightarrow \textit{Eq}~(\textit{Maybe}~a))~\wedge \\
  &(\textit{\$fEqInt} \ty \textit{Eq}~\textit{Int})
\end{align*}
Any class constraint C can be considered a degenerate axiom scheme A with $\overline{a}$ and $\overline{C}$ empty.
Similarly, any Q can be considered a degenerate $\QQ$.
Throughout this report, we often assume $\QQ$ is implicitly available.

For convenience, we use the following notations:
\[
\begin{array}{lrlrl}
\forall\,\overline{a}.\,\rho & = & \forall\,a_1\ldots{}a_n.\,\rho & = & \forall\,a_1.\,\ldots\,.\,\forall\,a_n.\,\rho \\
\forall\,\overline{a}.\,\upsilon & = & \forall\,a_1\ldots{}a_n.\,\upsilon & = & \forall\,a_1.\,\ldots\,.\,\forall\,a_n.\,\upsilon \\
\overline{C}        & = & C_1 \Rightarrow \ldots \Rightarrow C_n    & = & (C_1, \ldots, C_n) \\
\end{array}
\]

\section{Elaboration}%
\label{sec:elaboration}

\begin{figure}[tb]
\begin{framed}

\raggedright\boxed{\elab{\cdot}{\cdot}{\cdot}}
\vspace{-1em}
\[
\inference{\elab{\rho}{\rho}{\upsilon}}
{\elab{\sigma}{\rho}{\upsilon}}
\hspace{4em}
\inference{\elab{\sigma}{\sigma}{\upsilon}}
{\elab{\sigma}{\forall{}a.\,\sigma}{\forall{}a.\,\upsilon}}
\]
\vspace{0.75em}
\[
\inference{%
\elab{\tau}{\tau}{\upsilon}
}
{\elab{\rho}{\tau}{\upsilon}}
\hspace{4em}
\inference{%
\elab{C}{C}{\upsilon_1} &
\elab{\rho}{\rho}{\upsilon_2}
}
{\elab{\rho}{C \Rightarrow \rho}{\upsilon_1 \rightarrow \upsilon_2}}
\]
\vspace{0.75em}
\[
\inference{}{\elab{\tau}{a}{a}}
\hspace{3em}
\inference{%
\elab{\tau}{\tau_1}{\upsilon_1} &
\elab{\tau}{\tau_2}{\upsilon_2}}
{\elab{\tau}{\tau_1 \rightarrow \tau_2}{\upsilon_1 \rightarrow \upsilon_2}}
\hspace{3em}
\inference{\elab{\tau}{\tau}{\upsilon}}
{\elab{\tau}{\dict{\TC}~\tau}{\dict{\TC}~\upsilon}}
\]
\vspace{0.75em}
\[
\inference{\elab{\tau}{\tau}{\upsilon}}
{\elab{C}{\TC~\tau}{\dict{\TC}~\upsilon}}
\]
\vspace{0.75em}
\[
\inference{}
{\elab{Q}{\epsilon}{\epsilon}}
\hspace{3em}
\inference{%
\elab{C}{C}{\upsilon}
}
{\elab{Q}{\tev \ty C}{\tev \ty \upsilon}}
\hspace{3em}
\inference{%
\elab{Q}{Q_1}{Q_{\upsilon}^1} &
\elab{Q}{Q_2}{Q_{\upsilon}^2}
}
{\elab{Q}{Q_1 \wedge Q_2}{Q_{\upsilon}^1 \wedge Q_{\upsilon}^2}}
\]
\vspace{0.75em}
\[
\inference{\elab{A}{\overline{C} \Rightarrow C}{\upsilon}}
{\elab{A}{\forall\,\overline{a}.\,\overline{C} \Rightarrow C}{\forall\,\overline{a}.\,\upsilon}}
\hspace{3em}
\inference{\elab{C}{C}{\upsilon}}
{\elab{A}{C}{\upsilon}}
\hspace{3em}
\inference{%
\elab{C}{C}{\upsilon} &
\elab{A}{\overline{C} \Rightarrow C'}{\upsilon'}
}
{\elab{A}{C \Rightarrow \overline{C} \Rightarrow C'}{\upsilon \rightarrow \upsilon'}}
\]
\vspace{0.75em}
\[
\inference{}
{\elab{\QQ}{\epsilon}{\epsilon}}
\hspace{3em}
\inference{%
\elab{A}{A}{\upsilon}
}
{\elab{\QQ}{\tev \ty A}{\tev \ty \upsilon}}
\hspace{3em}
\inference{%
\elab{\QQ}{\QQ_1}{\QQ_{\upsilon}^1} &
\elab{\QQ}{\QQ_2}{\QQ_{\upsilon}^2}
}
{\elab{\QQ}{\QQ_1 \wedge \QQ_2}{\QQ_{\upsilon}^1 \wedge \QQ_{\upsilon}^2}}
\]
\vspace{0.75em}
\[
\inference{}
{\elab{\Gamma}{\epsilon}{\epsilon}}
\hspace{3em}
\inference{%
\elab{\Gamma}{\Gamma}{\Gamma_\upsilon}
}
{\elab{\Gamma}{a, \Gamma}{a, \Gamma_\upsilon}}
\hspace{3em}
\inference{%
\elab{\sigma}{\sigma}{\upsilon} &
\elab{\Gamma}{\Gamma}{\Gamma_\upsilon}
}
{\elab{\Gamma}{(x\ty{}\sigma), \Gamma}{(x\ty{}\upsilon), \Gamma_\upsilon}}
\]
\vspace{0.75em}
\[
\inference{%
\elab{Q}{Q}{Q_\upsilon} &
\elab{\Gamma}{\Gamma}{\Gamma_\upsilon}
}
{\elabqg{Q}{\Gamma}{Q_\upsilon \doubleplus \Gamma_\upsilon}}
\hspace{6em}
\inference{%
\elab{\QQ}{\QQ}{\QQ_\upsilon} &
\elab{\Gamma}{\Gamma}{\Gamma_\upsilon}
}
{\elabqqg{\QQ}{\Gamma}{\QQ_\upsilon \doubleplus \Gamma_\upsilon}}
\]
\end{framed}
\caption{Elaboration rules}%
\label{fig:elaboration-rules}
\end{figure}

Elaboration is the translation from the source language to the target language.
Besides source programs (which are elaborated according to their typing rules, see \cref{sec:typing-rules}), source types, constraints, axiom schemes, contexts, \ldots are also translated to the target language.
These translations are defined using a number of syntax-directed judgments of the form $\elab{\cdot}{\cdot}{\cdot}$ listed below in \cref{fig:elaboration-rules}, where the symbol on the arrow indicates the type of the syntax element.

We use the $_\upsilon$-suffix to denote a target-language variant of a source-level entity, e.g., $Q_\upsilon$ vs. $Q$.
As they are not used outside these elaboration rules, except for $\Gamma_\upsilon$ (see \cref{fig:syntax-target-lang}), we do not explicitly define them.
We use an overbar to indicate that a list is elaborated, e.g., $\elab{C}{\overline{C}}{\overline{\upsilon}}$.

In the target language, constraints and top-level axiom schemes are regular types, and evidence terms are regular target-level terms.
Therefore, we simply translate $Q$ and $\QQ$ to regular value bindings using the judgments $\elabqg{Q}{\Gamma}{\Gamma_\upsilon}$ and $\elabqqg{\QQ}{\Gamma}{\Gamma_\upsilon}$.

\section{Declarative Typing Rules}%
\label{sec:decl-typing-rules}

The typing rules depend on the \emph{constraint entailment} relation~\cite{OutsideIn}, which we will discuss first.

\subsection{Constraint Entailment}%
\label{sec:entailment}

\begin{figure}[tb]
\begin{framed}
\[
\begin{array}{llr}
\text{Reflexivity} & \entails{\QQ \wedge Q}{Q} & \text{(R1)} \\
\text{Transitivity} & \entails{\QQ \wedge Q_1}{Q_2} \text{ and } \entails{\QQ \wedge Q_2}{Q_3} \text{ implies } \entails{\QQ \wedge Q_1}{Q_3} & \text{(R2)} \\
\text{Substitution} & \entails{\QQ}{Q_2} \text{ implies } \entails{\theta\QQ}{\theta{}Q_2} \text{ where } \theta \text{ is a type substitution} & \text{(R3)} \\
\ldots \\
\text{Conjunctions} & \entails{\QQ}{Q_1} \text{ and } \entails{\QQ}{Q_2} \text{ implies } \entails{\QQ}{Q_1 \wedge Q_2} & \text{(R7)}
\end{array}
\]
\end{framed}
\caption{Entailment requirements}\label{fig:entailment-requirements}
\end{figure}

\[
\boxed{\entails{\QQ}{Q}}
\]
This relation can be read as: ``from the top-level axiom scheme $\QQ$, we can derive the constraints $Q$.''
Following \textsc{OutsideIn(X)}, we leave the details of entailment deliberately unspecified, because it is a parameter of the type system~\cite{OutsideIn}.
\cref{fig:entailment-requirements} lists the requirements of the entailment relation~\cite{OutsideIn}.
Rules about type equalities are omitted, hence the jump from R3 to R7.

Compared to \textsc{OutsideIn(X)}, we extend this relation to produce evidence for each entailment, i.e.\ the $\tev$ in $\tev \ty C$, which we will use to elaborate typing judgements.

Remember that any $Q$ is a degenerate $\QQ$, so $\entails{Q_1}{Q_2}$ is also valid.

When $\entails{\QQ}{Q}$, then for all $(\tev \ty C) \in Q$, $\sftc{\Gamma_\upsilon}{\tev}{\upsilon_C}$ for any $\Gamma$, where $\elabqqg{\QQ}{\Gamma}{\Gamma_\upsilon}$ and $\elab{C}{C}{\upsilon_C}$.
In other words, each evidence term $\tev$ will have the type of the elaborated constraint $C$ in a context obtained by elaborating $\QQ$ in combination with any $\Gamma$.

Sometimes we write $\sftc{\Gamma_\upsilon}{\tev}{C}$ instead of $\upsilon_C$ for clarity.

Another convenient notation we use is $\sftc{\Gamma_\upsilon}{\overline{\tev}}{Q}$ where the $\overline{\tev}$ correspond to the evidence terms in $Q$, which will be of the form $\tev[1] \ty C_1 \wedge \tev[2] \ty C_2 \wedge \ldots$.

\subsection{Typing Rules}%
\label{sec:typing-rules}

\begin{figure}[tb]
\begin{framed}
\small
\raggedright\boxed{\decelab{Q}{\Gamma}{e}{\sigma}{t}}

\[
\inference{%
(x\ty{}\sigma) \in \Gamma
}
{\decelab{Q}{\Gamma}{x}{\sigma}{x}}[\textsc{Var}]
\hspace{3em}
\inference{%
\decelab{Q}{(x\ty{}\tau_1), \Gamma}{e}{\tau_2}{t} &
\new{\elab{\tau}{\tau_1}{\upsilon_1}}
}
{\decelab{Q}{\Gamma}{\lambda{}x.\,e}{\tau_1 \rightarrow \tau_2}{\lambda{}(x\ty{}\upsilon_1).\,t}}[\textsc{Abs}]
\]
\vspace{0.5em}
\[
\inference{%
\decelab{Q}{\Gamma}{e_1}{\tau_1 \rightarrow \tau_2}{t_1} &
\decelab{Q}{\Gamma}{e_2}{\tau_1}{t_2}
}
{\decelab{Q}{\Gamma}{e_1~e_2}{\tau_2}{t_1~t_2}}[\textsc{App}]
\]
\vspace{0.5em}
\[
\inference{%
\decelab{Q}{a, \Gamma}{e}{\sigma}{t} &
a \notin \ftv{Q, \Gamma}
}
{\decelab{Q}{\Gamma}{e}{\forall{}a.\,\sigma}{\Lambda{}a.\,t}}[$\introforall$]
\hspace{2em}
\inference{%
\decelab{Q}{\Gamma}{e}{\forall{}a.\,\sigma}{t} &
\new{\elab{\tau}{\tau}{\upsilon}}
}
{\decelab{Q}{\Gamma}{e}{[a \mapsto \tau]\sigma}{t~\upsilon}}[$\elimforall$]
\]
\vspace{0.5em}
\[
\inference{%
\decelab{\new{d \ty} C \wedge Q}{\Gamma}{e}{\rho}{t} &
\new{\elab{C}{C}{\upsilon}} &
d \notin \fdv{Q}
}
{\decelab{Q}{\Gamma}{e}{C \Rightarrow \rho}{\lambda{}(d\ty{}\upsilon).\,t}}[$\introct$]
\]
\vspace{1em}
\[
\inference{%
\decelab{Q}{\Gamma}{e}{C \Rightarrow \rho}{t} &
\entails{\QQ \wedge Q}{\new{\tev \ty} C}
}
{\decelab{Q}{\Gamma}{e}{\rho}{t~\tev}}[$\elimct$]
\]
\vspace{0.5em}
\[
\inference{%
\decelabs{Q}{\Gamma}{e_1}{\forall\overline{b_1}a\overline{b_2}.\,(\overline{C_1}, \TC~a, \overline{C_2}) \Rightarrow \tau_1}{t_1} &
\decelab{Q}{\Gamma}{e_2}{\dict{\TC}~\tau_2}{t_2} \\
\new{\elab{\tau}{\tau_2}{\upsilon_2}} &
\new{\elab{C}{[a \mapsto \tau_2]\overline{C_1}}{\overline{\upsilon_1}}} \\
\forall C \ldotp \entails{\QQ \wedge \TC~a}{\tev \ty C} \Rightarrow ((\nentails{\QQ \wedge Q \wedge \overline{C_1} \wedge \overline{C_2}}{C}) \vee (\entails{\QQ}{\tev \ty C}))
}
{\decelab{Q}{\Gamma}{e_1~\dictapp{e_2}{\TC~a}}{\forall{}\overline{b_1}\overline{b_2}.\,[a \mapsto \tau_2]((\overline{C_1}, \overline{C_2}) \Rightarrow \tau_1)}{\Lambda{}\overline{b_1}\overline{b_2}.\,\lambda{}(\overline{d \ty \upsilon_1}).\,t_1~\overline{b_1}~\upsilon_2~\overline{b_2}~\overline{d}~t_2}}[\textsc{DictApp}]
\]
\begin{minipage}{0.5\textwidth}
\vspace{2em}
\boxed{\decelabs{Q}{\Gamma}{e}{\sigma}{t}}
\[
\inference{%
\decelab{Q}{\Gamma}{e}{\sigma}{t} \\
\text{The type of } e \text{ is specified to be } \sigma \\
\text{The principal type of } e \text{ is unambiguous}
}
{\decelabs{Q}{\Gamma}{e}{\sigma}{t}}[\textsc{Spec}]
\]
\end{minipage}%
\hspace{0.8em}
\begin{minipage}{0.45\textwidth}
\vspace{2em}
\boxed{\dectopelab{\QQ}{\Gamma}{e}{\sigma}{t}}
\vspace{0.4em}
\[
\inference{%
\decelab{Q}{\Gamma}{e}{\tau}{t} &
\entailselab{\QQ \wedge \overline{\new{d \ty} C}}{Q}{\eta} \\
\overline{a} = \ftv{\overline{C}, \tau} &
\elab{C}{\overline{C}}{\overline{\upsilon}}
}
{\dectopelab{\QQ}{\Gamma}{e}{\forall{}\overline{a}.\,\overline{C} \Rightarrow \tau}{\Lambda{}\overline{a}.\,\lambda{}(\overline{d \ty \upsilon}).\,\eta{}(t)}}[\textsc{Top}]
\]
\end{minipage}
\end{framed}
\caption{Declarative typing rules of the source language}%
\label{fig:typing-rules-source-lang}
\end{figure}

The typing judgment can be read as: ``under assumptions $Q$ and context $\Gamma$, expression $e$ has type $\sigma$ and its elaboration is term $t$''.
Except for the new \textsc{DictApp} rule, the rules are rather standard and based on~\cite{AssociatedTypeSynonyms}.
One of the main changes is that these rules derive a polytype $\sigma$ instead of a monotype $\tau$.
The reason for this is that a dictionary can only be applied to an expression with a $\sigma$ type, i.e.\ an expression that has the type class constraint corresponding to the dictionary in its type.
As a result, the polytype $\sigma$ obtained from the context in \textsc{Var} is no longer immediately instantiated, but is instantiated using the new $\elimforall$ and $\elimct$ rules.

Let us now discuss the new \textsc{DictApp} rule in more detail.
For simplicity, the annotation ``\ensuremath{\textbf{as}\;\Conid{TC}\;\Varid{a}}'' is mandatory in the formalisation, whereas it is optional in the implementation when there is only one constraint in the context matching the type of the dictionary.

The type of $e_1$, to which the dictionary will be passed, must be specified.
This can be done simply by writing a type signature, or by annotating the expression itself with its type.
This is expressed using the $\decelabs{Q}{\Gamma}{e}{\sigma}{t}$ judgment.
This judgment also states that the principal type of $e$ must be unambiguous, a requirement that is needed for coherence (even without explicit dictionary application), see \cref{sec:coherence}.
The type class constraint ($\TC~a$) to which a dictionary is passed may occur at any place in the type class context of $e_1$.
The same is true for the corresponding type variable ($a$).
This is captured by the zero or more constraints $\overline{C_1}$ and $\overline{C_2}$ coming before and after the type class constraint in question, and the zero or more type variables $\overline{b_1}$ and $\overline{b_2}$ coming before and after the type variable in question.
The dictionary $e_2$ must have a type ($\dict{\TC}~\tau_2$), matching the type class of the constraint.
After passing the dictionary, the type variable $a$ is instantiated with $\tau_2$ and the constraint in question is removed from the type class context.
The elaboration in \textsc{DictApp} is more elaborate, as type variables and constraints must be rearranged.
Consequently, type and evidence abstractions and applications must be added to align the resulting type with the resulting term, i.e.\ the type $\upsilon_2$ must be applied before $\overline{b_2}$ are applied, similarly, $t_2$ must be applied after the $\overline{d}$ corresponding with $\overline{C_1}$.
We have $\eta$-reduced the evidence abstractions and applications corresponding with $\overline{C_2}$.
The crux of the rule is that $t_2$, the term corresponding to the dictionary, is passed as an evidence argument to $t_1$.

For simplicity, multiple explicit dictionary applications cannot be chained one after the other in the formalisation, but this is supported in the implementation.

\paragraph{Top-level typing}
The judgment for top-level expressions can be read as ``under top-level axiom scheme $\QQ$ and in context $\Gamma$, term $e$ has type $\sigma$, and its elaboration is term $t$.''
Compared to the regular typing judgment, we make sure no free type variables occur by quantifying over them.
Also, the top-level axiom scheme ($\QQ$) is used to simplify the required constraints.
For example, if $e$ assumes \ensuremath{\Conid{Eq}\;(\Conid{Maybe}\;\Varid{a})} and $\QQ$ contains the axiom \ensuremath{\forall \Varid{a}\hsforall .\;\Conid{Eq}\;\Varid{a}\Rightarrow \Conid{Eq}\;(\Conid{Maybe}\;\Varid{a})}, we want to qualify over the ``extra information'' needed to satisfy the assumption, i.e.\ \ensuremath{\Conid{Eq}\;\Varid{a}}, as:
\begin{align*}
  &(\new{\textit{\$fEqMaybe} \ty} \forall{}a.\,\textit{Eq}~a \Rightarrow \textit{Eq}~(\textit{Maybe}~a))~\wedge (\new{d \ty} \textit{Eq}~a) \\
  &\Vdash \new{d' \ty} \textit{Eq}~(\textit{Maybe}~a) \new{\,; [d' \mapsto \textit{\$fEqMaybe}~a~d]}
\end{align*}

We require a monotype $\tau$ to be derived for $e$ even though the judgment can derive a polytype $\sigma$.
Using the rules $\elimforall$ and $\elimct$, every polytype can be instantiated to a monotype.

Note the ``; $\eta$'' in the entailment, where $\eta$ is a dictionary evidence substitution.
Let us explain this with an example, say we have that:
\begin{align*}
  &(d' \ty \textit{Eq}~(\textit{Maybe}~a))\,;\,((x \ty \textit{Maybe}~a), a) \vdash (==)~x~x \ty \textit{Bool} \\
  &\rightsquigarrow (==)~(\textit{Maybe}~a)~d'~x~x
\end{align*}
When simplifying $Q$, in this case $(d' \ty \textit{Eq}~(\textit{Maybe}~a))$, as demonstrated above, we get $\overline{d \ty C} = d \ty \textit{Eq}~a$.
However, the elaborated term $t$ still contains the dictionary variable $d'$.
Therefore, we must substitute the original dictionary variable $d'$ with the simplified evidence $\textit{\$fEqMaybe}~a~d$.

\paragraph{Algorithmic typing rules}
We do not provide algorithmic typing rules as they are standard~\cite{AssociatedTypeSynonyms,OutsideIn} and the algorithmic variant of \textsc{DictApp} can easily be derived from the declarative one.

\section{Coherence}%
\label{sec:coherence}

We use the definition from \cite{TypeClassesExploration} for coherence: \emph{every different valid typing derivation for a program leads to a resulting program that has the same dynamic semantics}.

So how does this translate to our setting?
Consider the following program:
\begin{hscode}\SaveRestoreHook
\column{B}{@{}>{\hspre}l<{\hspost}@{}}%
\column{E}{@{}>{\hspre}l<{\hspost}@{}}%
\>[B]{}\Varid{foo}\mathbin{::}\Conid{Eq}\;\Conid{Int}\Rightarrow \Conid{Bool}{}\<[E]%
\\
\>[B]{}\Varid{foo}\mathrel{=}\mathrm{1}\mathop{==}\mathrm{3}{}\<[E]%
\ColumnHook
\end{hscode}\resethooks
There are two valid typing derivations for this program: one that uses the global instance of \ensuremath{\Conid{Eq}\;\Conid{Int}} and one that uses the instance passed to \ensuremath{\Varid{foo}} whenever it is used.
It does not matter which instance or dictionary was chosen, because with global uniqueness of instances, it is the same instance in both cases.
With the ability to explicitly pass a dictionary, it suddenly does matter which typing derivation is used, because the different typing derivations can use potentially different dictionaries, which will directly influence the dynamic semantics.
The following cases have a similar risk of incoherence, as the compiler can choose between multiple dictionaries:

\begin{hscode}\SaveRestoreHook
\column{B}{@{}>{\hspre}l<{\hspost}@{}}%
\column{E}{@{}>{\hspre}l<{\hspost}@{}}%
\>[B]{}\mbox{\onelinecomment  Two \ensuremath{\Conid{Eq}\;\Varid{a}} dictionaries...}{}\<[E]%
\\
\>[B]{}\Varid{two}\mathbin{::}(\Conid{Eq}\;\Varid{a},\Conid{Eq}\;\Varid{a})\Rightarrow \Varid{a}\to \Varid{a}\to \Conid{Bool}{}\<[E]%
\\
\>[B]{}\mbox{\onelinecomment  A second \ensuremath{\Conid{Eq}\;\Varid{a}} hidden in \ensuremath{\Conid{Ord}\;\Varid{a}}...}{}\<[E]%
\\
\>[B]{}\Varid{three}\mathbin{::}(\Conid{Eq}\;\Varid{a},\Conid{Ord}\;\Varid{a})\Rightarrow \Varid{a}\to \Varid{a}\to \Conid{Bool}{}\<[E]%
\\
\>[B]{}\mbox{\onelinecomment  A second \ensuremath{\Conid{Eq}\;\Varid{a}} dictionary thanks to constraint \ensuremath{\Varid{a}\sim\Varid{b}}}{}\<[E]%
\\
\>[B]{}\Varid{four}\mathbin{::}(\Conid{Eq}\;\Varid{a},\Varid{a}\sim\Varid{b},\Conid{Eq}\;\Varid{b})\Rightarrow \Varid{a}\to \Varid{b}\to \Conid{Bool}{}\<[E]%
\ColumnHook
\end{hscode}\resethooks
As for \ensuremath{\Varid{foo}} above, global instances can create potential incoherence, also with type variables:
\begin{hscode}\SaveRestoreHook
\column{B}{@{}>{\hspre}l<{\hspost}@{}}%
\column{E}{@{}>{\hspre}l<{\hspost}@{}}%
\>[B]{}\mathbf{instance}\;\Conid{Eq}\;\Varid{a}\;\mathbf{where}\;\anonymous \mathop{==}\anonymous \mathrel{=}\Conid{False}{}\<[E]%
\\[\blanklineskip]%
\>[B]{}\Varid{five}\mathbin{::}\Conid{Eq}\;\Varid{a}\Rightarrow \Varid{a}\to \Varid{a}\to \Conid{Bool}{}\<[E]%
\ColumnHook
\end{hscode}\resethooks

All of the contrived programs above are valid and coherent in Haskell as long as global uniqueness of instances holds.
However, when we add explicit dictionary application, they become incoherent: in each case, the compiler has more than one choice for which dictionary the functions will use, so we cannot allow one of those dictionaries to be instantiated to something else than the others.
To preserve coherence, we restrict explicit type application using another safety criterion.
In this section, we explain this criterion and how we have proven that it effectively salvages coherence.

So how do we detect cases like the above?
Say we pass a dictionary for class constraint $\TC~a$ to a function with the following type:
\[
\forall\overline{b_1}a\overline{b_2}.\,(\overline{C_1}, \TC~a, \overline{C_2}) \Rightarrow \tau_1
\]
What all the examples had in common was that the type class constraint $\TC~a$ or any constraint that could be derived from it in combination with the top-level axiom scheme $\QQ$, could also be derived from the remaining constraints and the top-level axiom scheme ($\QQ \wedge \overline{C_1} \wedge \overline{C_2}$).
As this is also trivially true for any constraint in $\QQ$, we also require that the constraint cannot be derived from $\QQ$ while producing the same evidence.
Since we allow explicit type application to non-top-level expressions, we have to add to this list the local assumptions ($Q$).
Generally, an explicit type application can cause incoherence if:
\begin{align*}
  \exists C \ldotp &(\entails{\QQ \wedge \TC~a}{\tev \ty C}) \wedge (\entails{\QQ \wedge Q \wedge \overline{C_1} \wedge \overline{C_2}}{C})\, \wedge \\
             &\nentails{\QQ}{\tev \ty C}
\end{align*}
To forbid such cases and to safeguard coherence, we add the following condition to the \textsc{DictApp} rule:
\begin{align*}
  \forall C \ldotp &\entails{\QQ \wedge \TC~a}{\tev \ty C} \Rightarrow \\
             &((\nentails{\QQ \wedge Q \wedge \overline{C_1} \wedge \overline{C_2}}{C}) \vee \entails{\QQ}{\tev \ty C})
\end{align*}
In other words: to ensure coherence, the type class constraint that we provide an explicit dictionary for, or any constraint implied by it and the global instances, is either not implied by the remaining constraints and the global instances, or it is implied by the global instances while yielding the exact same evidence.
All of the examples above are caught by this check.

\section{Coherence Proof}%
\label{sec:coherence-proof}

To ensure this check is sufficient to prevent incoherence, we prove coherence of our formalisation.

First, a short sketch of the proof: to prove coherence, we must to prove that each typing derivation of a program $e$ results in the same dynamic semantics.
In addition to a type, a typing derivation also translates $e$ to a System F term $t$ in which all implicit information, type abstraction and application, and evidence/dictionary abstraction and application, is made explicit.
The System F term $t$ will determine the dynamic semantics of $e$.
The different typing derivations potentially differ in the selected dictionaries, which will have a direct impact on the dynamic semantics.
Our objective is to prove that the terms produced by any two different typing derivations of the same program $e$ are equivalent.

Before we can show the actual coherence theorem and proof, we first have to define a number of auxiliary concepts and lemmas.

\subsection{Equivalence Relation}%
\label{sec:equivalence-relation}

To prove coherence, we must prove that every different valid typing derivation for a program leads to a resulting program that has the same dynamic semantics.
The dynamic semantics of a program are fully determined by the elaborated program in the target language.
To be able to determine whether two target-language programs have the \emph{same} dynamic semantics, we use an axiomatic equivalence relation between two target-language terms, based on \cite{QualifiedTypes}, shown in \cref{fig:tred-relation}.
Note that this equivalence relation considers the behaviour of terms after \emph{type erasure}, i.e.\ the type annotations of binders are removed, as well as type applications and abstractions.
\begin{figure}[tb]
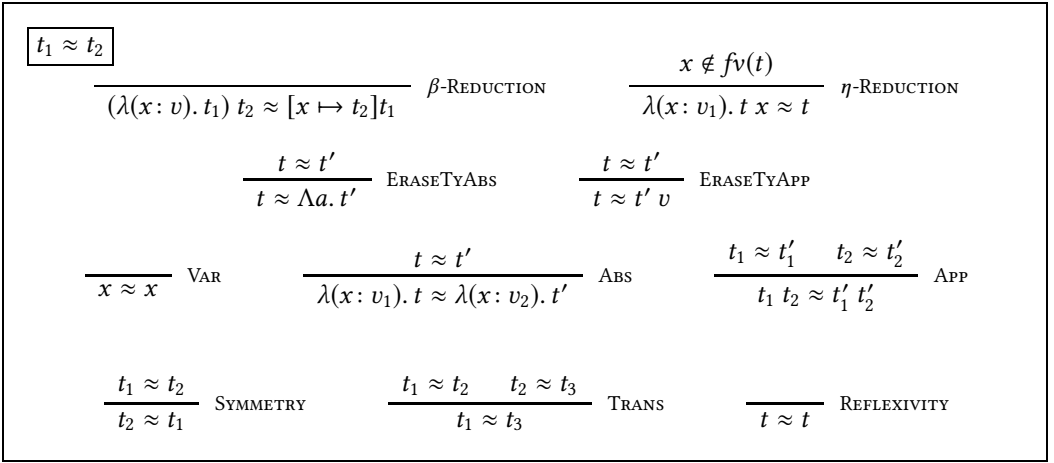

  \begin{framed}
\raggedright\boxed{\tred{t_1}{t_2}}
\vspace{-1em}
\[
\inference{}{\tred{(\lambda{}(x \ty \upsilon).\,t_1)~t_2}{[x \mapsto t_2]t_1}}[\betared]
\hspace{3em}
\inference{x \notin \fv{t}}{\tred{\lambda{}(x \ty \upsilon_1).\,t~x}{t}}[\etared]
\]
\vspace{0.5em}
\[
\inference{\tred{t}{t'}}{\tred{t}{\Lambda{}a.\,t'}}[\textsc{EraseTyAbs}]
\hspace{3em}
\inference{\tred{t}{t'}}{\tred{t}{t'~\upsilon}}[\textsc{EraseTyApp}]
\]
\vspace{0.5em}
\[
\inference{}{\tred{x}{x}}[\textsc{Var}]
\hspace{3em}
\inference{\tred{t}{t'}}{\tred{\lambda{}(x \ty \upsilon_1).\,t}{\lambda{}(x \ty \upsilon_2).\,t'}}[\textsc{Abs}]
\hspace{3em}
\inference{\tred{t_1}{t_1'} & \tred{t_2}{t_2'}}{\tred{t_1~t_2}{t_1'~t_2'}}[\textsc{App}]
\]
\vspace{2em}
\[
\inference{\tred{t_1}{t_2}}{\tred{t_2}{t_1}}[\textsc{Symmetry}]
\hspace{3em}
\inference{\tred{t_1}{t_2} & \tred{t_2}{t_3}}{\tred{t_1}{t_3}}[\textsc{Trans}]
\hspace{3em}
\inference{}{\tred{t}{t}}[\textsc{Reflexivity}]
\]
\end{framed}
\caption{Equivalence relation for terms}%
\label{fig:tred-relation}
\end{figure}

As our target language is lazy, we can safely use the \betared{} rule from \cref{fig:tred-relation}, which is only sound for languages with non-strict semantics.
By including the rules \textsc{Symmetry}, \textsc{Trans}, and \textsc{Reflexivity}, we extend this relation to an equivalence relation.

\begin{lemma}\label{lemma:erase-theta}
  For every $\tev$ and substitution $\theta$, $\tred{\tev}{\theta{}(\tev)}$.
\end{lemma}

As the equivalence is based on type-erased terms, applying a substitution to a
term preserves its equivalence.

\begin{proof}
  Straightforward induction on $\tev$ with the \textsc{EraseTyApp} rule.
\end{proof}

\subsection{Principal Types}%
\label{sec:principal-types}

As established in previous work~\cite[Theorem 5.32]{QualifiedTypes}, coherence in the presence of type classes requires that the principal type (defined below) of program $e$ must be (context-)unambiguous (defined in \cref{sec:ambiguity}).

The typical example used to explain this condition, is the following:
\begin{hscode}\SaveRestoreHook
\column{B}{@{}>{\hspre}l<{\hspost}@{}}%
\column{E}{@{}>{\hspre}l<{\hspost}@{}}%
\>[B]{}\mathbf{class}\;\Conid{Show}\;\Varid{a}\;\mathbf{where}\;\Varid{show}\mathbin{::}\Varid{a}\to \Conid{String}{}\<[E]%
\\
\>[B]{}\mathbf{class}\;\Conid{Read}\;\Varid{a}\;\mathbf{where}\;\Varid{read}\mathbin{::}\Conid{String}\to \Varid{a}{}\<[E]%
\\[\blanklineskip]%
\>[B]{}\Varid{f}\mathbin{::}\Conid{String}\to \Conid{String}{}\<[E]%
\\
\>[B]{}\Varid{f}\;\Varid{s}\mathrel{=}\Varid{show}\;(\Varid{read}\;\Varid{s}){}\<[E]%
\ColumnHook
\end{hscode}\resethooks
The function \ensuremath{\Varid{f}} converts a \ensuremath{\Conid{String}} to a value of some type, let us call it \ensuremath{\Varid{x}}, using \ensuremath{\Varid{read}}, and converts it back to a \ensuremath{\Conid{String}} using \ensuremath{\Varid{show}}.
Many different typing derivations are possible for this program, as any type may be chosen for \ensuremath{\Varid{x}}.
As the type \ensuremath{\Varid{x}} determines the dictionary that will be chosen, the dynamic semantics will vary along with the typing derivations, e.g., a program that reads and shows an \ensuremath{\Conid{Int}} obviously has different dynamic semantics than one reading and showing a \ensuremath{\Conid{Maybe}\;[\mskip1.5mu \Conid{String}\mskip1.5mu]}.

We exclude such programs from coherence, as their principal type, in this case \ensuremath{\forall \Varid{a}\hsforall .\;(\Conid{Show}\;\Varid{a},\Conid{Read}\;\Varid{a})} \ensuremath{\Rightarrow \Conid{String}\to \Conid{String}}, is \emph{ambiguous}, i.e.\ a type variable from the constraints does not occur in the remainder of the type, see \cref{sec:ambiguity}.

To define principal types, we first define an ordering relation on quantified type schemes, based on \cite[Proposition 3.4]{QualifiedTypes}.

\begin{definition}[More general quantified type scheme]
  A type scheme
  $\sigma_1 = \forall\overline{a}.\,\overline{C_1} \Rightarrow \tau_1$ with
  assumptions $Q_1$ is a more general quantified type scheme than
  $\sigma_2 = \forall\overline{b}.\,\overline{C_2} \Rightarrow \tau_2$ with
  assumptions $Q_2$, or $\mg{Q_1}{\sigma_1}{Q_2}{\sigma_2}$, iff:
  \begin{align*}
    &\exists \theta = [\overline{a \mapsto \tau}], \\
    &\quad\overline{b} \notin \ftv{\sigma_1, Q_1, Q_2} \text{ and}\\
    &\quad\theta{}(\tau_1) = \tau_2 \text{ and}\\
    &\quad\entails{Q_2 \wedge \overline{C_2}}{Q_1 \wedge \theta{}(\overline{C_1})}
  \end{align*}
\end{definition}

\begin{theorem}[Principal types]\label{thm:principal-types}
  If $e$ is well-typed, then there exists a principal type $\sigma$ where
  $\dec{\epsilon}{\Gamma}{e}{\sigma}$, such that for all $\sigma'$ where
  $\dec{Q'}{\Gamma}{e}{\sigma'}$, it is the case that
  $\mg{\epsilon}{\sigma}{Q'}{\sigma'}$.
\end{theorem}

As mentioned in \cref{sec:typing-rules}, we do not give a type inference algorithm.
However, as our typing rules are based on \cite{AssociatedTypeSynonyms, OutsideIn}, for which algorithms that infer principal types exist, we assume an algorithm that infers principal types exists.

\begin{lemma}\label{lemma:var-type-is-principal}
  If:
  \begin{align}
    &(x \ty \sigma) \in \Gamma \notag \\
    &\dec{Q}{\Gamma}{x}{\sigma'} \label{eq:given-var-deriv}
  \end{align}
  Then $\mg{\epsilon}{\sigma}{Q}{\sigma'}$, i.e.\ $\sigma$ is the principal
  type of $x$.
\end{lemma}

\begin{proof}
  Straightforward induction on \cref{eq:given-var-deriv}.
\end{proof}

\subsection{Ambiguity and Context-Ambiguity}%
\label{sec:ambiguity}

We define (context-)ambiguity and (context)-unambiguity.

\begin{definition}[Ambiguous type scheme]
  A type scheme $\sigma = \forall\overline{a}.\,\overline{C} \Rightarrow \tau$
  is ambiguous iff:
  \begin{align*}
    \exists a, a \in \overline{a} \cap \ftv{\overline{C}} \text{ and } a \notin \ftv{\tau}
  \end{align*}
\end{definition}

\begin{definition}[Unambiguous type scheme]
  A type scheme $\sigma = \forall\overline{a}.\,\overline{C} \Rightarrow \tau$
  is unambiguous iff:
  \begin{align*}
    \overline{a} \cap \ftv{\overline{C}} \subseteq \ftv{\tau}
  \end{align*}
\end{definition}

When applying the \textsc{Abs} rule, an unambiguous type may suddenly become ambiguous, as the disambiguating type variable may have moved from the monotype $\tau$ to the context $\Gamma$.
Therefore, we use the following definitions throughout the coherence proof to handle this as well:

\begin{definition}[Context-ambiguous type scheme]
  A type scheme $\sigma = \forall\overline{a}.\,\overline{C} \Rightarrow \tau$
  is context-ambiguous in context $\Gamma$ iff:
  \begin{align*}
    \exists a, a \in \overline{a} \cap \ftv{\overline{C}} \text{ and } a \notin \ftv{\Gamma, \tau}
  \end{align*}
\end{definition}

\begin{definition}[Context-unambiguous type scheme]
  A type scheme $\sigma = \forall\overline{a}.\,\overline{C} \Rightarrow \tau$
  is context-unambiguous in context $\Gamma$ iff:
  \begin{align*}
    \overline{a} \cap \ftv{\overline{C}} \subseteq \ftv{\Gamma, \tau}
  \end{align*}
\end{definition}

Whenever the context is closed, i.e.\ $\ftv{\Gamma} = \emptyset$, context-(un)ambiguity coincides with (un)ambiguity.
For this reason, we sometimes omit the \emph{context-} prefix in examples where the context is closed.

\begin{lemma}\label{lemma:ambi-app1}
  If the principal type
  $\sigma = \forall\overline{a}.\,\overline{C} \Rightarrow \tau$ of $e_1~e_2$
  is context-unambiguous where $\dec{\epsilon}{\Gamma}{e_1~e_2}{\sigma}$, then
  the principal type
  $\sigma_1 = \forall\overline{b}.\,\overline{C_2} \Rightarrow \tau_1 \rightarrow \tau_2$ of $e_1$
  is also context-unambiguous where $\dec{\epsilon}{\Gamma}{e_1}{\sigma_1}$.
\end{lemma}

\begin{proof}
  Similar to \cref{lemma:ambi-app2}.
\end{proof}

\begin{lemma}\label{lemma:ambi-app2}
  If the principal type
  $\sigma_1 = \forall\overline{a_1}.\,\overline{C_1} \Rightarrow \tau_1$ of
  $e_1~e_2$ is context-unambiguous where
  $\dec{\epsilon}{\Gamma}{e_1~e_2}{\sigma_1}$, then the principal type
  $\sigma_2 = \forall\overline{a_2}.\,\overline{C_2} \Rightarrow \tau_2$ of
  $e_2$ is also context-unambiguous where $\dec{\epsilon}{\Gamma}{e_2}{\sigma_2}$.
\end{lemma}

\begin{proof}
  Proof by contraposition. Assume that the principal type $\sigma_2$ of $e_2$
  is context-ambiguous. We will show that the principal type of $e_1~e_2$ is
  also context-ambiguous.

  Since the principal type
  $\sigma_2 = \forall\overline{a_2}.\,\overline{C_2} \Rightarrow \tau_2$ is
  context-ambiguous, there exists a type variable
  $a \in \overline{a_2} \cap \ftv{\overline{C_2}}$ for which
  $a \notin \ftv{\Gamma, \tau_2}$. Say $C_a$ is a constraint
  $\in \overline{C_2}$ containing $a$. Without loss of generality, we assume
  that $\overline{a_2} \mathrel{\#} \overline{a_1}$.

  Take a typing derivation $\dec{Q}{\Gamma}{e_1~e_2}{\sigma}$ with
  $\overline{a_2} \notin \Gamma$ with
  $\sigma = \forall\overline{a}.\,\overline{C} \Rightarrow \tau$. We prove by
  induction that
  $\mg{Q}{(\forall\overline{a_2}.\, C_a \Rightarrow \sigma)}{Q}{\sigma}$. In
  particular, this means that
  $\mg{\epsilon}{(\forall\overline{a_2}\,\overline{a}.\,C_a \Rightarrow
    \overline{C_1} \Rightarrow \tau_1)}{\epsilon}{\sigma_1}$ for the principal
  type $\sigma_1$ of $e_1~e_2$, so that
  $(\forall\overline{a_2}\,\overline{a_1}.\,C_a \Rightarrow \overline{C_1}
  \Rightarrow \tau_1)$ is also a principal type for $e_1~e_2$ in the empty
  context. Since $a \in \ftv{C_a}$ but $a \notin \ftv{\tau_1}$, this new
  principal type is ambiguous.

  By induction on $\dec{Q}{\Gamma}{e_1~e_2}{\sigma}$, we get the following
  five cases (see \cref{fig:typing-rules-source-lang}):
  \begin{itemize}
  \item \textsc{App}: We have that $\sigma = \tau_4$,
    $\dec{Q}{\Gamma}{e_1}{\tau_3 \rightarrow \tau_4}$ and
    $\dec{Q}{\Gamma}{e_2}{\tau_3}$.

    Because $\epsilon;\sigma_2$ is the principal type of $e_2$, we have that
    $\mg{\epsilon}{\sigma_2}{Q}{\tau_3}$.

    By definition, this means that $\tau_3 = \theta(\tau_2)$ and
    $\entails{Q}{\theta{}(\overline{C_2})}$ for some
    $\theta = [\overline{a_2 \mapsto \tau}]$. Particularly, we have that
    $\entails{Q}{\theta{}(C_a)}$.

    It follows that
    $\mg{Q}{(\forall\overline{a_2}.\,C_a \Rightarrow \tau_4)}{Q}{\tau_4}$,
    since we have that $\theta{}(\tau_4) = \tau_4$ and
    $\entails{Q}{\theta{}(C_a)}$ (see above).

  \item $\introforall$: We have that $\sigma = \forall{}a'.\,\sigma'$,
    $\dec{Q}{a', \Gamma}{e_1~e_2}{\sigma'}$ and $a' \notin \ftv{Q,\Gamma}$.

    By the induction hypothesis, we have that
    $\mg{Q}{\forall\overline{a_2}\,C_a \Rightarrow \sigma'}{Q}{\sigma'}$.

    It follows that $\mg{Q}{\forall{}a.\,C_a \Rightarrow \sigma}{Q}{\sigma}$,
    by extending the $\theta$ with $[a' \mapsto a']$.
  \item $\elimforall$: We have that $\sigma = [a' \mapsto \tau]\sigma'$ and
    $\dec{Q}{\Gamma}{e_1~e_2}{\forall{}a'.\,\sigma'}$.

    By the induction hypothesis, we have that
    $\mg{Q}{\forall\overline{a_2}.\,C_a \Rightarrow
      \forall{}a'.\,\sigma'}{Q}{\forall{}a'.\,\sigma'}$.

    It follows that
    $\mg{Q}{\forall{}a.\,C_a \Rightarrow [a' \mapsto \tau]\sigma'}{Q}{[a'
      \mapsto \tau]\sigma'}$ because we can construct a
    $\theta' = [a' \mapsto \tau, \theta]$.
  \item $\introct$: We have that $\sigma = (C \Rightarrow \rho)$ and
    $\dec{C \wedge Q}{\Gamma}{e_1~e_2}{\rho}$.

    By the induction hypothesis, we have that
    $\mg{C \wedge Q}{(\forall\overline{a_2}.\,C_a \Rightarrow \rho)}{C \wedge
      Q}{\rho}$.

    It follows that
    $\mg{Q}{(\forall{}a.\, C_a \Rightarrow C \Rightarrow \rho)}{Q}{(C
      \Rightarrow \rho)}$ because if
    $\entails{C \wedge Q \wedge C_\rho}{C \wedge Q \wedge \theta{}(C_a \wedge
      C_\rho)}$, then also
    $\entails{Q \wedge C \wedge C_\rho}{Q \wedge \theta{}(C_a \wedge C \wedge
      C_\rho)}$.

  \item $\elimct$: We have that $\sigma = \rho$ and
    $\dec{Q}{\Gamma}{e_1~e_2}{C \Rightarrow \rho}$ and $\entails{Q}{C}$.

    By the induction hypothesis, we have that
    $\mg{Q}{(\forall\overline{a_2}.\,C_a \Rightarrow C \Rightarrow
      \rho)}{Q}{(C \Rightarrow \rho)}$.

    It follows that
    $\mg{Q}{\forall\overline{a_2}.\,C_a \Rightarrow \rho}{Q}{\rho}$, since
    $\entails{Q \wedge C_\rho}{Q \wedge \theta{}(C_a \wedge C_\rho)}$ follows
    from
    $\entails{Q \wedge C \wedge C_\rho}{Q \wedge \theta{}(C_a \wedge C_\rho)}$
    since we know that $\entails{Q}{C}$.
    \end{itemize}
\end{proof}

\begin{lemma}\label{lemma:ambi-dictapp2}
  If the principal type
  $\sigma = \forall\overline{a}.\,\overline{C} \Rightarrow \tau$ of
  $e_1~\dictapp{e_2}{\TC~a}$ is context-unambiguous where
  $\dec{\epsilon}{\Gamma}{e_1~\dictapp{e_2}{\TC~a}}{\sigma}$, then the
  principal type
  $\sigma' = \forall\overline{b}.\,\overline{C'} \Rightarrow \dict{\TC}~\tau'$
  of $e_2$ is also context-unambiguous where
  $\dec{\epsilon}{\Gamma}{e_2}{\dict{\TC}~\tau'}$.
\end{lemma}

\begin{proof}
Similar to \cref{lemma:ambi-app2}.
\end{proof}

\begin{lemma}\label{lemma:abs-ih-ambi}
  If the principal type $\sigma$ of $\lambda{}x.\,e$ is context-unambiguous
  where $\dec{\epsilon}{\Gamma}{\lambda{}x.\,e}{\sigma}$, then the principal
  type $\sigma'$ of $e$ is context-unambiguous where
  $\dec{\epsilon}{(x \ty \tau_1), \Gamma}{e}{\sigma'}$.
\end{lemma}

\begin{proof}
  Similar to \cref{lemma:ambi-app2}.
  Intuitively, the $\tau_1$ simply moves from the derived type to the context.
\end{proof}

\subsection{Fully saturated types and terms}%
\label{sec:fully-saturated-types-and-terms}

Say we have the expression $e = \emph{show}$ where the variable \emph{show} has type $\forall{}a.\,\emph{Show}~a \Rightarrow a \rightarrow \emph{String}$.
Using the \textsc{Var} rule, we could derive the type $\sigma_1 = \forall{}a.\,\emph{Show}~a \Rightarrow a \rightarrow \emph{String}$ and the term $t_1 = \emph{show}$.
By additionally using the $\elimforall$ and $\elimct$ rules we could also derive the type $\sigma_2 = \emph{Int} \rightarrow \emph{String}$ and the term $t_2 = \emph{show}~\emph{Int}~d$ with $Q = d \ty \emph{Show}~\emph{Int}$.
While these terms were both derived for the same expression $e$, they are clearly not equivalent.

So before we can determine whether two terms are equivalent, we first have to saturate both derived terms by instantiating their type variables with types which we will call $\overline{\tau_a}$ and $\overline{\tau_b}$, and their constraints with evidence terms which we will call $\overline{\tev[C_1]}$ and $\overline{\tev[C_2]}$.

In the example above, we would choose $\overline{\tau_a} = \emph{Int}$ and $\overline{\tev[C_1]} = \emph{eqInt}$.
The second type and term are already saturated, hence $\overline{\tau_b}$ and $\overline{\tev[C_2]}$ are both empty.
The fully saturated types are then both \ensuremath{\Conid{Int}\to \Conid{String}}, the saturated terms will be \ensuremath{\Varid{show}\;\Conid{Int}\;\Varid{eqInt}} and \ensuremath{\Varid{show}\;\Conid{Int}\;\Varid{d}}.

Next, the terms $t_1$ and $t_2$ might still refer to assumptions from $Q_1$ or $Q_2$.
In our example above, $t_2$ refers to $d$.
Consequently, we also require evidence for these assumptions, namely $\overline{\tev[Q_1]}$ and $\overline{\tev[Q_2]}$.
We use the $\eta_1$ and $\eta_2$ substitutions to replace the respective dictionary variables in the terms.
For this example we choose $\eta_2 = [d \mapsto eqInt]$.
Thus, in the end, we want to know whether $\tred{\emph{show}~\emph{Int}~\emph{eqInt}}{\eta_2{}(\emph{show}~\emph{Int}~d)}$ is true, which clearly is the case, as after applying the substitution we have the same term on both sides.

\subsection{Equivalent Type Instantiations}%
\label{sec:equivalent-type-instantiations}

When saturating the types, we have to make sure the same types are chosen to instantiate the type variables in the constraints in both typing derivations.
\begin{example}\label{ex:show-a-b}
  $\textit{show} \ty \forall{}a.\,\textit{Show}~a \Rightarrow a \rightarrow
  \textit{String}$ and
  $\textit{show} \ty \forall{}b.\,\textit{Show}~b \Rightarrow b \rightarrow
  \textit{String}$.
\end{example}
In \cref{ex:show-a-b} $a$ and $b$ should be instantiated with the same type before we can determine whether their elaborations are equivalent.
If $a$ were instantiated with \textit{Int} and $b$ with \textit{Bool}, the elaborated terms would not be equivalent, as different dictionaries would have been chosen for the differing constraints.
It is tempting to express this as $\theta_1{}(\overline{C_1}) = \theta_2{}(\overline{C_2})$ where $\theta_1$ and $\theta_2$ are substitutions of the type variables $\overline{a}$ and $\overline{b}$.
In the \cref{ex:show-a-b}, this would work, but consider the following example:
\begin{example}\label{ex:eq-a-eq-b}
  Say there is some $e$ for which the type $\forall{}a\,b.\,(\textit{Eq}~a, \textit{Eq}~b) \Rightarrow (a, b) \rightarrow Bool$ and term $t_1 = t$ are derived.
  Then it is possible that for this same $e$, the following type and term are derived: $\forall{}a\,b.\,(\textit{Eq}~a, \textit{Eq}~b) \Rightarrow (b, a) \rightarrow Bool$ with $t_2 = \Lambda{}a.\,\Lambda{}b.\,\lambda{}(d_1 \ty \textit{Eq}~a).\,\lambda{}(d_2 \ty \textit{Eq}~b).\,t~b~a~d_2~d_1$.
  Note that $a$ and $b$ are swapped in the monotype ($\tau$), but the order of the type variables and the constraints is the same.
\end{example}
If we know that $\theta_1 = [a \mapsto \textit{Int}, b \mapsto \textit{Bool}]$ and $\theta_2 = [a \mapsto \textit{Int}, b \mapsto \textit{Bool}]$, we have that $\theta_1{}(\textit{Eq}~a, \textit{Eq}~b) = \theta_2{}(\textit{Eq}~a, \textit{Eq}~b)$.
But the saturated terms:
\begin{align*}
  &\ensuremath{\Varid{t}\;\Conid{Int}\;\Conid{Bool}\;\Varid{eqInt}\;\Varid{eqBool}} \text{ and} \\
  &(\Lambda{}a.\,\Lambda{}b.\,\lambda{}(d_1 \ty \textit{Eq}~a).\,\lambda{}(d_2 \ty \textit{Eq}~b).\,t~b~a~d_2~d_1)~\ensuremath{\Conid{Int}\;\Conid{Bool}\;\Varid{eqInt}\;\Varid{eqBool}} \\
  & \text{\quad which reduces to:} \\
  &\ensuremath{\Varid{t}\;\Conid{Bool}\;\Conid{Int}\;\Varid{eqBool}\;\Varid{eqInt}}
\end{align*}
are clearly not equivalent.

In summary, we cannot express this using just the type variables or the constraints, because their order, number (as they might already been instantiated in one of the derivations), and names might differ.

However, this can be expressed by requiring the monotypes to be equal: $\theta_1{}(\tau_1) = \theta_2{}(\tau_2)$.
If we apply this to \cref{ex:eq-a-eq-b}, this would lead to:
\begin{align*}
  [a \mapsto \ensuremath{\Conid{Int}}, b \mapsto \ensuremath{\Conid{Bool}}]((a, b) \rightarrow \ensuremath{\Conid{Bool}}) = (\ensuremath{\Conid{Int}}, \ensuremath{\Conid{Bool}}) \rightarrow \ensuremath{\Conid{Bool}} = [a \mapsto \ensuremath{\Conid{Bool}}, b \mapsto \ensuremath{\Conid{Int}}]((b, a) \rightarrow \ensuremath{\Conid{Bool}})
\end{align*}
Note that we do not actually need this information about the monotypes, we just want to know that on both sides the constraints are instantiated with the same types.

However, this requirement is too strict.
Consider the following example:
\begin{example}\label{ex:a-bool} $e_1~e_2$ where $e_1 \ty \forall{}a.\,(a \rightarrow a) \rightarrow \textit{Bool}$ and $e_2 \ty \forall{}a.\,a \rightarrow a$.
\end{example}
In one typing derivation, $a$ could be instantiated with \textit{Int}, giving us:\\
$e_1~e_2 \ty \textit{Bool} \rightsquigarrow (t_1~\textit{Int})~(t_2~\textit{Int})$.

In the other typing derivation, $a$ could be instantiated with \textit{Bool}, giving us:\\
$e_1~e_2 \ty \textit{Bool} \rightsquigarrow (t_1'~\textit{Bool})~(t_2'~\textit{Bool})$.

The condition $\theta_1{}(\tau_1') = \theta_2{}(\tau_2')$, or $(\textit{Int} \rightarrow \textit{Int}) = (\textit{Bool} \rightarrow \textit{Bool})$, is clearly not true, but the resulting terms are still equivalent, as equivalence is defined for erased terms, so the type applications disappear anyway.
Thus, we cannot prove the condition $\theta_1{}(\tau_1) = \theta_2{}(\tau_2)$ (for the subderivation) in this case.
Moreover, it does not seem necessary in this case.

However, if we add a constraint to \cref{ex:a-bool}:
\begin{example}\label{ex:a-eq-a-bool}
  $e_1 \ty \forall{}a.\,\textit{Eq}~a \Rightarrow (a \rightarrow a)
  \rightarrow \textit{Bool}$ and $e_2 \ty \forall{}a.\,a \rightarrow a$.
\end{example}

Now the choice for $a$ does matter, as it determines which dictionary will be chosen.
For instance, the type \textit{Bool} may be derived using the assumptions $Q_1 = \textit{Eq}~\textit{Int}$ when $a$ was instantiated with \textit{Int}, but the type \textit{Bool} may also be derived using the assumptions $Q_2 = \textit{Eq}~\textit{Bool}$ when $a$ was instantiated with \textit{Bool}.
What is especially unfortunate, is that the type to instantiate $a$ with is guessed!
The condition $\theta_1{}(\tau_1) = \theta_2{}(\tau_2)$, or in this example $\textit{Bool} = \textit{Bool}$, does not ensure that the type variables in the constraints are instantiated with the same types.
The condition does not seem strict enough.

Fortunately, this case is caught by the condition established in \cref{sec:principal-types}, namely that \emph{the principal type scheme of e must be context-unambiguous}.

Let us apply this to \cref{ex:a-eq-a-bool}.
First, we must determine the principal type scheme of $e_1~e_2$.
This is not \textit{Bool}, but in fact $\forall{}a.\,\textit{Eq}~a \Rightarrow \textit{Bool}$, which is an ambiguous type, as the type variable $a$ occurs in the constraints, but not in the monotype (\textit{Bool}).
The condition thus rules out the \cref{ex:a-eq-a-bool}.

Let us apply the condition to \cref{ex:a-bool}.
The principal type scheme is $\forall{}a.\,\textit{Bool}$, which is \emph{unambiguous}.
Consequently, the condition does not rule out this example.
Rightfully so, as derivations with different choices for $a$ will result in equivalent terms.

The condition $\theta_1{}(\tau_1) = \theta_2{}(\tau_2)$ along with the condition that the principal type scheme of $e$ must be unambiguous is too strict.
As a next step, we distinguish between type variables that occur in the constraints and type variables that do not.
Only the former matter, as their instantiation determines the constraints, which is what we are interested in.
Moreover, their presence can indicate ambiguity, whereas the presence of the former cannot.
The latter type variables may differ, as demonstrated in the previous paragraph.
This leads to the following definition:

\begin{definition}[Monotype equality modulo unconstrained type variables]
  We define $\tau_1$ and $\tau_2$ to be equal modulo unconstrained type variables of $\sigma$ if the principal type scheme $\sigma = \forall\overline{c}.\,\overline{C_p} \Rightarrow \tau_p$ is unambiguous and there exist $\theta_1'$, $\theta_2'$, $\theta_p$, $\theta_1{}(\tau_1) = \theta_1'{}(\theta_p{}(\tau_p))$ and $\theta_2{}(\tau_2) = \theta_2'{}(\theta_p{}(\tau_p))$ where $\dom{\theta_1', \theta_2'} \cap \ftv{\overline{C_p}} = \emptyset$.
\end{definition}

We use the notation $\taueq{\sigma}{\theta_1{}(\tau_1)}{\theta_1{}(\tau_2)}$ for this.

\newcommand{\thetainst}[1]{\theta_{#1}''}
\newcommand{\thetachoose}[1]{\theta_{#1}'''}
\newcommand{\thetachosen}[1]{\theta_{#1}'}
\begin{lemma}\label{lemma:tau-eq-app12}
  If:
  \begin{align}
    &\dec{Q_1}{\Gamma}{e_1}{\tau_1' \rightarrow \tau_2'} \notag \\
    &\dec{Q_2}{\Gamma}{e_1}{\tau_1'' \rightarrow \tau_2''} \notag \\
    &\dec{Q_1}{\Gamma}{e_2'}{\tau_1'} \notag \\
    &\dec{Q_2}{\Gamma}{e_2'}{\tau_1''} \notag \\
    &\taueq{\sigma}{\theta_1{}(\tau_2')}{\theta_2{}(\tau_2'')} \label{eq:tau-eq-app-tau-eq} \\
    &\text{the principal type } \sigma_1 \text{ is context-unambiguous where } \dec{\epsilon}{\Gamma}{e_1'}{\sigma_1} \label{eq:tau-eq-app-principal}
  \end{align}
  Then: $\taueq{\sigma_1}{\theta_1{}(\tau_1' \rightarrow \tau_2')}{\theta_2{}(\tau_1'' \rightarrow \tau_2'')}$.
\end{lemma}

\begin{proof}
  From  \cref{eq:tau-eq-app-tau-eq} follows that:
  \begin{align}
    &\theta_1{}(\tau_2') = \thetachosen{1}(\theta_p{}(\tau_{2p})) \label{eq:g1} \\
    &\theta_2{}(\tau_2'') = \thetachosen{2}(\theta_p{}(\tau_{2p})) \notag \\
    &\dom{\thetachosen{1}, \thetachosen{2}} \cap \ftv{\overline{C_p}} = \emptyset \label{eq:g3}
  \end{align}
  Where $\overline{C_p}$ are the constraints in $\sigma$.

  Say
  $\sigma_1 = \forall\overline{a}.\,\overline{C_p'} \Rightarrow \tau_{1p}
  \rightarrow \tau_{2p}$. Both $\tau_1' \rightarrow \tau_2'$ and
  $\tau_1'' \rightarrow \tau_2''$ are instantiation of this principal type for
  $e_1'$, which can be stated as:
  \begin{align}
    &\tau_1' \rightarrow \tau_2' = \thetainst{1}(\tau_{1p} \rightarrow \tau_{2p}) \label{eq:g2} \\
    &\tau_1'' \rightarrow \tau_2'' = \thetainst{2}(\tau_{1p} \rightarrow \tau_{2p}) \notag
  \end{align}

  From \cref{eq:tau-eq-app-principal} follows that:
  \begin{align}
    &\dom{\ftv{\theta_p{}(\tau_{1p})}} \cap \ftv{\overline{C_p'}} = \emptyset \label{eq:g4}
  \end{align}

  To prove
  $\taueq{\sigma_1}{\theta_1{}(\tau_1' \rightarrow
    \tau_2')}{\theta_2{}(\tau_1'' \rightarrow \tau_2'')}$, we must prove that
  there exist $\thetachoose{1}$, $\thetachoose{2}$, such that:
  \begin{align}
    &\theta_1{}(\tau_1') = \thetachoose{1}(\theta_p{}(\tau_{1p})) \label{eq:tb1} \\
    &\theta_2{}(\tau_1'') = \thetachoose{2}(\theta_p{}(\tau_{1p})) \notag \\
    &\theta_1{}(\tau_2') = \thetachoose{1}(\theta_p{}(\tau_{2p})) \label{eq:tb3} \\
    &\theta_2{}(\tau_2'') = \thetachoose{2}(\theta_p{}(\tau_{2p})) \notag \\
    &\dom{\thetachoose{1}, \thetachoose{2}} \cap \ftv{\overline{C_p'}} = \emptyset \label{eq:tb5}
  \end{align}

  \[ \text{Take } \thetachoose{1}(a) =
    \begin{cases}
      \thetachosen{1}(a)           & \quad \text{ if } a \in \dom{\thetachosen{1}} \\
      \theta_1{}(\thetainst{1}(a)) & \quad \text{ if } a \in \ftv{\theta_p{}(\tau_{1p})} \setminus \dom{\thetachosen{1}}
    \end{cases}
  \]

  \begin{itemize}
  \item \cref{eq:tb3} follows from the definition of $\thetachoose{1}$ and \cref{eq:g1}.
  \item \cref{eq:tb5} follows from the definition of $\thetachoose{1}$,
    \cref{eq:g3}, and \cref{eq:g4}.
  \item To prove \cref{eq:tb1}, we must prove:
    \[
      \theta_1{}(\tau_1') = \thetachoose{1}(\theta_p{}(\tau_{1p}))
    \]
    As $\tau_1' = \thetainst{1}(\tau_{1p})$, which follows from \cref{eq:g2},
    we can rewrite this to:
    \[
      \theta_1{}(\thetainst{1}(\tau_{1p})) = \thetachoose{1}(\theta_p{}(\tau_{1p}))
    \]
    Using \cref{lemma:subst-equiv}, we can restate our goal as:
    \[
      \forall a \in \ftv{\tau_{1p}}, \theta_1{}(\thetainst{1}(a)) = \thetachoose{1}(\theta_p{}(a))
    \]
    We distinguish two cases:
    \begin{itemize}
    \item If $a \in \dom{\thetachosen{1}}$, then, based on the definition of
      $\thetachoose{1}$, we must prove that:
      \begin{align}
        \theta_1{}(\thetainst{1}(a)) = \thetachosen{1}(\theta_p{}(a)) \label{eq:goal1}
      \end{align}
      From \cref{eq:g2}, we have that:
      \[
        \tau_2' = \thetainst{1}(\tau_{2p})
      \]
      We apply $\theta_1$ to both sides:
      \[
        \theta_1{}(\tau_2') = \theta_1{}(\thetainst{1}(\tau_{2p}))
      \]
      From \cref{eq:g1}, we get:
      \[
        \thetachosen{1}(\theta_p{}(\tau_{2p})) = \theta_1{}(\tau_2') = \theta_1{}(\thetainst{1}(\tau_{2p}))
      \]
      And thus:
      \[
        \theta_1{}(\thetainst{1}(\tau_{2p})) = \thetachosen{1}(\theta_p{}(\tau_{2p}))
      \]
      If we now apply \cref{lemma:subst-equiv}, we get:
      \[
        \forall a \in \ftv{\tau_{1p}}, \theta_1{}(\thetainst{1}(a)) = \thetachosen{1}(\theta_p{}(a))
      \]
      Since $\dom{\thetachosen{1}} \subseteq \ftv{\tau_{2p}}$, we get:
      \[
        \forall a \in \dom{\thetachosen{1}}, \theta_1{}(\thetainst{1}(a)) = \thetachosen{1}(\theta_p{}(a))
      \]
      From which our goal \cref{eq:goal1} follows.

    \item If
      $a \in \ftv{\theta_p{}(\tau_{1p})} \setminus \dom{\thetachosen{1}}$,
      then, based on the definition of $\thetachoose{1}$, we must prove that:
      \[
        \theta_1{}(\thetainst{1}(a)) = \theta_1{}(\thetainst{1}(\theta_p{}(a)))
      \]
      Since
      $a \in \ftv{\theta_p{}(\tau_{1p})} \setminus \dom{\thetachosen{1}}$, it
      must be that $a \notin \dom{\theta_p}$, and thus that
      $\theta_p{}(a) = a$. Using this, we can rewrite our goal to:
      \[
        \theta_1{}(\thetainst{1}(a)) = \theta_1{}(\thetainst{1}(a))
      \]
      Thereby finishing this case.
    \end{itemize}

  \end{itemize}

  Analogously for $\thetachoose{2}$.
\end{proof}

\begin{lemma}\label{lemma:tau-eq-app1}
  If:
  \begin{align*}
    &\dec{Q_1}{\Gamma}{e_1}{\tau_1' \rightarrow \tau_2'} \\
    &\dec{Q_2}{\Gamma}{e_1}{\tau_1'' \rightarrow \tau_2''} \\
    &\dec{Q_1}{\Gamma}{e_2'}{\tau_1'} \\
    &\dec{Q_2}{\Gamma}{e_2'}{\tau_1''} \\
    &\taueq{\sigma}{\theta_1{}(\tau_2')}{\theta_2{}(\tau_2'')} \\
    &\text{the principal type } \sigma_2 \text{ is context-unambiguous where } \dec{\epsilon}{\Gamma}{e_2'}{\sigma_2}
  \end{align*}
  Then: $\taueq{\sigma_1}{\theta_1{}(\tau_1')}{\theta_2{}(\tau_1'')}$.
\end{lemma}

\begin{proof}
  Similar to \cref{lemma:tau-eq-app12}
\end{proof}

\begin{lemma}\label{lemma:subst-equiv}
  $\theta_1{}(\tau) = \theta_2{}(\tau)$ iff
  $\forall a \in \ftv{\tau}, \theta_1{}(a) = \theta_2{}(a)$.
\end{lemma}

\begin{proof}
  Straightforward induction on $\tau$.
\end{proof}

\subsection{Canonical Evidence}%
\label{sec:canonical-evidence}

We assume a function $\cev{C}$, which maps a closed constraint $C$ to its unique, canonical evidence as given by the constraint solver.

An evidence term is canonical when it is equivalent with the canonical evidence of the corresponding constraint:
\begin{gather*}
  \sftc{\epsilon}{\tev}{\dict{\TC}~\upsilon} \text{ is canonical} \\
  \Leftrightarrow \\
  \tred{\tev}{\cev{C}} \text{ where } \elab{C}{C}{\dict{\TC}~\upsilon}
\end{gather*}
A context substitution $\phi$ maps all dictionary bindings in a context $\Gamma_\upsilon$ to evidence terms ($\dom{\phi} = \fdv{\Gamma_\upsilon}$).
Other bindings in the context are left unchanged.
We call this context without dictionary bindings $\Gamma_\upsilon'$ ($\fdv{\Gamma_\upsilon'} = \emptyset$).
The context substitution $\phi$ is canonical when, for each dictionary binding in the context, the evidence term obtained after applying $\phi$ canonical is:
\begin{gather*}
  \phi \ty \Gamma_\upsilon \rightarrow \Gamma_\upsilon' \text{ is canonical} \\
  \Leftrightarrow \\
  \forall~(d \ty \dict{\TC}~\upsilon) \in \Gamma_\upsilon, \sftc{\Gamma_\upsilon'}{\phi{}(d)}{\dict{\TC}~\upsilon} \text{ is canonical}
\end{gather*}
An evidence term in a context $\Gamma_\upsilon$ is canonical when the evidence
term obtained after applying any canonical context substitution canonical is:
\begin{gather*}
  \sftc{\Gamma_\upsilon}{\tev}{\dict{\TC}~\upsilon} \text{ is canonical} \\
  \Leftrightarrow \\
  \forall~\phi \ty \Gamma_\upsilon \rightarrow \Gamma_\upsilon', \phi \text{ is canonical}, \text{then } \sftc{\Gamma_\upsilon'}{\phi{}(\tev)}{\dict{\TC}~\upsilon} \text{ is canonical}
\end{gather*}

\begin{lemma}\label{lemma:canonicity-context-weakening}
  If $\tev$ is canonical in $\Gamma_\upsilon$, then $\tev$ is canonical in
  $\Gamma_\upsilon'$ where $\Gamma_\upsilon \subseteq \Gamma_\upsilon'$.
\end{lemma}

Say we have a program $e$ with the following type:
\begin{example}\label{ex:ab}
\[
  \sigma = \forall{}a\,b.\, (\textit{Eq}~a, \textit{Ord}~a, \textit{Num}~b) \Rightarrow (a, b) \rightarrow a
\]
\end{example}
Besides a typing derivation producing the type above, another typing derivation might produce the following type:
\begin{example}\label{ex:vw}
\[
  \forall{}v\,w.\, (\textit{Eq}~v, \textit{Ord}~v, \textit{Num}~w) \Rightarrow (v, w) \rightarrow v
\]
\end{example}
As mentioned in \cref{sec:fully-saturated-types-and-terms}, to reason about the dynamic semantics of the terms produced by the two typing derivations, the two types and terms must first be saturated.
This is done by choosing types for the type variables such that the two monotypes are equal modulo unconstrained type variables, see \cref{sec:equivalent-type-instantiations}.
We choose $\theta_1$ for \cref{ex:ab} and $\theta_2$ for \cref{ex:vw}:
\begin{align*}
  \theta_1 &= [a \mapsto \textit{Bool}, b \mapsto \textit{Int}] \\
  \theta_2 &= [v \mapsto \textit{Bool}, w \mapsto \textit{Int}]
\end{align*}
We check whether we have the desired monotype equality:
\begin{align*}
  &\taueq{\sigma}{\theta_1{}((a, b) \rightarrow a)}{\theta_2{}((v, w) \rightarrow v)} \\
  &\taueq{\sigma}{(\textit{Bool}, \textit{Int}) \rightarrow \textit{Bool}}{(\textit{Bool}, \textit{Int}) \rightarrow \textit{Bool}}
\end{align*}
The next step is to pick evidence for the constraints, closed with $\theta_1$ and $\theta_2$:
\begin{align*}
  \theta_1{}(\overline{C_1}) &= (\textit{Eq}~\textit{Bool}, \textit{Ord}~\textit{Bool}, \textit{Num}~\textit{Int}) \\
  \theta_2{}(\overline{C_2}) &= (\textit{Eq}~\textit{Bool}, \textit{Ord}~\textit{Bool}, \textit{Num}~\textit{Int})
\end{align*}
Now is the question: which evidence should be chosen?
To ensure the same dynamic semantics, we should make sure the same evidence is chosen for $\theta_1{}(\overline{C_1})$ and $\theta_2{}(\overline{C_2})$.
If we just pick the canonical evidence produced by the constraint solver, we are safe, as we know the same evidence will be chosen for both.
However, if we were to require that all evidence be canonical, we would exclude the possibility of passing a custom dictionary, since a custom dictionary is most likely not the canonical one.

So we relax this requirement and instead of requiring canonical evidence to be chosen for both, we just require pairwise equivalent evidence for both to be chosen, e.g., $\tev[11]$ for $[a \mapsto \textit{Bool}](\textit{Eq}~a)$, $\tev[21]$ for $[v \mapsto \textit{Bool}](\textit{Eq}~v)$ where $\tred{\tev[11]}{\tev[21]}$, and so on.
This is enough to guarantee us that we have the same dynamic semantics for both.
Furthermore, picking canonical evidence also satisfies this requirement, as the canonical evidence is of course equivalent with itself.

However, what if in the first typing derivation for $e$, the dictionary for $\textit{Eq}~a$ is used while in the second typing derivation for $e$, the super-class dictionary $\textit{Eq}~v$, stored in the dictionary for $\textit{Ord}~v$, is used?
We have that $\tred{\tev[11]}{\tev[21]}$ and $\tred{\tev[12]}{\tev[22]}$, but we do not know whether $\tred{\tev[11]}{\textit{superclass}(\tev[22])}$.

To remedy this, we formulate a more correct rule.
\newlist{evcanon}{enumerate}{10}
\setlist[evcanon]{label*=\arabic*.,leftmargin=1cm}

\begin{definition}[Evidence Canonicity]\label{def:evidence-canonicity}

  We define that $\textit{evidenceCanonicity}(\Gamma_\upsilon^1, \Gamma_\upsilon^2, \theta_1, \theta_2, \overline{\tev[1]}, \overline{\tev[2]},\allowbreak{} \overline{C_1}, \overline{C_2})$ iff all of the following hold:
  \begin{evcanon}
  \item $\overline{\tev[1] \ty \theta_1{}(C_1)} = \overline{\tev[1,0] \ty \theta_1{}(C_{1,0})} \wedge \overline{\tev[1,1] \ty \theta_1{}(C_{1,1})} \wedge \cdots \wedge \overline{\tev[1,k] \ty \theta_1{}(C_{1,k})}$
  \item $\overline{\tev[2] \ty \theta_2{}(C_2)} = \overline{\tev[2,0] \ty \theta_2{}(C_{2,0})} \wedge \overline{\tev[2,1] \ty \theta_2{}(C_{2,1})} \wedge \cdots \wedge \overline{\tev[2,k] \ty \theta_2{}(C_{2,k})}$
  \item All $\overline{\tev[1,0] \ty \theta_1{}(C_{1,0})}$ and $\overline{\tev[2,0] \ty \theta_2{}(C_{2,0})}$ are canonical in respectively $\Gamma_\upsilon^1$ and $\Gamma_\upsilon^2$.
  \item For all $i = 1 \ldots k$, there exists a $C^p_i$ such that:
    \begin{evcanon}
    \item $\exists (\tev[1,i]^p \ty \theta_1{}(C_{i}^p))$
    \item $\exists (\tev[2,i]^p \ty \theta_2{}(C_{i}^p))$
    \item $\tred{\tev[1,i]^p}{\tev[2,i]^p}$

    \item For all $C_1$, there exist a $\tev[1]'$ and $\overline{C_1'}$ such that
      \begin{evcanon}
      \item $\entails{\QQ \wedge \overline{d' \ty C_1'} \wedge d \ty C_i^p}{\tev[1]' \ty C_1}$
      \item if $\entails{\QQ \wedge \overline{C_{1,0}} \wedge \left( \bigwedge_{j \le i}{\overline{C_{1,j}}} \right)}{\tev[1] \ty C_1}$, then $\tred{\tev[1] }{[d \mapsto \tev[1,i]^p][\overline{d'\mapsto \tev[1]''}]\tev[1]'}$ for some $\entails{\QQ \wedge \overline{C_{1,0}} \wedge \left( \bigwedge_{j < i}{\overline{C_{1,j}}} \right)}{\overline{\tev[1]'' \ty C_1'}}$
      \end{evcanon}
      For any given $C_1$, we denote these $\overline{C_1'}$ as $(C_1 \setminus C_{i}^p)$ and $\overline{\tev[1]'}$ as the evidence for $C_1\setminus C_{i}^p$.

    \item For all $C_2$, there exist a $\tev[2]'$ and $\overline{C_2'}$ such that
      \begin{evcanon}
      \item $\entails{\QQ \wedge \overline{d' \ty C_2'} \wedge d \ty C_i^p}{\tev[2]' \ty C_2}$
      \item if $\entails{\QQ \wedge \overline{C_{2,0}} \wedge \left( \bigwedge_{j \le i}{\overline{C_{2,j}}} \right)}{\tev[2] \ty C_2}$, then $\tred{\tev[2] }{[d \mapsto \tev[2,i]^p][\overline{d'\mapsto \tev[2]''}]\tev[2]'}$ for some $\entails{\QQ \wedge \overline{C_{2,0}} \wedge \left( \bigwedge_{j < i}{\overline{C_{2,j}}} \right)}{\overline{\tev[2]'' \ty C_2'}}$
      \end{evcanon}
      For any given $C_2$, we denote this $C_2'$ as $(C_2 \setminus C_{i}^p)$ and $\tev[2]'$ as $\tev[C_2\setminus C_{i}^p]$.

    \item For all $(\tev[1,i] \ty \theta_1{}(C_{1,i})) \in \overline{\tev[1,i] \ty \theta_1{}(C_{1,i})}$, there exists $\tev[1,i]',C_{1,i}'$ such that:
      \begin{evcanon}
      \item $\tred{\theta_1([\overline{d_j \mapsto \tev[1,j]^p}]_{j \le i}[d \mapsto \tev[1,i]'']\tev[1,i]')}{\tev[1,i]}$
      \item $\tev[1,i]'$ and $C_{1,i}'$ are uniquely determined by $C_{1,i}$ and the $\overline{C_{j}^p}$ for $j \le i$.
      \item $\entails{\QQ \wedge \overline{C_{1,0}}}{\tev[1,i]'' \ty C_{1,i}'}$
      \end{evcanon}

    \item For all $(\tev[2,i] \ty \theta_2{}(C_{2,i})) \in \overline{\tev[2,i] \ty \theta_2{}(C_{2,i})}$, there exists $\tev[2,i]',C_{2,i}'$ such that:
      \begin{evcanon}
      \item $\tred{\theta_2([\overline{d_j \mapsto \tev[2,j]^p}]_{j \le i}[d \mapsto \tev[2,i]'']\tev[2,i]')}{\tev[2,i]}$
      \item $\tev[2,i]'$ and $C_{2,i}'$ are uniquely determined by $C_{2,i}$ and the $\overline{C_{j}^p}$ for $j \le i$.
      \item $\entails{\QQ \wedge \overline{C_{2,0}}}{\tev[2,i]'' \ty C_{2,i}'}$
      \end{evcanon}
    \end{evcanon}
  \end{evcanon}
\end{definition}

\begin{lemma}
  \label{lem:canonical-ev-form}
  If
  \begin{itemize}
  \item $\textit{evidenceCanonicity}(\Gamma_\upsilon^1, \Gamma_\upsilon^2,
    \theta_1, \theta_2, \overline{\tev[1]}, \overline{\tev[2]}, \overline{C_1},
    \overline{C_2})$
  \item $\entails{\QQ \wedge \overline{C_1}}{\tev[1] \ty C_1}$
  \end{itemize}
  Then there exists a $\tev[1]'$ and $C_1'$ such that (using the existentially quantified objects in the evidenceCanonicity definition)
  \begin{itemize}
  \item $\tred{[\overline{d_j \mapsto
        \tev[1,j]^p}][d\mapsto \tev[1]'']\tev[1]'}{\tev[1]}$
  \item $\tev[1]'$ and $C_1'$ are uniquely determined by $C_1$ and the $\overline{C_{j}^p}$
  \item $\entails{\QQ \wedge \overline{C_{1,0}}}{\tev[1]'' \ty C_1'}$
  \end{itemize}

  The symmetric result holds for constraints that follow from $C_2$.
\end{lemma}
\begin{proof}
  By symmetry, we only prove the stated result, not the symmetric one.

  We know that $\overline{\tev[1] \ty \theta_1{}(C_1)} = \overline{\tev[1,0] \ty \theta_1{}(C_{1,0})} \wedge \overline{\tev[1,1] \ty \theta_1{}(C_{1,1})} \wedge \cdots \wedge \overline{\tev[1,k] \ty \theta_1{}(C_{1,k})}$.

  We prove the result by induction on $k$.
  For $k = 0$, we have that $\entails{\QQ \wedge \overline{C_{1,0}}}{\tev[1] \ty C_1}$, so we can take $\tev[1]' = d$ and $C_1' = C_1$ and the result follows trivially.

  For $k +1$, we have that $\entails{\QQ \wedge \overline{C_{1,0}} \wedge \overline{C_{1,1}} \wedge \cdots \wedge \overline{C_{1,k+1}}}{\tev[1] \ty C_1}$.
  By evidence canonicity, we know that there exist $\tev[1,k]'$ and $\overline{C_{1,k}'}$ such that
  \begin{itemize}
  \item $\entails{\QQ \wedge \overline{d' \ty C_{1,k}'} \wedge d \ty C_{k+1}^p}{\tev[1,k]' \ty
      C_1}$
  \item For all $\tev[1],C_1$, if $\entails{\QQ \wedge \overline{C_{1,0}} \wedge \left( \bigwedge_{j \le k+1}{\overline{C_{1,j}}} \right)}{\tev[1] \ty C_1}$, then $\tred{\tev[1] }{[d \mapsto \tev[1,k+1]^p][\overline{d' \mapsto \tev[1]''}]\tev[1]'}$ for some $\entails{\QQ \wedge \overline{C_{1,0}} \wedge \left( \bigwedge_{j < k+1}{\overline{C_{1,j}}} \right)}{\overline{\tev[1]'' \ty C_1'}}$
  \end{itemize}

  Since $\entails{\QQ \wedge \overline{C_{1,0}} \wedge \overline{C_{1,1}} \wedge \cdots \wedge \overline{C_{1,k+1}}}{\tev[1] \ty C_1}$, the second point particularly implies (for our $\tev[1],C_1$) that $\tred{\tev[1] }{[d \mapsto \tev[1,k+1]^p][\overline{d'\mapsto \tev[1]''}]\tev[1,k]'}$ for some $\entails{\QQ \wedge \overline{C_{1,0}} \wedge \left( \bigwedge_{j < k+1}{\overline{C_{1,j}}} \right)}{\overline{\tev[1]'' \ty C_{1,k}'}}$.

  But by induction, for this new entailment, we get $\tev[1,k]''$ and $C_{1,k}''$, such that
  \begin{itemize}
  \item $\tred{[\overline{d_j \mapsto
        \tev[1,j]^p}][d\mapsto \tev[1]'']\tev[1,k]''}{\tev[1]''}$
  \item $\tev[1,k]''$ and $C_{1,k}''$ are uniquely determined by $C_{1,k}'$ and the $\overline{C_{j}^p}$
  \item $\entails{\QQ \wedge \overline{C_{1,0}}}{\tev[1]'' \ty C_{1,k}'}$
  \end{itemize}

  We can now construct $\tev[1]' = [d \mapsto \tev[1,k+1]^p][d'\mapsto \tev[1,k]'']\tev[1,k]'$  and $C_1' = C_{1,k}''$ and the desired equalities follow.
\end{proof}

\begin{lemma}
  \label{lem:canonical-ev-equiv}
  If
  \begin{itemize}
  \item $\textit{evidenceCanonicity}(\Gamma_\upsilon^1, \Gamma_\upsilon^2,
    \theta_1, \theta_2, \overline{\tev[1]}, \overline{\tev[2]}, \overline{C_1},
    \overline{C_2})$
  \item $\tev[1] \ty \theta_1(C_1) \in \overline{\tev[1] \ty \theta_1(C_1)}$
  \item $\tev[2] \ty C_2 \in \overline{\tev[2] \ty C_2}$
  \item $C_1 = C_2$
  \item The constraint solver produces canonical evidence.
  \end{itemize}
  Then $\tred{\tev[1]}{\tev[2]}$.
\end{lemma}
\begin{proof}
  Since $C_1 \in \overline{C_1}$, we trivially have that $\entails{\QQ \wedge \overline{\tev[1] \ty C_1}}{\tev[1] \ty C_1}$.

  By the previous lemma, there exists a $\tev[1]'$ and $C_1'$ such that (using the existentially quantified objects in the evidenceCanonicity definition)
  \begin{itemize}
  \item $\tred{[\overline{d_j \mapsto
        \tev[1,j]^p}][d\mapsto \tev[1]'']\tev[1]'}{\tev[1]}$
  \item $\tev[1]'$ and $C_1'$ are uniquely determined by $C_1$ and the $\overline{C_{j}^p}$
  \item $\entails{\QQ \wedge \overline{\tev[1,0] \ty C_{1,0}}}{\tev[1]'' \ty C_1'}$
  \end{itemize}

  Because we also know that the $\overline{\tev[1,0] \ty C_{1,0}}$ are canonical in $\Gamma_\upsilon^1$, it follows from the last entailment and the fact that the constraint solver produces canonical evidence, that $\tev[1]''$ is also canonical, i.e. $[\overline{d_{1,0}\mapsto \tev[1,0]}]\tev[1]''$ is equivalent to the canonical evidence for $C_1'$.

  Likewise, there exists a $\tev[2]'$ and $C_2'$ such that (using the existentially quantified objects in the evidenceCanonicity definition)
  \begin{itemize}
  \item $\tred{[\overline{d_j \mapsto
        \tev[2,j]^p}][d\mapsto \tev[2]'']\tev[2]'}{\tev[2]}$
  \item $\tev[2]'$ and $C_2'$ are uniquely determined by $C_2$ and the $\overline{C_{j}^p}$
  \item $\entails{\QQ \wedge \overline{\tev[2,0] \ty C_{2,0}}}{\tev[2]'' \ty C_2'}$
  \end{itemize}

  Again, we also know that the $\overline{\tev[2,0] \ty C_{2,0}}$ are canonical in $\Gamma_\upsilon^2$, it follows from the last entailment and the fact that the constraint solver produces canonical evidence, that $\tev[2]''$ is also canonical, i.e. $[\overline{d_{2,0}\mapsto \tev[2,0]}]\tev[2]''$ is equivalent to the canonical evidence for $C_2'$.

  Because $\tev[1]'$, $\tev[2]'$, $C_1'$ and $C_2'$ are uniquely determined by $C_1 = C_2$ and the $\overline{C_j^p}$, we have that $\tev[1]' = \tev[2]'$, $C_1' = C_2'$ and by the above that $\tred{\tev[1]''}{\tev[2]''}$.
  Since we also have that $\tred{\tev[1,i]^p}{\tev[2,i]^p}$ for all $i$, it follows by equivalence congruence that $\tred{\tev[1]}{\tev[2]}$.
\end{proof}

\begin{lemma}\label{lem:coherence-asmpt}
  If
  \begin{itemize}
  \item For all $C'$ such that $\entails{\QQ \wedge C^p}{\tev \ty C'}$, we have that one of the following holds:
    \begin{itemize}
    \item $\nentails{\QQ\wedge \overline{C}}{C'}$
    \item $\entails{\QQ}{\tev \ty C'}$
    \end{itemize}
  \end{itemize}
  then for all $C$, there exist a $\tev'$ and $\overline{C'}$ such that
  \begin{itemize}
  \item $\entails{\QQ \wedge \overline{d' \ty C'} \wedge d \ty C^p}{\tev' \ty C}$
  \item For all $\tev$, if $\entails{\QQ \wedge \overline{C} \wedge d \ty C^p}{\tev \ty C}$, then $\tred{\tev }{[d'\mapsto \tev'']\tev'}$ for some $\entails{\QQ \wedge \overline{C}}{\overline{\tev'' \ty C'}}$
  \end{itemize}
\end{lemma}
\begin{proof}[Proof sketch]
  We cannot actually prove this lemma because our proof is parametric in the entailment relation, but we sketch why we think it holds true for constraint solvers that deal with type classes.

  Take $C^p$ as given constraint and $C$ as wanted constraint and apply the constraint solver to obtain a not-further-simplifiable set of simplified constraints.
  Define $\overline{C'}$ as these simplified constraints and $\tev'$ as the evidence produced by the constraint solver.
  Then the first constraint entailment clearly holds.

  To understand why the second holds, consider how the constraint solver would go about finding evidence $\tev$: essentially, it would consider the wanted constraint $C$, recursively apply instances to it from $\QQ$ to obtain new wanted constraints, and then drop those that are entailed by one of the givens.
  The insight here is that for every wanted constraint separately, this process is largely deterministic: when an instance is applied, there can be no other choice because there are no overlapping instances.
  When a wanted constraint is resolved because it is entailed by $C^p$ (or one of its descendants), then we know by assumption that either the evidence found does not depend on $C^p$ or the wanted constraint is not also entailed by $\QQ \wedge \overline{C}$, so there is no other choice here either.
  If a wanted constraint is removed because it is entailed by $\overline{C}$, then it must be part of $\overline{C'}$ and the result follows.
\end{proof}

\subsection{Target Language Lemmas}%
\label{sec:target-language-lemmas}

We need some additional lemmas handling about the target language.

\begin{lemma}\label{lemma:subst-typing}
  If $\sftc{\Gamma_\upsilon}{\eta}{\Gamma_\upsilon'}$ and
  $\sftc{\Gamma_\upsilon'}{t}{\upsilon}$, then
  $\sftc{\Gamma_\upsilon}{\eta{}(t)}{\upsilon}$.
\end{lemma}

\begin{proof}
  Easy induction on the proof of $\sftc{\Gamma_\upsilon'}{t}{\upsilon}$.
\end{proof}

\begin{example}
  We demonstrate the lemma above with:
  \begin{align*}
    \Gamma_\upsilon &= (x \ty \textit{Int}) \\
    \Gamma_\upsilon' &= (y \ty (\textit{Int}, \textit{Int})) \\
    \eta &= [y \mapsto (x, x)] \\
    t &= \textit{fst}~y \\
    \upsilon &= \textit{Int}
  \end{align*}
  Given
  $\sftc{(x \ty \textit{Int})}{[y \mapsto (x, x)]}{(y \ty (\textit{Int},
    \textit{Int}))}$ and
  $\sftc{(y \ty (\textit{Int}, \textit{Int}))}{\textit{fst}~y}{\textit{Int}}$,
  we have that
  $\sftc{(x \ty \textit{Int})}{[y \mapsto (x,
    x)](\textit{fst}~y)}{\textit{Int}}$.
\end{example}

\begin{lemma}[Context Weakening]\label{lemma:systemf-context-weakening}
  If $\sftc{\Gamma_\upsilon}{t}{\upsilon}$ then
  $\sftc{\Gamma_\upsilon'}{t}{\upsilon}$ where $\Gamma_\upsilon \subseteq
  \Gamma_\upsilon'$.
\end{lemma}

\begin{proof}
Straightforward induction.
\end{proof}



\subsection{Coherence Proof}%
\label{sec:actual-coherence-proof}

Finally, we can state the coherence theorem and prove it.
The proof is itself technically novel and simpler than previous proofs in the literature, as it avoids the use of a specific algorithmic typing relation (instead simply assuming that principal typing holds, see \cref{thm:principal-types}) and is parametric in the constraint implication judgement (simply assuming that it produces canonical evidence when given canonical evidence for assumptions, see \cref{sec:canonical-evidence}).

\begin{theorem}[Coherence]\label{theorem:coherence}
  Given typing derivations $\decelab{Q}{\Gamma}{e}{\sigma}{t_1}$ and
  $\decelab{Q}{\Gamma}{e}{\sigma}{t_2}$, if the principal type $\sigma_0$ of
  $e$ is context-unambiguous, and the constraint solver produces canonical
  evidence, it must be that $\tred{t_1}{t_2}$.
\end{theorem}

\begin{proof}
  Say $\sigma = \forall\overline{a}.\,\overline{C} \Rightarrow \tau$.
  Pick fresh $\overline{c}$ and $\overline{d}$ such that:
  \begin{align*}
    &\overline{c} \mathrel{\#} \ftv{Q, \Gamma, \sigma} \\
    &\forall{}d \in \overline{d}, d \notin \fv{t_1, t_2}
  \end{align*}

  Apply \cref{lemma:coherence} with the two typing derivations above, use $\overline{\tau_a} = \overline{\tau_b} = \overline{c}$, use evidence variables for $\overline{\tev[Q_1]} = \overline{\tev[Q_2]}$ and for $\overline{\tev[C_1]} = \overline{\tev[C_2]}$. %
  Evidence canonicity follows from the fact that all evidence terms produced by the constraint solver are canonical, and thus pairwise equivalent, combined with the fact that the $\overline{d}$ are also pairwise equivalent.

  This gives us:
  \[
    \tred{(\eta{}(t_1)~\overline{\upsilon_c}~\overline{d})}
    {(\eta{}(t_2)~\overline{\upsilon_c}~\overline{d})}
  \]
  We now apply \textsc{EraseTyApp} inside the evidence applications:
  \[
    \tred{(\eta{}(t_1)~\overline{d})}
    {(\eta{}(t_2)~\overline{d})}
  \]
  We can apply \textsc{Abs} on both sides for all $\overline{d}$:
  \[
    \tred{\lambda(\overline{d \ty \upsilon_C}).\,(\eta{}(t_1)~\overline{d})}
    {\lambda(\overline{d \ty \upsilon_C}).\,(\eta{}(t_2)~\overline{d})}
  \]
  Since we know that $\forall{}d \in \overline{d}, d \notin \fv{t_1, t_2}$, we
  can apply \etared{} on both sides:
  \[
    \tred{\eta{}(t_1)}{\eta{}(t_2)}
  \]
  From this follows our goal:
  \[
    \tred{t_1}{t_2}
  \]
\end{proof}


\newcommand{\evidenceCanonicity}{\textit{evidenceCanonicity}\hspace{1pt}}

We use the following generalised lemma with a stronger induction hypothesis to prove coherence:

\begin{lemma}\label{lemma:coherence}
  Let $\elabqqg{\QQ \wedge Q_1\wedge\overline{C_1}}{\Gamma}{\Gamma_{\upsilon'}^1}$ and
  $\elabqqg{\QQ \wedge Q_2 \wedge\overline{C_2}}{\Gamma}{\Gamma_{\upsilon'}^2}$. If:
  \begin{align}
    & \decelab{Q_1}{\Gamma}{e}{\forall\,\overline{a}.\,\overline{C_1} \Rightarrow \tau_1}{t_1} \label{eq:hty1} \\
    & \decelab{Q_2}{\Gamma}{e}{\forall\,\overline{b}.\,\overline{C_2} \Rightarrow \tau_2}{t_2} \label{eq:hty2} \\
    & \overline{a_f} = \ftv{Q_1, \Gamma} \label{eq:coh-as} \\
    & \overline{b_f} = \ftv{Q_2, \Gamma} \label{eq:coh-bs} \\
    & \Gamma_{\upsilon'}^1 \subseteq \Gamma_{\upsilon}^1\\
    & \Gamma_{\upsilon'}^2 \subseteq \Gamma_{\upsilon}^2\\
    & \Gamma_\upsilon^1 \vdash \overline{\tev[Q_1]} \ty [\overline{a_f \mapsto \tau_{a_f}}]Q_1 \label{eq:coh-tevQ1} \\
    & \Gamma_\upsilon^2 \vdash \overline{\tev[Q_2]} \ty [\overline{b_f \mapsto \tau_{b_f}}]Q_2 \label{eq:coh-tevQ2} \\
    & \theta_1 = [\overline{a \mapsto \tau_a}, \overline{a_f \mapsto \tau_{a_f}}] \label{eq:coh-theta1} \\
    & \theta_2 = [\overline{b \mapsto \tau_b}, \overline{b_f \mapsto \tau_{b_f}}] \label{eq:coh-theta2} \\
    & \Gamma_\upsilon^1 \vdash \overline{\tev[C_1] \ty \theta_1{}(C_1)} \label{eq:coh-tevC1} \\
    & \Gamma_\upsilon^2 \vdash \overline{\tev[C_2] \ty \theta_2{}(C_2)} \label{eq:coh-tevC2} \\
    & \eta_1 = [\overline{d_1 \mapsto \tev[Q_1]}] \text{ for each } (d_1 \ty C_1) \in Q_1 \label{eq:coh-eta1} \\
    & \eta_2 = [\overline{d_2 \mapsto \tev[Q_2]}] \text{ for each } (d_2 \ty C_2) \in Q_2 \label{eq:coh-eta2} \\
    & \text{the principal type } \sigma \text{ is context-unambiguous where } \dec{\epsilon}{\Gamma}{e}{\sigma} \label{eq:coh-principal} \\
    & \taueq{\sigma}{\theta_1{}(\tau_1)}{\theta_2{}(\tau_2)} \label{eq:coh-taueq} \\
    & \evidenceCanonicity{}(\Gamma_\upsilon^1, \Gamma_\upsilon^2, \theta_1, \theta_2, (\overline{\tev[Q_1]} \wedge \overline{\tev[C_1]}), (\overline{\tev[Q_2]} \wedge \overline{\tev[C_2]}), Q_1 \wedge \overline{C_1}, Q_2 \wedge \overline{C_2}) \label{eq:coh-evcanon}
  \end{align}

  Then
  $\tred{(\eta_1{}(t_1)~\overline{\upsilon_a}~\overline{\tev[C_1]})}
  {(\eta_2{}(t_2)~\overline{\upsilon_b}~\overline{\tev[C_2]})}$ where
  $\elab{\tau}{\overline{\tau_a}}{\overline{\upsilon_a}}$ and
  $\elab{\tau}{\overline{\tau_b}}{\overline{\upsilon_b}}$.
\end{lemma}

We have two different typing derivations \cref{eq:hty1} and \cref{eq:hty2} for the same program $e$.
As there may still be free type variables, we collect them in \cref{eq:coh-as} and \cref{eq:coh-bs} and close over them by replacing them with $\overline{\tau_{a_f}}$ and $\overline{\tau_{b_f}}$.
We need evidence, $\overline{\tev[Q_1]}$ and $\overline{\tev[Q_2]}$, for the closed assumptions $Q_1$ and $Q_2$, see \cref{eq:coh-tevQ1} and \cref{eq:coh-tevQ2}.
In \cref{eq:coh-theta1} and \cref{eq:coh-theta2} we define substitutions for the free type variables \cref{eq:coh-as} and \cref{eq:coh-bs}, and also for the type variables $\overline{a}$ and $\overline{b}$ that will be free in $\overline{C_1}$ and $\tau_1$, and $\overline{C_2}$ and $\tau_2$ respectively.
For notational convenience we sometimes apply $\theta_i$ to $Q_i$, even though $\theta_i$ contains the substitutions $[\overline{a \mapsto \tau_a}]$ or $[\overline{b \mapsto \tau_b}]$ that do not have to be applied to $Q_i$.
However, since these variables are not free in $Q_i$, this does not matter anyway and permits us to use a single $\theta_i$ instead of two different ones.
We also need evidence, $\overline{\tev[C_1]}$ and $\overline{\tev[C_2]}$, for the closed constraints of the derivations, see \cref{eq:coh-tevC1} and \cref{eq:coh-tevC2}.
Next, in \cref{eq:coh-eta1} and \cref{eq:coh-eta2} we define $\eta_1$ and $\eta_2$ to be the substitutions that map the dictionary variables mentioned in $Q_1$ and $Q_2$ to the evidence we have for them, see \cref{eq:coh-tevQ1} and \cref{eq:coh-tevQ2}.
We will use the substitutions to make sure the terms no longer refer to the assumptions $Q_1$ and $Q_2$.
As discussed in \cref{sec:principal-types}, we require \cref{eq:coh-principal}.
See \cref{sec:equivalent-type-instantiations} for \cref{eq:coh-taueq}, and \cref{sec:canonical-evidence} for \cref{eq:coh-evcanon}.
For convenience we sometimes use $Q$ to mean the constraints in $Q$ without their evidence.
When not relevant, we omit the evidence in entailments.

Finally, we prove that the fully saturated terms are equivalent.
The substitutions $\theta_1$ and $\theta_2$ should also be applied to $t_1$ and $\overline{\upsilon_a}$, and to $t_2$ and $\overline{\upsilon_b}$, respectively.
However, as the equivalence relation erases the types, we omit these substitutions for simplicity.

\begin{proof}

  Throughout this proof, we will use $\upsilon_x$ to signify the System F type
  obtained by the elaboration $\elab{\tau}{\tau_x}{\upsilon_x}$ for some $x$
  without explicitly mentioning the elaboration, at least when it is clear
  from the context.

  We perform simultaneous induction on the typing derivations \cref{eq:hty1} and
  \cref{eq:hty2}. The table below displays the cases we prove. We eliminate
  symmetrical cases (empty cells), impossible cases (marked in grey). We also
  group together common cases.

  \begin{center}
  \newcommand{\rot}[1]{\rotatebox{90}{#1}}
  \newcommand{\impos}{\cellcolor{gray}}
  \begin{tabular}{ c | c c c c c c c c}
    & \rot{$\introforall$} & \rot{$\elimforall$} & \rot{$\introct$} & \rot{$\elimct$} & \rot{\textsc{DictApp}} & \rot{\textsc{Var}} & \rot{\textsc{Abs}} & \rot{\textsc{App}} \\
    \hline
    $\introforall$   & \multicolumn{8}{c|}{\cref{foralli}} \\
    \cline{2-9}
    $\elimforall$    &  & \multicolumn{7}{|c|}{\cref{foralle}} \\
    \cline{3-9}
    $\introct$       &  &  & \multicolumn{6}{|c|}{\cref{cti}} \\
    \cline{4-9}
    $\elimct$        &  &  &  & \multicolumn{5}{|c|}{\cref{cte}} \\
    \cline{5-9}
    \textsc{DictApp} &  &  &  &  & \multicolumn{1}{|c}{\cref{dictapp-dictapp}} & \multicolumn{3}{|c|}{\impos} \\
    \cline{6-7}
    \textsc{Var}     &  &  &  &  &  & \multicolumn{1}{|c}{\cref{var-var}} & \multicolumn{2}{|c|}{\impos} \\
    \cline{7-8}
    \textsc{Abs}     &  &  &  &  &  &  & \multicolumn{1}{|c}{\cref{abs-abs}} & \multicolumn{1}{|c|}{\impos} \\
    \cline{8-9}
    \textsc{App}     &  &  &  &  &  &  &  & \multicolumn{1}{|c|}{\cref{app-app}} \\
    \cline{9-9}
  \end{tabular}
  \end{center}

  \newcommand{\hspacev}{\phantom{\evidenceCanonicity{}(}}


  \begin{enumerate}[itemsep=2em]

  \item[\pcase{foralli}.] The $\introforall$ rule was last used in the first
    or second derivation. By symmetry, we assume the second. We have that:
      \begin{align}
        &\decelab{Q_1}{\Gamma}{e}{\forall\,\overline{a}.\,\overline{C_1} \Rightarrow \tau_1}{e} \notag \\
        &\overline{a_f} = \ftv{Q_1, \Gamma} \notag \\
        &\Gamma_\upsilon^1 \vdash \overline{\tev[Q_1]} \ty [\overline{a_f \mapsto \tau_{a_f}}]Q_1 \notag \\
        &\theta_1 = [\overline{a \mapsto \tau_a}, \overline{a_f \mapsto \tau_{a_f}}] \notag \\
        &\Gamma_\upsilon^1 \vdash \overline{\tev[C_1] \ty \theta_1{}(C_1)} \notag \\
        &\eta_1 = [\overline{d_1 \mapsto \tev[Q_1]}] \text{ for each } (d_1 \ty C_1) \in Q_1 \notag \\
        & \notag \\
        &\decelab{Q_2}{\Gamma}{e}{\forall{}c\overline{b}.\,\overline{C_2} \Rightarrow \tau_2}{\Lambda{}c.\,t_2'} \notag \\
        &\decelab{Q_2}{c, \Gamma}{e}{\forall{}\overline{b}.\,\overline{C_2} \Rightarrow \tau_2}{t_2'} \notag \\
        &c \notin \ftv{Q_2, \Gamma} \label{eq:foralli-cftv} \\
        &\overline{b_f} = \ftv{Q_2, \Gamma} \label{eq:foralli-ftv} \\
        &\Gamma_\upsilon^2 \vdash \overline{\tev[Q_2]} \ty [\overline{b_f \mapsto \tau_{b_f}}]Q_2 \label{eq:foralli-tevQ2} \\
        &\theta_2 = [c \mapsto \tau_c, \overline{b \mapsto \tau_b}, \overline{b_f \mapsto \tau_{b_f}}] \label{eq:foralli-theta2} \\
        &\Gamma_\upsilon^2 \vdash \overline{\tev[C_2] \ty \theta_2{}(C_2)} \label{eq:foralli-tevC2} \\
        &\eta_2 = [\overline{d_2 \mapsto \tev[Q_2]}] \text{ for each } (d_2 \ty C_2) \in Q_2 \label{eq:foralli-eta2} \\
        &\text{the principal type } \sigma \text{ is context-unambiguous where } \dec{\epsilon}{\Gamma}{e}{\sigma} \label{eq:foralli-ambi} \\
        &\taueq{\sigma}{\theta_1{}(\tau_1)}{\theta_2{}(\tau_2)} \label{eq:foralli-tau-eq} \\
        &\evidenceCanonicity{}(\Gamma_\upsilon^1, \Gamma_\upsilon^2, \theta_1, \theta_2, (\overline{\tev[Q_1]} \wedge \overline{\tev[C_1]}), (\overline{\tev[Q_2]} \wedge \overline{\tev[C_2]}), Q_1 \wedge \overline{C_1}, Q_2 \wedge \overline{C_2}) \label{eq:foralli-ev-canon}
      \end{align}
      And must prove that:
      \begin{align*}
        &\tred{(\eta_1{}(t_1)~\overline{\upsilon_a}~\overline{\tev[C_1]})}
          {(\eta_2{}(\Lambda{}c.\,t_2')~\upsilon_c~\overline{\upsilon_b}~\overline{\tev[C_2]})}
      \end{align*}
      We can rewrite this to:
      \begin{align}
        &\tred{(\eta_1{}(t_1)~\overline{\upsilon_a}~\overline{\tev[C_1]})}
          {((\Lambda{}c.\,\eta_2{}(t_2'))~\upsilon_c~\overline{\upsilon_b}~\overline{\tev[C_2]})} \label{eq:foralli-goal}
      \end{align}
      We have the following induction hypothesis:
      \begin{align}
        &\forall~\overline{b_f'}~\overline{\tau_{b_f}'}~\overline{\tev[Q_2]'}~\overline{\tev[C_2]'}~\overline{\tau_b'}~\theta_2'~\eta_2'~\sigma', \notag \\
        &\quad\overline{b_f'} = \ftv{Q_2, (c, \Gamma)} \rightarrow \notag \\
        &\quad{}(c, \Gamma_\upsilon^2) \vdash \overline{\tev[Q_2]'} \ty [\overline{b_f' \mapsto \tau_{b_f}'}]Q_2 \rightarrow \notag \\
        &\quad\theta_2' = [\overline{b \mapsto \tau_b'}, \overline{b_f' \mapsto \tau_{b_f}'}] \rightarrow \notag \\
        &\quad{}(c, \Gamma_\upsilon^2) \vdash \overline{\tev[C_2]' \ty \theta_2'{}(C_2)} \rightarrow \notag \\
        &\quad\eta_2' = [\overline{d_2 \mapsto \tev[Q_2]'}] \text{ for each } (d_2 \ty C_2) \in Q_2 \rightarrow \notag \\
        &\quad\text{the principal type } \sigma' \text{ is context-unambiguous where } \dec{\epsilon}{\Gamma}{e}{\sigma'} \rightarrow \notag \\
        &\quad\taueq{\sigma'}{\theta_1{}(\tau_1)}{\theta_2'{}(\tau_2)} \rightarrow \notag \\
        &\quad\evidenceCanonicity{}((c{,} \Gamma_\upsilon^1), (c{,} \Gamma_\upsilon^2), \theta_1, \theta_2', (\overline{\tev[Q_1]} \wedge \overline{\tev[C_1]}), (\overline{\tev[Q_2]'} \wedge \overline{\tev[C_2]'}), Q_1 \wedge \overline{C_1}, \notag \\
        &\quad\hspacev{}Q_2 \wedge \overline{C_2}) \rightarrow \notag \\
        &\quad\tred{(\eta_1{}(t_1)~\overline{\upsilon_a}~\overline{\tev[C_1]})}
          {(\eta_2'{}(t_2')~\overline{\upsilon_b'}~\overline{\tev[C_2]'})} \label{eq:foralli-ih}
      \end{align}
      If we instantiate \cref{eq:foralli-ih} with the following variables:
      $\overline{b_f'} = (c, \overline{b_f})$,
      $\overline{\tau_{b_f}'} = (\tau_c, \overline{\tau_{b_f}})$,
      $\overline{\tev[Q_2]'} = \overline{\tev[Q_2]}$,
      $\overline{\tev[C_2]'} = \overline{\tev[C_2]}$,
      $\overline{\tau_b'} = \overline{\tau_b}$, $\sigma' = \sigma'$, we must
      prove the following:
      \begin{align}
        &c, \overline{b_f} = \ftv{Q_2, (c, \Gamma)} \label{eq:foralli-ih-ftv} \\
        &(c, \Gamma_\upsilon^2) \vdash \overline{\tev[Q_2]} \ty [c \mapsto \tau_c, \overline{b_f \mapsto \tau_{b_f}}]Q_2 \label{eq:foralli-ih-tevQ2} \\
        &(c, \Gamma_\upsilon^2) \vdash \overline{\tev[C_2] \ty \theta_2'{}(C_2)} \label{eq:foralli-ih-tevC2} \\
        &\text{the principal type } \sigma \text{ is context-unambiguous where } \dec{\epsilon}{\Gamma}{e}{\sigma} \label{eq:foralli-ih-ambi} \\
        &\taueq{\sigma}{\theta_1{}(\tau_1)}{\theta_2'{}(\tau_2)} \label{eq:foralli-ih-tau-eq} \\
        &\evidenceCanonicity{}((c{,} \Gamma_\upsilon^1), (c{,} \Gamma_\upsilon^2), \theta_1, \theta_2', (\overline{\tev[Q_1]} \wedge \overline{\tev[C_1]}), (\overline{\tev[Q_2]} \wedge \overline{\tev[C_2]}), Q_1 \wedge \overline{C_1}, \notag \\
        &\hspacev{}Q_2 \wedge \overline{C_2}) \label{eq:foralli-ih-ev-canon}
      \end{align}
      We also have that:
      \begin{align}
        &\theta_2' = [\overline{b \mapsto \tau_b}, c \mapsto \tau_c, \overline{b_f \mapsto \tau_{b_f}}] \label{eq:foralli-ih-theta2} \\
        &\eta_2' = [\overline{d_2 \mapsto \tev[Q_2]}] \text{ for each } (d_2 \ty C_2) \in Q_2 \label{eq:foralli-ih-eta2}
      \end{align}
      From \cref{eq:foralli-theta2} and \cref{eq:foralli-ih-theta2}
      follows that $\theta_2 = \theta_2'$. We replace each $\theta_2'$ by
      $\theta_2$. From \cref{eq:foralli-eta2} and
      \cref{eq:foralli-ih-eta2} follows that $\eta_2 = \eta_2'$. We
      replace each $\eta_2'$ by $\eta_2$.

      \cref{eq:foralli-ih-ftv} follows from \cref{eq:foralli-ftv} and
      \cref{eq:foralli-cftv}, \cref{eq:foralli-ih-tevQ2} from
      \cref{eq:foralli-tevQ2} and \cref{eq:foralli-cftv},
      \cref{eq:foralli-ih-tevC2} from \cref{eq:foralli-tevC2},
      \cref{eq:foralli-ih-ambi} from \cref{eq:foralli-ambi},
      \cref{eq:foralli-ih-tau-eq} from \cref{eq:foralli-tau-eq},
      \cref{eq:foralli-ih-ev-canon} from \cref{eq:foralli-ev-canon}.

      From the induction hypothesis we obtain that:
      \begin{align}
        &\tred{(\eta_1{}(t_1)~\overline{\upsilon_a}~\overline{\tev[C_1]})}
          {(\eta_2{}(t_2')~\overline{\upsilon_b}~\overline{\tev[C_2]})} \label{eq:foralli-ih-conclusion}
      \end{align}

      If we now apply \textsc{EraseTyAbs} and \textsc{EraseTyApp} inside the
      type and evidence applications, we get our goal
      \cref{eq:foralli-goal}:
      \begin{align}
        &\tred{(\eta_1{}(t_1)~\overline{\upsilon_a}~\overline{\tev[C_1]})}
          {((\Lambda{}c.\,\eta_2{}(t_2'))~\upsilon_c~\overline{\upsilon_b}~\overline{\tev[C_2]})} \notag
      \end{align}



  \item[\pcase{foralle}.] The $\elimforall$ rule was last used in the first
    or second derivation. By symmetry, we assume the second. We have that:
    \begin{align*}
      &\decelab{Q_1}{\Gamma}{e}{\forall\,\overline{a}.\,\overline{C_1} \Rightarrow \tau_1}{t_1} \\
      &\overline{a_f} = \ftv{Q_1, \Gamma} \\
      &\Gamma_\upsilon^1 \vdash \overline{\tev[Q_1]} \ty [\overline{a_f \mapsto \tau_{a_f}}]Q_1 \\
      &\theta_1 = [\overline{a \mapsto \tau_a}, \overline{a_f \mapsto \tau_{a_f}}] \\
      &\Gamma_\upsilon^1 \vdash \overline{\tev[C_1] \ty \theta_1{}(C_1)} \\
      &\eta_1 = [\overline{d_1 \mapsto \tev[Q_1]}] \text{ for each } (d_1 \ty C_1) \in Q_1
    \end{align*}
    From the second derivation we have that:
    \begin{align}
      &\decelab{Q_2}{\Gamma}{e}{[c \mapsto \tau_c](\forall\,\overline{b}.\,\overline{C_2} \Rightarrow \tau_2)}{t_2'~\upsilon_c} \notag
    \end{align}
    Which we can rewrite to:
    \begin{align}
      &\decelab{Q_2}{\Gamma}{e}{\forall\,\overline{b}.\,[c \mapsto \tau_c]\overline{C_2} \Rightarrow [c \mapsto \tau_c]\tau_2}{t_2'~\upsilon_c} \notag
    \end{align}
    From the second derivation we also have that:
    \begin{align}
      &\decelab{Q_2}{\Gamma}{e}{\forall{}c\overline{b}.\,\overline{C_2} \Rightarrow \tau_2}{t_2'} \notag \\
      &\elab{\tau}{\tau_c}{\upsilon_c} \notag \\
      &\overline{b_f} = \ftv{Q_2, \Gamma} \label{eq:foralle-ftv} \\
      &\Gamma_\upsilon^2 \vdash \overline{\tev[Q_2]} \ty [\overline{b_f \mapsto \tau_{b_f}}]Q_2 \label{eq:foralle-tevQ2} \\
      &\theta_2 = [\overline{b \mapsto \tau_b}, \overline{b_f \mapsto \tau_{b_f}}] \label{eq:foralle-theta2} \\
      &\Gamma_\upsilon^2 \vdash \overline{\tev[C_2] \ty \theta_2{}([c \mapsto \tau_c]C_2)} \label{eq:foralle-tevC2} \\
      &\eta_2 = [\overline{d_2 \mapsto \tev[Q_2]}] \text{ for each } (d_2 \ty C_2) \in Q_2 \label{eq:foralle-eta2} \\
      &\text{the principal type } \sigma \text{ is context-unambiguous where } \dec{\epsilon}{\Gamma}{e}{\sigma} \label{eq:foralle-ambi} \\
      &\taueq{\sigma}{\theta_1{}(\tau_1)}{\theta_2{}([c \mapsto \tau_c]\tau_2)} \label{eq:foralle-tau-eq} \\
      &\evidenceCanonicity{}(\Gamma_\upsilon^1, \Gamma_\upsilon^2, \theta_1, \theta_2, (\overline{\tev[Q_1]} \wedge \overline{\tev[C_1]}), (\overline{\tev[Q_2]} \wedge \overline{\tev[C_2]}), Q_1 \wedge \overline{C_1}, \notag \\
      &\hspacev{}Q_2 \wedge [c \mapsto \tau_c]\overline{C_2}) \label{eq:foralle-ev-canon}
    \end{align}

    And must prove that:
    \begin{align*}
      &\tred{(\eta_1{}(t_1)~\overline{\upsilon_a}~\overline{\tev[C_1]})}
        {(\eta_2{}(t_2'~\upsilon_c)~\overline{\upsilon_b}~\overline{\tev[C_2]})}
    \end{align*}
    We can rewrite this to:
    \begin{align}
      &\tred{(\eta_1{}(t_1)~\overline{\upsilon_a}~\overline{\tev[C_1]})}
        {(\eta_2{}(t_2')~\upsilon_c~\overline{\upsilon_b}~\overline{\tev[C_2]})} \label{eq:foralle-goal}
    \end{align}

    We have the following induction hypothesis. For notational convenience we
    generalise over a separate variable $\tau_c'$, which would otherwise be
    part of $\overline{\tau_b'}$.
    \begin{align}
      &\forall~\overline{b_f'}~\overline{\tau_{b_f}'}~\overline{\tev[Q_2]'}~\overline{\tev[C_2]'}~\overline{\tau_b'}~\theta_2'~\eta_2'~\tau_c'~\sigma', \notag \\
      &\quad\overline{b_f'} = \ftv{Q_2, \Gamma} \rightarrow \notag \\
      &\quad\Gamma_\upsilon^2 \vdash \overline{\tev[Q_2]'} \ty [\overline{b_f' \mapsto \tau_{b_f}'}]Q_2 \rightarrow \notag \\
      &\quad\theta_2' = [c \mapsto \tau_c', \overline{b \mapsto \tau_b'}, \overline{b_f' \mapsto \tau_{b_f}'}] \rightarrow \notag \\
      &\quad\Gamma_\upsilon^2 \vdash \overline{\tev[C_2]' \ty \theta_2'{}(C_2)} \rightarrow \notag \\
      &\quad\eta_2' = [\overline{d_2 \mapsto \tev[Q_2]'}] \text{ for each } (d_2 \ty C_2) \in Q_2 \rightarrow \notag \\
      &\quad\text{the principal type } \sigma' \text{ is context-unambiguous where } \dec{\epsilon}{\Gamma}{e}{\sigma'} \rightarrow \notag \\
      &\quad\taueq{\sigma'}{\theta_1{}(\tau_1)}{\theta_2'{}(\tau_2)} \rightarrow \notag \\
      &\quad\evidenceCanonicity{}(\Gamma_\upsilon^1, \Gamma_\upsilon^2, \theta_1, \theta_2', (\overline{\tev[Q_1]} \wedge \overline{\tev[C_1]}), (\overline{\tev[Q_2]'} \wedge \overline{\tev[C_2]'}), Q_1 \wedge \overline{C_1}, Q_2 \wedge \overline{C_2}) \rightarrow \notag \\
      &\quad\tred{(\eta_1{}(t_1)~\overline{\upsilon_a}~\overline{\tev[C_1]})}
        {(\eta_2'{}(t_2')~\upsilon_c'~\overline{\upsilon_b'}~\overline{\tev[C_2]'})} \label{eq:foralle-ih}
    \end{align}

    If we instantiate \cref{eq:foralle-ih} with the following variables:
    $\overline{b_f'} = \overline{b_f}$,
    $\overline{\tau_{b_f}'} = \overline{\tau_{b_f}}$,
    $\overline{\tev[Q_2]'} = \overline{\tev[Q_2]}$,
    $\overline{\tev[C_2]'} = \overline{\tev[C_2]}$,
    $\overline{\tau_b'} = \overline{\tau_b}$, $\tau_c' = \tau_c$,
    $\sigma' = \sigma$, we must prove the following:
    \begin{align}
      &\overline{b_f} = \ftv{Q_2, \Gamma} \label{eq:foralle-ih-ftv} \\
      &\Gamma_\upsilon^2 \vdash \overline{\tev[Q_2]} \ty [\overline{b_f \mapsto \tau_{b_f}}]Q_2 \label{eq:foralle-ih-tevQ2} \\
      &\Gamma_\upsilon^2 \vdash \overline{\tev[C_2] \ty \theta_2'{}(C_2)} \label{eq:foralle-ih-tevC2} \\
      &\text{the principal type } \sigma \text{ is context-unambiguous where } \dec{\epsilon}{\Gamma}{e}{\sigma} \label{eq:foralle-ih-ambi} \\
      &\taueq{\sigma}{\theta_1{}(\tau_1)}{\theta_2'{}(\tau_2)} \label{eq:foralle-ih-tau-eq} \\
      &\evidenceCanonicity{}(\Gamma_\upsilon^1, \Gamma_\upsilon^2, \theta_1, \theta_2', (\overline{\tev[Q_1]} \wedge \overline{\tev[C_1]}), (\overline{\tev[Q_2]} \wedge \overline{\tev[C_2]}), Q_1 \wedge \overline{C_1}, Q_2 \wedge \overline{C_2}) \label{eq:foralle-ih-ev-canon}
    \end{align}
    We also have that:
    \begin{align}
      &\theta_2' = [c \mapsto \tau_c, \overline{b \mapsto \tau_b}, \overline{b_f \mapsto \tau_{b_f}}]  \label{eq:foralle-ih-theta2} \\
      &\eta_2' = [\overline{d_2 \mapsto \tev[Q_2]}] \text{ for each } (d_2 \ty C_2) \in Q_2 \label{eq:foralle-ih-eta2}
    \end{align}
    From \cref{eq:foralle-eta2} and \cref{eq:foralle-ih-eta2}
    follows that $\eta_2 = \eta_2'$. We replace each $\eta_2'$ by $\eta_2$.

    \cref{eq:foralle-ih-ftv} follows from \cref{eq:foralle-ftv},
    \cref{eq:foralle-ih-tevQ2} from \cref{eq:foralle-tevQ2},
    \cref{eq:foralle-ih-tevC2} from \cref{eq:foralle-tevC2} and
    \cref{eq:foralle-ih-theta2}, \cref{eq:foralle-ih-tau-eq} from
    \cref{eq:foralle-tau-eq} and \cref{eq:foralle-ih-theta2},
    \cref{eq:foralle-ih-ambi} from \cref{eq:foralle-ambi},
    \cref{eq:foralle-ih-ev-canon} from \cref{eq:foralle-ev-canon}.

    From the induction hypothesis we obtain that:
    \begin{align}
      &\tred{(\eta_1{}(t_1)~\overline{\upsilon_a}~\overline{\tev[C_1]})}
        {(\eta_2{}(t_2')~\upsilon_c~\overline{\upsilon_b}~\overline{\tev[C_2]})} \label{eq:foralle-ih-conclusion}
    \end{align}
    Which corresponds to our goal \cref{eq:foralle-goal}, proving this
    case.



  \item[\pcase{cti}.] The $\introct$ rule was last used in the first or second
    derivation. By symmetry, we assume the second. We have that:
    \begin{align}
      &\decelab{Q_1}{\Gamma}{e}{\forall\,\overline{a}.\,\overline{C_1} \Rightarrow \tau_1}{t_1} \notag \\
      &\overline{a_f} = \ftv{Q_1, \Gamma} \notag \\
      &\Gamma_\upsilon^1 \vdash \overline{\tev[Q_1]} \ty [\overline{a_f \mapsto \tau_{a_f}}]Q_1 \notag \\
      &\theta_1 = [\overline{a \mapsto \tau_a}, \overline{a_f \mapsto \tau_{a_f}}] \notag \\
      &\Gamma_\upsilon^1 \vdash \overline{\tev[C_1] \ty \theta_1{}(C_1)} \notag \\
      &\eta_1 = [\overline{d_1 \mapsto \tev[Q_1]}] \text{ for each } (d_1 \ty C_1) \in Q_1 \\
      &\notag \\
      &\decelab{Q_2}{\Gamma}{e}{C \Rightarrow \overline{C_2} \Rightarrow \tau_2}{\lambda{}(d\ty{}\upsilon_C).\,t_2'} \label{eq:cti-hty2} \\
      &\decelab{d \ty C \wedge Q_2}{\Gamma}{e}{\overline{C_2} \Rightarrow \tau_2}{t_2'} \notag \\
      &\elab{C}{C}{\upsilon_C} \notag \\
      &d \notin \fdv{Q_1} \label{eq:cti-d-fresh} \\
      &\overline{b_f} = \ftv{Q_2, \Gamma} \label{eq:cti-ftv} \\
      &\Gamma_\upsilon^2 \vdash \overline{\tev[Q_2]} \ty [\overline{b_f \mapsto \tau_{b_f}}]Q_2 \label{eq:cti-tevQ2} \\
      &\theta_2 = [\overline{b \mapsto \tau_b}, \overline{b_f \mapsto \tau_{b_f}}] \label{eq:cti-theta2} \\
      &\Gamma_\upsilon^2 \vdash \tev[C] \ty \theta_2{}(C) \wedge \overline{\tev[C_2] \ty \theta_2{}(C_2)} \label{eq:cti-tevC2} \\
      &\eta_2 = [\overline{d_2 \mapsto \tev[Q_2]}] \text{ for each } (d_2 \ty C_2) \in Q_2 \label{eq:cti-eta2} \\
      &\text{the principal type } \sigma \text{ is context-unambiguous where } \dec{\epsilon}{\Gamma}{e}{\sigma} \label{eq:cti-ambi} \\
      &\taueq{\sigma}{\theta_1{}(\tau_1)}{\theta_2{}(\tau_2)} \label{eq:cti-tau-eq} \\
      &\evidenceCanonicity{}(\Gamma_\upsilon^1, \Gamma_\upsilon^2, \theta_1, \theta_2, (\overline{\tev[Q_1]} \wedge \overline{\tev[C_1]}), (\overline{\tev[Q_2]} \wedge \tev[C] \wedge \overline{\tev[C_2]}), Q_1 \wedge \overline{C_1}, \notag \\
      &\hspacev{}Q_2 \wedge C \wedge \overline{C_2}) \label{eq:cti-ev-canon}
    \end{align}
    And must prove that:
    \begin{align*}
      &\tred{(\eta_1{}(t_1)~\overline{\upsilon_a}~\overline{\tev[C_1]})}
        {(\eta_2{}(\lambda{}(d\ty{}\upsilon_C).\,t_2')~\overline{\upsilon_b}~\tev[C]~\overline{\tev[C_2]})}
    \end{align*}

    Since we know from \cref{eq:cti-hty2} that:
    \begin{align}
      &\overline{b} \text{ is empty} \label{eq:cti-b-empty}
    \end{align}
    Consequently, it must also be that:
    \begin{align}
      &\overline{\tau_b} \text{ is empty} \label{eq:cti-tau-b-empty}
    \end{align}
    Using \cref{eq:cti-b-empty} and \cref{eq:cti-tau-b-empty} we can
    rewrite \cref{eq:cti-theta2}:
    \begin{align}
      &\theta_2 = [\overline{b_f \mapsto \tau_{b_f}}] \label{eq:cti-theta2'}
    \end{align}

    We must prove that:
    \begin{align}
      &\tred{(\eta_1{}(t_1)~\overline{\upsilon_a}~\overline{\tev[C_1]})}
        {(\eta_2{}(\lambda{}(d\ty{}\upsilon_C).\,t_2')~\tev[C]~\overline{\tev[C_2]})} \label{eq:cti-goal}
    \end{align}

    We have the following induction hypothesis, which has been simplified
    using \cref{eq:cti-b-empty} and \cref{eq:cti-tau-b-empty}.
    Consequently, we omit the $\overline{\tau_b'}$ variable as it can only be empty.
    For notational convenience we generalise over a separate variable
    $\tev[Q_2]'$, which would otherwise be part of $\overline{\tev[Q_2]'}$.
    \begin{align}
      &\forall~\overline{b_f'}~\overline{\tau_{b_f}'}~\overline{\tev[Q_2]'}~\overline{\tev[C_2]'}~\theta_2'~\eta_2'~\tev[Q_2]'~\sigma', \notag \\
      &\quad\overline{b_f'} = \ftv{d \ty C \wedge Q_2, \Gamma} \rightarrow \notag \\
      &\quad(d \ty \upsilon_C \wedge \Gamma_\upsilon^2) \vdash \tev[Q_2]' \ty [\overline{b_f' \mapsto \tau_{b_f}'}]C \wedge \overline{\tev[Q_2]'} \ty [\overline{b_f' \mapsto \tau_{b_f}'}]Q_2 \rightarrow \notag \\
      &\quad\theta_2' = [\overline{b_f' \mapsto \tau_{b_f}'}] \rightarrow \notag \\
      &\quad(d \ty \upsilon_C \wedge \Gamma_\upsilon^2) \vdash \overline{\tev[C_2]' \ty \theta_2'{}(C_2)} \rightarrow \notag \\
      &\quad\eta_2' = [\overline{d_2 \mapsto \tev[Q_2]'}] \text{ for each } (d_2 \ty C_2) \in (d \ty C \wedge Q_2) \rightarrow \notag \\
      &\quad\text{the principal type } \sigma' \text{ is context-unambiguous where } \dec{\epsilon}{\Gamma}{e}{\sigma'} \rightarrow \notag \\
      &\quad\taueq{\sigma'}{\theta_1{}(\tau_1)}{\theta_2'{}(\tau_2)} \rightarrow \notag \\
      &\quad\evidenceCanonicity{}(\Gamma_\upsilon^1, (d \ty \upsilon_C \wedge \Gamma_\upsilon^2), \theta_1, \theta_2', (\overline{\tev[Q_1]} \wedge \overline{\tev[C_1]}), (\tev[Q_2]' \wedge \overline{\tev[Q_2]'} \wedge \overline{\tev[C_2]'}), \notag \\
      &\quad\hspacev{}Q_1 \wedge \overline{C_1}, C \wedge Q_2 \wedge \overline{C_2}) \rightarrow \notag \\
      &\quad\tred{(\eta_1{}(t_1)~\overline{\upsilon_a}~\overline{\tev[C_1]})}{(\eta_2'{}(t_2')~\overline{\tev[C_2]'})} \label{eq:cti-ih}
    \end{align}
    If we instantiate \cref{eq:cti-ih} with the following variables:
    $\overline{b_f'} = \overline{b_f}$,
    $\overline{\tau_{b_f}'} = \overline{\tau_{b_f}}$,
    $\overline{\tev[Q_2]'} = \overline{\tev[Q_2]}$,
    $\overline{\tev[C_2]'} = \overline{\tev[C_2]}$, $\tev[Q_2]' = \tev[C]$,
    $\sigma' = \sigma$, we must prove the following:
    \begin{align}
      &\overline{b_f} = \ftv{d \ty C \wedge Q_2, \Gamma} \label{eq:cti-ih-ftv} \\
      &(d \ty \upsilon_C \wedge \Gamma_\upsilon^2) \vdash \tev[C] \ty [\overline{b_f \mapsto \tau_{b_f}}]C \wedge \overline{\tev[Q_2]} \ty [\overline{b_f \mapsto \tau_{b_f}}]Q_2 \label{eq:cti-ih-tevQ2} \\
      &(d \ty \upsilon_C \wedge \Gamma_\upsilon^2) \vdash \overline{\tev[C_2] \ty \theta_2'{}(C_2)} \label{eq:cti-ih-tevC2} \\
      &\text{the principal type } \sigma \text{ is context-unambiguous where } \dec{\epsilon}{\Gamma}{e}{\sigma} \label{eq:cti-ih-ambi} \\
      &\taueq{\sigma}{\theta_1{}(\tau_1)}{\theta_2'{}(\tau_2)} \label{eq:cti-ih-tau-eq} \\
      &\evidenceCanonicity{}(\Gamma_\upsilon^1, (d \ty \upsilon_C \wedge \Gamma_\upsilon^2), \theta_1, \theta_2', (\overline{\tev[Q_1]} \wedge \overline{\tev[C_1]}), (\tev[C] \wedge \overline{\tev[Q_2]} \wedge \overline{\tev[C_2]}), \notag \\
      &\hspacev{}Q_1 \wedge \overline{C_1}, C \wedge Q_2 \wedge \overline{C_2}) \label{eq:cti-ih-ev-canon}
    \end{align}
    We also have that:
    \begin{align}
      &\theta_2' = [\overline{b_f \mapsto \tau_{b_f}}] \label{eq:cti-ih-theta2} \\
      &\eta_2' = [\overline{d_2 \mapsto \tev[Q_2]}] \text{ for each } (d_2 \ty C_2) \in (d \ty C \wedge Q_2) \label{eq:cti-ih-eta2}
    \end{align}
    From \cref{eq:cti-theta2'} and \cref{eq:cti-ih-theta2} follows
    that $\theta_2 = \theta_2'$. We replace each $\theta_2'$ by $\theta_2$.

    \cref{eq:cti-ih-ftv} follows from \cref{eq:cti-ftv} and the fact that $C$
    only contains type variables that are in $\Gamma$ \cref{eq:cti-hty2}.
    \cref{eq:cti-ih-tevQ2} follows from
    \cref{lemma:systemf-context-weakening}, \cref{eq:cti-tevQ2}, and
    \cref{eq:cti-tevC2} combined with \cref{eq:cti-theta2'},
    \cref{eq:cti-ih-tevC2} from \cref{lemma:systemf-context-weakening} and
    \cref{eq:cti-tevC2}, \cref{eq:cti-ih-ambi} from \cref{eq:cti-ambi},
    \cref{eq:cti-ih-tau-eq} from \cref{eq:cti-tau-eq}.
    \cref{eq:cti-ih-ev-canon} follows from
    \cref{lemma:canonicity-context-weakening} and \cref{eq:cti-ev-canon}, as
    the only difference with \cref{eq:cti-ev-canon} is the extended context,
    besides the reordering of the constraints and evidence.

    From the induction hypothesis we obtain that:
    \begin{align}
      &\tred{(\eta_1{}(t_1)~\overline{\upsilon_a}~\overline{\tev[C_1]})}
        {(\eta_2'{}(t_2')~\overline{\tev[C_2]})} \label{eq:cti-ih-conclusion}
    \end{align}
    Because of \cref{eq:cti-d-fresh} we know that $d \notin \fdv{\eta_2}$
    \cref{eq:cti-eta2}. This means we can rewrite $\eta_2'$
    \cref{eq:cti-ih-eta2} to:
    \begin{align*}
      &\eta_2' = [d \mapsto \tev[C], \eta_2]
    \end{align*}
    Using this we can rewrite \cref{eq:cti-ih-conclusion} to:
    \begin{align*}
      &\tred{(\eta_1{}(t_1)~\overline{\upsilon_a}~\overline{\tev[C_1]})}
        {([d \mapsto \tev[C], \eta_2](t_2')~\overline{\tev[C_2]})}
    \end{align*}
    If we now apply \betared{} inside the evidence applications of $\overline{\tev[C_2]}$, we get:
    \begin{align*}
      &\tred{(\eta_1{}(t_1)~\overline{\upsilon_a}~\overline{\tev[C_1]})}
        {((\lambda{}(d \ty \upsilon_C).\,\eta_2(t_2'))~\tev[C]~\overline{\tev[C_2]})}
    \end{align*}
    Because $d \notin \dom{\eta_2}$, we can rewrite this to our goal
    \cref{eq:cti-goal}:
    \begin{align*}
      &\tred{(\eta_1{}(t_1)~\overline{\upsilon_a}~\overline{\tev[C_1]})}
        {(\eta_2{}(\lambda{}(d\ty{}\upsilon_C).\,t_2')~\tev[C]~\overline{\tev[C_2]})}
    \end{align*}



  \item[\pcase{cte}.] The $\elimct$ rule was last used in the first or second
    derivation. By symmetry, we assume the second. We have that:
    \begin{align}
      &\decelab{Q_1}{\Gamma}{e}{\forall\,\overline{a}.\,\overline{C_1} \Rightarrow \tau_1}{t_1} \notag \\
      &\overline{a_f} = \ftv{Q_1, \Gamma} \notag \\
      &\Gamma_\upsilon^1 \vdash \overline{\tev[Q_1]} \ty [\overline{a_f \mapsto \tau_{a_f}}]Q_1 \notag \\
      &\theta_1 = [\overline{a \mapsto \tau_a}, \overline{a_f \mapsto \tau_{a_f}}] \notag \\
      &\Gamma_\upsilon^1 \vdash \overline{\tev[C_1] \ty \theta_1{}(C_1)} \notag \\
      &\eta_1 = [\overline{d_1 \mapsto \tev[Q_1]}] \text{ for each } (d_1 \ty C_1) \in Q_1 \\
      &\notag \\
      &\decelab{Q_2}{\Gamma}{e}{\overline{C_2} \Rightarrow \tau_2}{t_2'~\tev} \label{eq:cte-hty2} \\
      &\decelab{Q_2}{\Gamma}{e}{C \Rightarrow \overline{C_2} \Rightarrow \tau_2}{t_2'} \notag \\
      &\entails{\QQ \wedge Q_2}{\tev \ty C} \label{eq:cte-tev} \\
      &\overline{b_f} = \ftv{Q_2, \Gamma} \label{eq:cte-ftv} \\
      &\Gamma_\upsilon^2 \vdash \overline{\tev[Q_2]} \ty [\overline{b_f \mapsto \tau_{b_f}}]Q_2 \label{eq:cte-tevQ2} \\
      &\theta_2 = [\overline{b \mapsto \tau_b}, \overline{b_f \mapsto \tau_{b_f}}] \label{eq:cte-theta2} \\
      &\Gamma_\upsilon^2 \vdash \overline{\tev[C_2] \ty \theta_2{}(C_2)} \label{eq:cte-tevC2} \\
      &\eta_2 = [\overline{d_2 \mapsto \tev[Q_2]}] \text{ for each } (d_2 \ty C_2) \in Q_2 \label{eq:cte-eta2} \\
      &\text{the principal type } \sigma \text{ is context-unambiguous where } \dec{\epsilon}{\Gamma}{e}{\sigma} \label{eq:cte-ambi} \\
      &\taueq{\sigma}{\theta_1{}(\tau_1)}{\theta_2{}(\tau_2)} \label{eq:cte-tau-eq} \\
      &\evidenceCanonicity{}(\Gamma_\upsilon^1, \Gamma_\upsilon^2, \theta_1, \theta_2, (\overline{\tev[Q_1]} \wedge \overline{\tev[C_1]}), (\overline{\tev[Q_2]} \wedge \overline{\tev[C_2]}), Q_1 \wedge \overline{C_1}, Q_2 \wedge \overline{C_2}) \label{eq:cte-ev-canon}
    \end{align}
    Since we know from \cref{eq:cte-hty2} that:
    \begin{align}
      &\overline{b} \text{ is empty} \label{eq:cte-b-empty}
    \end{align}
    Consequently, it must also be that:
    \begin{align}
      &\overline{\tau_b} \text{ is empty} \label{eq:cte-tau-b-empty}
    \end{align}
    Using \cref{eq:cte-b-empty} and \cref{eq:cte-tau-b-empty} we can rewrite
    \cref{eq:cte-theta2} to:
    \begin{align}
      &\theta_2 = [\overline{b_f \mapsto \tau_{b_f}}] \label{eq:cte-theta2'}
    \end{align}

    And must prove that:
    \begin{align}
      &\tred{(\eta_1{}(t_1)~\overline{\upsilon_a}~\overline{\tev[C_1]})}
        {(\eta_2{}(t_2'~\tev)~\overline{\tev[C_2]})} \label{eq:cte-goal}
    \end{align}

    We have the following induction hypothesis, which has been simplified
    using \cref{eq:cte-b-empty} and \cref{eq:cte-tau-b-empty}. Consequently,
    we omit the $\overline{\tau_b'}$ variable as it can only be empty. For notational
    convenience we generalise over a separate variable $\tev[C_2]'$, which
    would otherwise be part of $\overline{\tev[C_2]'}$.
    \begin{align}
      &\forall~\overline{b_f'}~\overline{\tau_{b_f}'}~\overline{\tev[Q_2]'}~\overline{\tev[C_2]'}~\theta_2'~\eta_2'~\tev[C_2]'~\sigma', \notag \\
      &\quad\overline{b_f'} = \ftv{Q_2, \Gamma} \rightarrow \notag \\
      &\quad\Gamma_\upsilon^2 \vdash \overline{\tev[Q_2]'} \ty [\overline{b_f' \mapsto \tau_{b_f}'}]Q_2 \rightarrow \notag \\
      &\quad\theta_2' = [\overline{b_f' \mapsto \tau_{b_f}'}] \rightarrow \notag \\
      &\quad\Gamma_\upsilon^2 \vdash \tev[C_2]' \ty \theta_2'{}(C) \wedge \overline{\tev[C_2]' \ty \theta_2'{}(C_2)} \rightarrow \notag \\
      &\quad\eta_2' = [\overline{d_2 \mapsto \tev[Q_2]'}] \text{ for each } (d_2 \ty C_2) \in Q_2 \rightarrow \notag \\
      &\quad\text{the principal type } \sigma' \text{ is context-unambiguous where } \dec{\epsilon}{\Gamma}{e}{\sigma'} \rightarrow \notag \\
      &\quad\taueq{\sigma'}{\theta_1{}(\tau_1)}{\theta_2'{}(\tau_2)} \rightarrow \notag \\
      &\quad\evidenceCanonicity{}(\Gamma_\upsilon^1, \Gamma_\upsilon^2, \theta_1, \theta_2', (\overline{\tev[Q_1]} \wedge \overline{\tev[C_1]}), (\overline{\tev[Q_2]'} \wedge \tev[C_2]' \wedge \overline{\tev[C_2]'}), Q_1 \wedge \overline{C_1}, \notag \\
      &\quad\hspacev{}Q_2 \wedge C \wedge \overline{C_2}) \rightarrow \notag \\
      &\quad\tred{(\eta_1{}(t_1)~\overline{\upsilon_a}~\overline{\tev[C_1]})}{(\eta_2'{}(t_2')~\tev[C_2]'~\overline{\tev[C_2]'})} \label{eq:cte-ih}
    \end{align}
    If we instantiate \cref{eq:cte-ih} with the following variables:
    $\overline{b_f'} = \overline{b_f}$,
    $\overline{\tau_{b_f}'} = \overline{\tau_{b_f}}$,
    $\overline{\tev[Q_2]'} = \overline{\tev[Q_2]}$,
    $\overline{\tev[C_2]'} = \overline{\tev[C_2]}$,
    $\tev[C_2]' = \theta_2'{}(\eta_2'{}(\tev))$, $\sigma' = \sigma$, we must
    prove the following:
    \begin{align}
      &\overline{b_f} = \ftv{Q_2, \Gamma} \label{eq:cte-ih-ftv} \\
      &\Gamma_\upsilon^2 \vdash \overline{\tev[Q_2]} \ty [\overline{b_f \mapsto \tau_{b_f}}]Q_2 \label{eq:cte-ih-tevQ2} \\
      &\Gamma_\upsilon^2 \vdash \theta_2'{}(\eta_2'{}(\tev)) \ty \theta_2'{}(C) \wedge \overline{\tev[C_2] \ty \theta_2'{}(C_2)} \label{eq:cte-ih-tevC2} \\
      &\text{the principal type } \sigma \text{ is context-unambiguous where } \dec{\epsilon}{\Gamma}{e}{\sigma} \label{eq:cte-ih-ambi} \\
      &\taueq{\sigma}{\theta_1{}(\tau_1)}{\theta_2'{}(\tau_2)} \label{eq:cte-ih-tau-eq} \\
      &\evidenceCanonicity{}(\Gamma_\upsilon^1, \Gamma_\upsilon^2, \theta_1, \theta_2', (\overline{\tev[Q_1]} \wedge \overline{\tev[C_1]}), (\overline{\tev[Q_2]} \wedge \theta_2'{}(\eta_2'{}(\tev)) \wedge \overline{\tev[C_2]}), \notag \\
      &\hspacev{}Q_1 \wedge \overline{C_1}, Q_2 \wedge C \wedge \overline{C_2}) \label{eq:cte-ih-ev-canon}
    \end{align}
    We also have that:
    \begin{align}
      &\theta_2' = [\overline{b_f \mapsto \tau_{b_f}}] \label{eq:cte-ih-theta2} \\
      &\eta_2' = [\overline{d_2 \mapsto \tev[Q_2]}] \text{ for each } (d_2 \ty C_2) \in Q_2 \label{eq:cte-ih-eta2}
    \end{align}
    From \cref{eq:cte-theta2'} and \cref{eq:cte-ih-theta2} follows that
    $\theta_2 = \theta_2'$. We replace each $\theta_2'$ by $\theta_2$. From
    \cref{eq:cte-eta2} and \cref{eq:cte-ih-eta2} follows that
    $\eta_2 = \eta_2'$. We replace each $\eta_2'$ by $\eta_2$.

    \cref{eq:cte-ih-ftv} follows from \cref{eq:cte-ftv},
    \cref{eq:cte-ih-tevQ2} from \cref{eq:cte-tevQ2}, \cref{eq:cte-ih-ambi}
    from \cref{eq:cte-ambi}, and \cref{eq:cte-ih-tau-eq} from
    \cref{eq:cte-tau-eq}.

    We now prove \cref{eq:cte-ih-ev-canon} given that we have \cref{eq:cte-ev-canon}.
    The difference between the two is that now we have the additional constraint $C$ with evidence $\theta_2{}(\eta_2{}(\tev))$ in the $\overline{C_2}$ and $\overline{\tev[2]}$.
    This constraint $C$ is derived from $\QQ \wedge Q_2$, see \cref{eq:cte-tev}.

    By \cref{lem:canonical-ev-form} with \cref{eq:cte-ev-canon}, we get that
    there exists a $\tev[2]'$ and $C_2'$ such that (using the existentially quantified objects in the evidenceCanonicity definition)
    \begin{itemize}[leftmargin=2em]
    \item $\tred{[\overline{d_j \mapsto
          \tev[2,j]^p}][d\mapsto \tev[2]'']\tev[2]'}{\eta_2{}(\tev)}$
    \item $\tev[2]'$ and $C_2'$ are uniquely determined by $C_2$ and the $\overline{C_{j}^p}$
    \item $\entails{\QQ \wedge \overline{C_{2,0}}}{\tev[2]'' \ty C_2'}$
    \end{itemize}
    This is precisely what we need for proving evidenceCanonicity for the constraints extended with $\eta_2{}(\tev) \ty C$.

    To prove \cref{eq:cte-ih-tevC2}, we use \cref{eq:cte-tevC2} to prove
    $\Gamma_\upsilon^2 \vdash \overline{\tev[C_2] \ty \theta_2{}(C_2)}$. That
    leaves us with the following left to prove:
    \begin{align}
      &\Gamma_\upsilon^2 \vdash \theta_2{}(\eta_2{}(\tev)) \ty \theta_2{}(C) \label{cref:cte-subgoal}
    \end{align}
    If we apply \Rthree{} (applying a type substitution to both sides of an
    entailment) to \cref{eq:cte-tev} with $\theta_2$, we get:
    \begin{align}
      &\entails{\theta_2{}(\QQ) \wedge \theta_2{}(Q_2)}{\theta_2{}(\tev) \ty \theta_2{}(C)} \label{cref:cte-subgoal-g1}
    \end{align}
    We can use an alternative notation of \cref{eq:cte-tevQ2} (also see
    \cref{eq:cte-eta2}):
    \begin{align*}
      &\sftc{\Gamma_\upsilon^2}{\eta_2}{[\overline{b_f \mapsto \tau_{b_f}}]Q_2}
    \end{align*}
    As $\theta_2 = [\overline{b_f \mapsto \tau_{b_f}}]$, we can rewrite this to:
    \begin{align}
      &\sftc{\Gamma_\upsilon^2}{\eta_2}{\theta_2{}(Q_2)} \label{cref:cte-subgoal-g2}
    \end{align}
    If we then apply \cref{lemma:subst-typing} using
    \cref{cref:cte-subgoal-g1} ($\theta_2{}(\QQ)$ is elaborated in $\Gamma_\upsilon^2$) and \cref{cref:cte-subgoal-g2}, we get:
    \begin{align*}
      &\Gamma_\upsilon^2 \vdash \eta_2{}(\theta_2{}(\tev)) \ty \theta_2{}(C)
    \end{align*}
    Which we can rewrite to our subgoal \cref{cref:cte-subgoal}, thereby proving it:
    \begin{align*}
      &\Gamma_\upsilon^2 \vdash \theta_2{}(\eta_2{}(\tev)) \ty \theta_2{}(C)
    \end{align*}

    From the induction hypothesis we obtain that:
    \begin{align}
      &\tred{(\eta_1{}(t_1)~\overline{\upsilon_a}~\overline{\tev[C_1]})}
        {(\eta_2{}(t_2')~\theta_2{}(\eta_2{}(\tev))~\overline{\tev[C_2]})} \label{eq:cte-ih-conclusion}
    \end{align}

    To prove our goal \cref{eq:cte-goal}, we apply \textsc{Trans} using
    \cref{eq:cte-ih-conclusion} so that we must only prove the following:
    \begin{align*}
      \tred{(\eta_2{}(t_2')~\theta_2{}(\eta_2{}(\tev))~\overline{\tev[C_2]})}
      {(\eta_2{}(t_2'~\tev)~\overline{\tev[C_2]})}
    \end{align*}
    By using the \textsc{App} and \textsc{Reflexivity} rules, we can remove
    the $\overline{\tev[C_2]}$ applications on both sides:
    \begin{align*}
      \tred{(\eta_2{}(t_2')~\theta_2{}(\eta_2{}(\tev)))}{(\eta_2{}(t_2'~\tev))}
    \end{align*}
    Because of \cref{lemma:erase-theta}, we can ignore the type substitution
    $\theta_2$ as the equivalence relation of terms is after type erasure:
    \begin{align*}
      \tred{(\eta_2{}(t_2')~\eta_2{}(\tev))}{(\eta_2{}(t_2'~\tev))}
    \end{align*}
    We can rewrite this to:
    \begin{align*}
      \tred{(\eta_2{}(t_2'~\tev))}{(\eta_2{}(t_2'~\tev))}
    \end{align*}
    We can use the \textsc{Reflexivity} to prove this, thereby finishing this case.



  \item[\pcase{dictapp-dictapp}.] The \textsc{DictApp} rule was last used in both derivations.
    We have that:
    \begin{align}
      &Q_1\,;\,\Gamma\vdash{}e_1~\dictapp{e_2}{\TC~a}\ty{}\forall{}\overline{a_1}\overline{a_2}.\,[a \mapsto \tau_2']((\overline{C_{11}}, \overline{C_{12}}) \Rightarrow \tau_1') \notag \\
      &\quad\new{\rightsquigarrow{}\Lambda{}\overline{a_1}\overline{a_2}.\,\lambda{}(\overline{d_{11} \ty \upsilon_{C_{11}}}).\,t_1'~\overline{a_1}~\upsilon_2'~\overline{a_2}~\overline{d_{11}}~t_2'} \label{eq:dictapp-hty1} \\
      &Q_2\,;\,\Gamma\vdash{}e_1~\dictapp{e_2}{\TC~a}\ty{}\forall{}\overline{b_1}\overline{b_2}.\,[a \mapsto \tau_2'']((\overline{C_{21}}, \overline{C_{22}}) \Rightarrow \tau_1'') \notag \\
      &\quad\new{\rightsquigarrow{}\Lambda{}\overline{b_1}\overline{b_2}.\,\lambda{}(\overline{d_{21} \ty \upsilon_{C_{21}}}).\,t_1''~\overline{b_1}~\upsilon_2''~\overline{b_2}~\overline{d_{21}}~t_2''} \label{eq:dictapp-hty2}
    \end{align}
    We can rewrite \cref{eq:dictapp-hty1} and \cref{eq:dictapp-hty2} to:
    \begin{align}
      &Q_1\,;\,\Gamma\vdash{}e_1~\dictapp{e_2}{\TC~a}\ty{}\forall{}\overline{a_1}\overline{a_2}.\,(\overline{[a \mapsto \tau_2']C_{11}}, \overline{[a \mapsto \tau_2']C_{12}}) \Rightarrow [a \mapsto \tau_2']\tau_1' \notag \\
      &\quad\new{\rightsquigarrow{}\Lambda{}\overline{a_1}\overline{a_2}.\,\lambda{}(\overline{d_{11} \ty \upsilon_{C_{11}}}).\,t_1'~\overline{a_1}~\upsilon_2'~\overline{a_2}~\overline{d_{11}}~t_2'} \label{eq:dictapp-hty1'} \\
      &Q_2\,;\,\Gamma\vdash{}e_1~\dictapp{e_2}{\TC~a}\ty{}\forall{}\overline{b_1}\overline{b_2}.\,(\overline{[a \mapsto \tau_2'']C_{21}}, \overline{[a \mapsto \tau_2'']C_{22}}) \Rightarrow [a \mapsto \tau_2'']\tau_1'' \notag \\
      &\quad\new{\rightsquigarrow{}\Lambda{}\overline{b_1}\overline{b_2}.\,\lambda{}(\overline{d_{21} \ty \upsilon_{C_{21}}}).\,t_1''~\overline{b_1}~\upsilon_2''~\overline{b_2}~\overline{d_{21}}~t_2''} \label{eq:dictapp-hty2'}
    \end{align}
    From induction we also have that:
    \begin{align}
      &\decelab{Q_1}{\Gamma}{e_1}{\forall\overline{a_1}a\overline{a_2}.\,(\overline{C_{11}}, \TC~a, \overline{C_{12}}) \Rightarrow \tau_1'}{t_1'} \notag \\
      &\text{The type of } e_1 \text{ is specified to be } \forall\overline{a_1}a\overline{a_2}.\,(\overline{C_{11}}, \TC~a, \overline{C_{12}}) \Rightarrow \tau_1' \notag \\
      &\text{The principal type of } e_1 \text{ is unambiguous} \label{eq:dictapp-hty-unambi11} \\
      &\decelab{Q_1}{\Gamma}{e_2}{\dict{\TC}~\tau_2'}{t_2'} \label{eq:dictapp-hty12} \\
      &\elab{\tau}{\tau_2'}{\upsilon_2'} \notag \\
      &\elab{C}{[a \mapsto \tau_2']\overline{C_{11}}}{\overline{\upsilon_{C_{11}}}} \notag \\
      &\forall C \ldotp \entails{\QQ \wedge \TC~a}{\tev \ty C} \Rightarrow ((\nentails{\QQ \wedge Q_1 \wedge \overline{C_{11}} \wedge \overline{C_{12}}}{C}) \vee (\entails{\QQ}{\tev \ty C})) \label{eq:dictapp-nentails1} \\
      &\overline{a_f} = \ftv{Q_1, \Gamma} \label{eq:dictapp-ftv1} \\
      &\Gamma_\upsilon^1 \vdash \overline{\tev[Q_1]} \ty [\overline{a_f \mapsto \tau_{a_f}}]Q_1 \label{eq:dictapp-tevQ1} \\
      &\theta_1 = [\overline{a_1 \mapsto \tau_{a_1}}, \overline{a_2 \mapsto \tau_{a_2}}, \overline{a_f \mapsto \tau_{a_f}}] \label{eq:dictapp-theta1} \\
      &\Gamma_\upsilon^1 \vdash \overline{\tev[C_{11}] \ty \theta_1{}([a \mapsto \tau_2']C_{11})} \wedge \overline{\tev[C_{12}] \ty \theta_1{}([a \mapsto \tau_2']C_{12})} \label{eq:dictapp-tevC1} \\
      &\eta_1 = [\overline{d_1 \mapsto \tev[Q_1]}] \text{ for each } (d_1 \ty C_1) \in Q_1 \label{eq:dictapp-eta1} \\
      &\decelab{Q_2}{\Gamma}{e_1}{\forall\overline{b_1}a\overline{b_2}.\,(\overline{C_{21}}, \TC~a, \overline{C_{22}}) \Rightarrow \tau_1''}{t_1''} \notag \\
      &\text{The type of } e_2 \text{ is specified to be } \forall\overline{b_1}a\overline{b_2}.\,(\overline{C_{21}}, \TC~a, \overline{C_{22}}) \Rightarrow \tau_1'' \notag \\
      &\text{The principal type of } e_1 \text{ is unambiguous} \label{eq:dictapp-hty-unambi21} \\
      &\decelab{Q_2}{\Gamma}{e_2}{\dict{\TC}~\tau_2''}{t_2''} \label{eq:dictapp-hty22} \\
      &\elab{\tau}{\tau_2''}{\upsilon_2''} \notag \\
      &\elab{C}{[a \mapsto \tau_2'']\overline{C_{21}}}{\overline{\upsilon_{C_{21}}}} \notag \\
      &\forall C \ldotp \entails{\QQ \wedge \TC~a}{\tev \ty C} \Rightarrow ((\nentails{\QQ \wedge Q_2 \wedge \overline{C_{21}} \wedge \overline{C_{22}}}{C}) \vee (\entails{\QQ}{\tev \ty C})) \label{eq:dictapp-nentails2} \\
      &\overline{b_f} = \ftv{Q_2, \Gamma} \label{eq:dictapp-ftv2} \\
      &\Gamma_\upsilon^2 \vdash \overline{\tev[Q_2]} \ty [\overline{b_f \mapsto \tau_{b_f}}]Q_2 \label{eq:dictapp-tevQ2} \\
      &\theta_2 = [\overline{b_1 \mapsto \tau_{b_1}}, \overline{b_2 \mapsto \tau_{b_2}}, \overline{b_f \mapsto \tau_{b_f}}] \label{eq:dictapp-theta2} \\
      &\Gamma_\upsilon^2 \vdash \overline{\tev[C_{21}] \ty \theta_2{}([a \mapsto \tau_2'']C_{21})} \wedge \overline{\tev[C_{22}] \ty \theta_2{}([a \mapsto \tau_2'']C_{22})} \label{eq:dictapp-tevC2} \\
      &\eta_2 = [\overline{d_2 \mapsto \tev[Q_2]}] \text{ for each } (d_2 \ty C_2) \in Q_2 \label{eq:dictapp-eta2} \\
      &\text{the principal type } \sigma \text{ is context-unambiguous where } \dec{\epsilon}{\Gamma}{e_1~\dictapp{e_2}{\TC~a}}{\sigma} \label{eq:dictapp-ambi} \\
      &\taueq{\sigma}{\theta_1{}([a \mapsto \tau_2']\tau_1')}{\theta_2{}([a \mapsto \tau_2'']\tau_1'')}  \label{eq:dictapp-tau-eq} \\
      &\evidenceCanonicity{}(\Gamma_\upsilon^1, \Gamma_\upsilon^2, \theta_1, \theta_2, (\overline{\tev[Q_1]} \wedge \overline{\tev[C_{11}]} \wedge \overline{\tev[C_{12}]}), (\overline{\tev[Q_2]} \wedge \overline{\tev[C_{21}]} \wedge \overline{\tev[C_{22}]}), \notag \\
      &\hspacev{}(Q_1 \wedge \overline{[a \mapsto \tau_2']C_{11}} \wedge \overline{[a \mapsto \tau_2']C_{12}}), \notag \\
      &\hspacev{}(Q_2 \wedge \overline{[a \mapsto \tau_2'']C_{21}} \wedge \overline{[a \mapsto \tau_2'']C_{22}}) \label{eq:dictapp-ev-canon}
    \end{align}

    Given \cref{eq:dictapp-tau-eq} and the fact that $a$ occurs in the
    constraints, it must be that:
    \begin{align}
      &\theta_1{}(\tau_2') = \theta_2{}(\tau_2'') \label{eq:dict-tau-2-eq}
    \end{align}
    We must prove that:
    \begin{align}
      &(\eta_1{}(\Lambda{}\overline{a_1}\overline{a_2}.\,\lambda{}(\overline{d_{11} \ty \upsilon_{C_{11}}}).\,t_1'~\overline{a_1}~\upsilon_2'~\overline{a_2}~\overline{d_{11}}~t_2')~\overline{\upsilon_{a_1}}~\overline{\upsilon_{a_2}}~\overline{\tev[C_{11}]}~\overline{\tev[C_{12}]}) \approx \notag \\
      &(\eta_2{}(\Lambda{}\overline{b_1}\overline{b_2}.\,\lambda{}(\overline{d_{21} \ty \upsilon_{C_{21}}}).\,t_1''~\overline{b_1}~\upsilon_2''~\overline{b_2}~\overline{d_{21}}~t_2'')~\overline{\upsilon_{b_1}}~\overline{\upsilon_{b_2}}~\overline{\tev[C_{21}]}~\overline{\tev[C_{22}]}) \label{eq:dictapp-goal}
    \end{align}
    Since $d_{11} \notin \dom{\eta_1}$ \cref{eq:dictapp-eta1} and
    $d_{21} \notin \dom{\eta_2}$ \cref{eq:dictapp-eta1}, we can rewrite this to:
    \begin{align}
      &((\Lambda{}\overline{a_1}\overline{a_2}.\,\lambda{}(\overline{d_{11} \ty \upsilon_{C_{11}}}).\,\eta_1{}(t_1')~\overline{a_1}~\upsilon_2'~\overline{a_2}~\overline{d_{11}}~\eta_1{}(t_2'))~\overline{\upsilon_{a_1}}~\overline{\upsilon_{a_2}}~\overline{\tev[C_{11}]}~\overline{\tev[C_{12}]}) \approx \notag \\
      &((\Lambda{}\overline{b_1}\overline{b_2}.\,\lambda{}(\overline{d_{21} \ty \upsilon_{C_{21}}}).\,\eta_2{}(t_1'')~\overline{b_1}~\upsilon_2''~\overline{b_2}~\overline{d_{21}}~\eta_2{}(t_2''))~\overline{\upsilon_{b_1}}~\overline{\upsilon_{b_2}}~\overline{\tev[C_{21}]}~\overline{\tev[C_{22}]}) \label{eq:dictapp-goal'}
    \end{align}
    We can rewrite this using the \textsc{EraseTyApp} and \textsc{EraseTyAbs} rules to:
    \begin{align}
      &((\lambda{}(\overline{d_{11} \ty \upsilon_{C_{11}}}).\,\eta_1{}(t_1')~\overline{d_{11}}~\eta_1{}(t_2'))~\overline{\tev[C_{11}]}~\overline{\tev[C_{12}]}) \approx \notag \\
      &((\lambda{}(\overline{d_{21} \ty \upsilon_{C_{21}}}).\,\eta_2{}(t_1'')~\overline{d_{21}}~\eta_2{}(t_2''))~\overline{\tev[C_{21}]}~\overline{\tev[C_{22}]}) \label{eq:dictapp-goal''}
    \end{align}
    Using the \betared{} rule we can further rewrite our goal to:
    \begin{align}
      \tred{(\eta_1{}(t_1')~\overline{\tev[C_{11}]}~\eta_1{}(t_2')~\overline{\tev[C_{12}]})}
      {(\eta_2{}(t_1'')~\overline{\tev[C_{21}]}~\eta_2{}(t_2'')~\overline{\tev[C_{22}]})} \label{eq:dictapp-goal'''}
    \end{align}
    We have the following induction hypothesis:
    \begin{align}
      &\forall~\overline{a_f'}~\overline{\tau_{a_f}'}~\overline{\tev[Q_1]'}~\theta_1'~\eta_1'~\sigma', \notag \\
      &\quad\overline{a_f'} = \ftv{Q_1, \Gamma} \rightarrow \notag \\
      &\quad\Gamma_\upsilon^1 \vdash \overline{\tev[Q_1]'} \ty [\overline{a_f' \mapsto \tau_{a_f}'}]Q_1 \rightarrow \notag \\
      &\quad\theta_1' = [\overline{a_f' \mapsto \tau_{a_f}'}] \rightarrow \notag \\
      &\quad\eta_1' = [\overline{d_1 \mapsto \tev[Q_1]'}] \text{ for each } (d_1 \ty C_1) \in Q_1 \rightarrow \notag \\
      &\quad\forall~Q_2~\overline{b}~\overline{C_2}~\tau_2~t_2, \notag \\
      &\quad\quad\decelab{Q_2}{\Gamma}{e_2}{\forall\,\overline{b}.\,\overline{C_2} \Rightarrow \tau_2}{t_2} \rightarrow \notag \\
      &\quad\quad\forall~\overline{b_f'}~\overline{\tau_{b_f}'}~\overline{\tev[Q_2]'}~\overline{\tev[C_2]'}~\overline{\tau_b'}~\theta_2'~\eta_2', \notag \\
      &\quad\quad\quad\overline{b_f'} = \ftv{Q_2, \Gamma} \rightarrow \notag \\
      &\quad\quad\quad\Gamma_\upsilon^2 \vdash \overline{\tev[Q_2]'} \ty [\overline{b_f' \mapsto \tau_{b_f}'}]Q_2 \rightarrow \notag \\
      &\quad\quad\quad\theta_2' = [\overline{b \mapsto \tau_b'}, \overline{b_f' \mapsto \tau_{b_f}'}] \rightarrow \notag \\
      &\quad\quad\quad\Gamma_\upsilon^2 \vdash \overline{\tev[C_2]' \ty \theta_2'{}(C_2)} \rightarrow \notag \\
      &\quad\quad\quad\eta_2' = [\overline{d_2 \mapsto \tev[Q_2]'}] \text{ for each } (d_2 \ty C_2) \in Q_2 \rightarrow \notag \\
      &\quad\quad\quad\text{the principal type } \sigma' \text{ is context-unambiguous where } \dec{\epsilon}{\Gamma}{e_2}{\sigma'} \rightarrow \notag \\
      &\quad\quad\quad\taueq{\sigma'}{\theta_1'{}(\dict{\TC}~\tau_2')}{\theta_2{}(\tau_2)} \rightarrow \notag \\
      &\quad\quad\quad\evidenceCanonicity{}(\Gamma_\upsilon^1, \Gamma_\upsilon^2, \theta_1', \theta_2', \overline{\tev[Q_1]'}, (\overline{\tev[Q_2]'} \wedge \overline{\tev[C_2]'}), Q_1, Q_2 \wedge \overline{C_2}) \rightarrow \notag \\
      &\quad\quad\quad\tred{(\eta_1{}'(t_2'))}
        {(\eta_2'{}(t_2)~\overline{\upsilon_b'}~\overline{\tev[C_2]'})} \label{eq:dictapp-ih2}
    \end{align}

    If we instantiate \cref{eq:dictapp-ih2} with the following variables:
    $\overline{a_f'} = \overline{a_f}$,
    $\overline{\tau_{a_f}'} = \overline{\tau_{a_f}}$,
    $\overline{\tev[Q_1]'} = \overline{\tev[Q_1]'}$, $\theta_1'$ with its
    right-hand side, $\eta_1'$ with $\eta_1$, $Q_2 = Q_2$, $\overline{b}$
    empty, $\overline{C_2}$ empty, $\tau_2 = \dict{\TC}~\tau_2''$,
    $t_2 = t_2''$, $\overline{b_f'} = \overline{b_f}$,
    $\overline{\tau_{b_f}'} = \overline{\tau_{b_f}}$,
    $\overline{\tev[Q_2]'} = \overline{\tev[Q_2]}$, $\overline{\tev[C_2]'}$
    empty, $\overline{\tau_b'}$ empty, $\theta_2'$ with its right-hand side,
    $\eta_2'$ with $\eta_2$, and $\sigma'$ with the principal type $\sigma'$
    inferred by the algorithm (see \cref{thm:principal-types}), we
    must prove the following:
    \begin{align}
      &\overline{a_f} = \ftv{Q_1, \Gamma} \label{eq:dictapp-ih2-ftv1} \\
      &\Gamma_\upsilon^1 \vdash \overline{\tev[Q_1]} \ty [\overline{a_f \mapsto \tau_{a_f}}]Q_1 \label{eq:dictapp-ih2-tevQ1} \\
      &\decelab{Q_2}{\Gamma}{e_2}{\dict{\TC}~\tau_2''}{t_2''} \label{eq:dictapp-ih2-hty2} \\
      &\overline{b_f} = \ftv{Q_2, \Gamma} \label{eq:dictapp-ih2-ftv2} \\
      &\Gamma_\upsilon^2 \vdash \overline{\tev[Q_2]} \ty [\overline{b_f \mapsto \tau_{b_f}}]Q_2 \label{eq:dictapp-ih2-tevQ2} \\
      &\text{the principal type } \sigma' \text{ is context-unambiguous where } \dec{\epsilon}{\Gamma}{e_2}{\sigma'} \label{eq:dictapp-ih2-ambi} \\
      &\taueq{\sigma'}{[\overline{a_f \mapsto \tau_{a_f}}](\dict{\TC}~\tau_2')}{[\overline{b_f \mapsto \tau_{b_f}}](\dict{\TC}~\tau_2'')} \label{eq:dictapp-ih2-tau-eq} \\
      &\evidenceCanonicity{}(\Gamma_\upsilon^1, \Gamma_\upsilon^2, [\overline{a_f \mapsto \tau_{a_f}}], [\overline{b_f \mapsto \tau_{b_f}}], \overline{\tev[Q_1]}, \overline{\tev[Q_2]}, Q_1, Q_2) \label{eq:dictapp-ih2-ev-canon}
    \end{align}
    From the induction hypothesis we obtain that:
    \begin{align}
      &\tred{(\eta_1{}(t_2'))}{(\eta_2{}(t_2''))} \label{eq:dictapp-ih2-conclusion}
    \end{align}
    \cref{eq:dictapp-ih2-ftv1} follows from \cref{eq:dictapp-ftv1},
    \cref{eq:dictapp-ih2-tevQ1} follows from \cref{eq:dictapp-tevQ1},
    \cref{eq:dictapp-ih2-hty2} from \cref{eq:dictapp-hty22},
    \cref{eq:dictapp-ih2-ftv2} from \cref{eq:dictapp-ftv2},
    \cref{eq:dictapp-ih2-tevQ2} from \cref{eq:dictapp-tevQ2},
    \cref{eq:dictapp-ih2-ambi} from \cref{lemma:ambi-dictapp2} and
    \cref{eq:dictapp-ambi}, and \cref{eq:dictapp-ih2-ev-canon} from
    \cref{eq:dictapp-ev-canon}, as the constraints and the subsequent subgroups are a subset of those in \cref{eq:dictapp-ev-canon}.

    For \cref{eq:dictapp-ih2-tau-eq} we must prove that:
    \begin{align*}
      &\taueq{\sigma'}{[\overline{a_f \mapsto \tau_{a_f}}](\dict{\TC}~\tau_2')}{[\overline{b_f \mapsto \tau_{b_f}}](\dict{\TC}~\tau_2'')}
    \end{align*}
    We can rewrite this to the following using \cref{eq:dictapp-theta1} and
    \cref{eq:dictapp-theta2}:
    \begin{align*}
      &\taueq{\sigma'}{\theta_1(\dict{\TC}~\tau_2')}{\theta_2(\dict{\TC}~\tau_2'')}
    \end{align*}
    We can simplify this to:
    \begin{align*}
      &\taueq{\sigma'}{\theta_1(\tau_2')}{\theta_2(\tau_2'')}
    \end{align*}
    Which follows from \cref{eq:dict-tau-2-eq}.

    We also have the following induction hypothesis:
    \begin{align}
      &\forall~\overline{a_f'}~\overline{\tau_{a_f}'}~\overline{\tev[Q_1]'}~\overline{\tau'_{a_1}}~\tau'_a~\overline{\tau'_{a_2}}~\theta_1'~\overline{\tev[C_{11}]'}~\tev[\TC_{1}]'~\overline{\tev[C_{12}]'}~\eta_1'~\sigma', \notag \\
      &\quad\overline{a_f'} = \ftv{Q_1, \Gamma} \rightarrow \notag \\
      &\quad\Gamma_\upsilon^1 \vdash \overline{\tev[Q_1]'} \ty [\overline{a_f' \mapsto \tau_{a_f}'}]Q_1 \rightarrow \notag \\
      &\quad\theta_1' = [\overline{a_1 \mapsto \tau'_{a_1}}, a \mapsto \tau'_a, \overline{a_2 \mapsto \tau'_{a_2}}, \overline{a_f' \mapsto \tau_{a_f}'}] \rightarrow \notag \\
      &\quad\Gamma_\upsilon^1 \vdash \overline{\tev[C_{11}]' \ty \theta_1'{}(C_{11})} \wedge \tev[\TC_{1}]' \ty \theta_1'{}(\TC~a) \wedge \overline{\tev[C_{12}]' \ty \theta_1'{}(C_{12})} \rightarrow \notag \\
      &\quad\eta_1' = [\overline{d_1 \mapsto \tev[Q_1]'}] \text{ for each } (d_1 \ty C_1) \in Q_1 \rightarrow \notag \\
      &\quad\forall~Q_2~\overline{b}~\overline{C_2}~\tau_2~t_2, \notag \\
      &\quad\quad\decelab{Q_2}{\Gamma}{e_1}{\forall\,\overline{b}.\,\overline{C_2} \Rightarrow \tau_2}{t_2} \rightarrow \notag \\
      &\quad\quad\forall~\overline{b_f'}~\overline{\tau_{b_f}'}~\overline{\tev[Q_2]'}~\overline{\tev[C_2]'}~\overline{\tau_b'}~\theta_2'~\eta_2', \notag \\
      &\quad\quad\quad\overline{b_f'} = \ftv{Q_2, \Gamma} \rightarrow \notag \\
      &\quad\quad\quad\Gamma_\upsilon^2 \vdash \overline{\tev[Q_2]'} \ty [\overline{b_f' \mapsto \tau_{b_f}'}]Q_2 \rightarrow \notag \\
      &\quad\quad\quad\theta_2' = [\overline{b \mapsto \tau_b'}, \overline{b_f' \mapsto \tau_{b_f}'}] \rightarrow \notag \\
      &\quad\quad\quad\Gamma_\upsilon^2 \vdash \overline{\tev[C_2]' \ty \theta_2'{}(C_2)} \rightarrow \notag \\
      &\quad\quad\quad\eta_2' = [\overline{d_2 \mapsto \tev[Q_2]'}] \text{ for each } (d_2 \ty C_2) \in Q_2 \rightarrow \notag \\
      &\quad\quad\quad\text{the principal type } \sigma' \text{ is context-unambiguous where } \dec{\epsilon}{\Gamma}{e_1}{\sigma'} \rightarrow \notag  \\
      &\quad\quad\quad\taueq{\sigma'}{\theta_1'{}(\tau_1')}{\theta_2{}(\tau_2)} \rightarrow \notag \\
      &\quad\quad\quad\evidenceCanonicity{}(\Gamma_\upsilon^1, \Gamma_\upsilon^2, \theta_1', \theta_2', (\overline{\tev[Q_1]'} \wedge \overline{\tev[C_{11}]'} \wedge \tev[\TC_{1}]' \wedge \overline{\tev[C_{12}]'}), (\overline{\tev[Q_2]'} \wedge \overline{\tev[C_2]'}), \notag \\
      &\quad\quad\quad\hspacev{}Q_1, Q_2, \overline{C_2}) \rightarrow \notag \\ %
      &\quad\quad\quad\tred{(\eta_1{}'(t_2')~\overline{\upsilon_{a_1}}~\upsilon_{a}'~\overline{\upsilon_{a_2}}~\overline{\tev[C_{11}]'}~\tev[\TC_{1}]'~\overline{\tev[C_{12}]'})}
        {(\eta_2'{}(t_2)~\overline{\upsilon_b'}~\overline{\tev[C_2]'})} \label{eq:dictapp-ih1}
    \end{align}

    If we instantiate \cref{eq:dictapp-ih1} with the following variables:
    $\overline{a_f'} = \overline{a_f}$,
    $\overline{\tau_{a_f}'} = \overline{\tau_{a_f}}$,
    $\overline{\tev[Q_1]'} = \overline{\tev[Q_1]}$,
    $\overline{\tau_{a_1}'} = \overline{\tau_{a_1}}$, $\tau_{a}' = \tau_2'$,
    $\overline{\tau_{a_2}'} = \overline{\tau_{a_2}}$,
    $\theta_1' = [a \mapsto \tau_2', \theta_1]$,
    $\overline{\tev[C_{11}]'} = \overline{\tev[C_{11}]}$,
    $\tev[\TC_1]' = \eta_1{}(t_2')$,
    $\overline{\tev[C_{12}]'} = \overline{\tev[C_{12}]}$, $\eta_1'$ with
    $\eta_1$, $Q_2 = Q_2$,
    $\overline{b} = (\overline{b_1}, a, \overline{b_2})$,
    $\overline{C_2} = (\overline{C_{21}}, \TC~a, \overline{C_{22}})$,
    $\tau_2 = \tau_1''$, $t_2 = t_1''$, $\overline{b_f'} = \overline{b_f}$,
    $\overline{\tau_{b_f}'} = \overline{\tau_{b_f}}$,
    $\overline{\tev[Q_2]'} = \overline{\tev[Q_2]}$,
    $\overline{\tev[C_2]'} = (\overline{\tev[C_{21}]} \wedge \eta_2{}(t_2'')
    \wedge \overline{\tev[C_{22}]})$,
    $\overline{\tau_b'} = (\overline{\tau_{b_1}}, \tau_2'',
    \overline{\tau_{b_2}})$, $\theta_2' = [a \mapsto \tau_2'', \theta_2]$,
    $\eta_2'$ with $\eta_2$, and $\sigma'$ with the principal type $\sigma'$
    inferred by the algorithm (see \cref{thm:principal-types}), we
    must prove the following:
    \begin{align}
      &\overline{a_f} = \ftv{Q_1, \Gamma} \label{eq:dictapp-ih1-ftv1} \\
      &\Gamma_\upsilon^1 \vdash \overline{\tev[Q_1]} \ty [\overline{a_f \mapsto \tau_{a_f}}]Q_1 \label{eq:dictapp-ih1-tevQ1} \\
      &[a \mapsto \tau_2', \theta_1] = [\overline{a_1 \mapsto \tau_{a_1}}, a \mapsto \tau_2', \overline{a_2 \mapsto \tau_{a_2}}, \overline{a_f \mapsto \tau_{a_f}}] \label{eq:dictapp-ih1-theta1} \\
      &\Gamma_\upsilon^1 \vdash \overline{\tev[C_{11}] \ty [a \mapsto \tau_2', \theta_1]{}(C_{11})} \wedge \eta_1{}(t_2') \ty [a \mapsto \tau_2', \theta_1](\TC~a) \wedge \notag \\
      &\quad\quad\overline{\tev[C_{12}] \ty [a \mapsto \tau_2', \theta_1](C_{12})} \label{eq:dictapp-ih1-tevC1} \\
      &\decelab{Q_2}{\Gamma}{e_1}{\forall\,\overline{b_1}a\overline{b_2}.\,(\overline{C_{21}}, \TC~a, \overline{C_{22}}) \Rightarrow \tau_1''}{t_1''} \label{eq:dictapp-ih1-hty2} \\
      &\overline{b_f} = \ftv{Q_2, \Gamma} \label{eq:dictapp-ih1-ftv2} \\
      &\Gamma_\upsilon^2 \vdash \overline{\tev[Q_2]} \ty [\overline{b_f \mapsto \tau_{b_f}}]Q_2 \label{eq:dictapp-ih1-tevQ2} \\
      &[a \mapsto \tau_2'', \theta_2] = [\overline{b_1 \mapsto \tau_{b_1}}, a \mapsto \tau_2'', \overline{b_2 \mapsto \tau_{b_2}}, \overline{b_f \mapsto \tau_{b_f}}] \label{eq:dictapp-ih1-theta2} \\
      &\Gamma_\upsilon^2 \vdash \overline{\tev[C_{21}] \ty [a \mapsto \tau_2'', \theta_2]{}(C_{21})} \wedge \eta_2{}(t_2'') \ty [a \mapsto \tau_2'', \theta_2](\TC~a) \wedge \notag \\
      &\quad\quad\overline{\tev[C_{22}] \ty [a \mapsto \tau_2'', \theta_2](C_{22})} \label{eq:dictapp-ih1-tevC2} \\
      &\text{the principal type } \sigma' \text{ is context-unambiguous where } \dec{\epsilon}{\Gamma}{e_1}{\sigma'} \label{eq:dictapp-ih1-ambi} \\
      &\taueq{\sigma'}{[a \mapsto \tau_2', \theta_1](\tau_1')}{[a \mapsto \tau_2'', \theta_2](\tau_1'')} \label{eq:dictapp-ih1-tau-eq} \\
      &\evidenceCanonicity{}(\Gamma_\upsilon^1, \Gamma_\upsilon^2, [a \mapsto \tau_2', \theta_1], [a \mapsto \tau_2''{,} \theta_2], \notag \\
      &\hspacev{}(\overline{\tev[Q_1]} \wedge \overline{\tev[C_{11}]} \wedge \eta_1{}(t_2') \wedge \overline{\tev[C_{12}]}), \notag \\
      &\hspacev{}(\overline{\tev[Q_2]} \wedge \overline{\tev[C_{21}]} \wedge \eta_2{}(t_2'') \wedge \overline{\tev[C_{22}]}), \notag \\
      &\hspacev{}(Q_1 \wedge \overline{C_{11}} \wedge \TC~a \wedge \overline{C_{12}}), \notag \\
      &\hspacev{}(Q_2 \wedge \overline{C_{21}} \wedge \TC~a \wedge \overline{C_{22}}) \label{eq:dictapp-ih1-ev-canon}
    \end{align}
    From the induction hypothesis we obtain that:
    \begin{align}
      &\tred{(\eta_1{}(t_1')~\overline{\upsilon_{a_1}}~\upsilon_2'~\overline{\upsilon_{a_2}}~\overline{\tev[C_{11}]}~\eta_1{}(t_2')~\overline{\tev[C_{12}]})}{(\eta_2{}(t_1'')~\overline{\upsilon_{b_1}}~\upsilon_2''~\overline{\upsilon_{b_2}}~\overline{\tev[C_{21}]}~\eta_2{}(t_2'')~\overline{\tev[C_{22}]})} \label{eq:dictapp-ih1-conclusion}
    \end{align}
    \cref{eq:dictapp-ih1-ftv1} follows from \cref{eq:dictapp-ftv1},
    \cref{eq:dictapp-ih1-tevQ1} from \cref{eq:dictapp-tevQ1},
    \cref{eq:dictapp-ih1-theta1} from \cref{eq:dictapp-theta1},
    \cref{eq:dictapp-ih1-hty2} from \cref{eq:dictapp-hty22},
    \cref{eq:dictapp-ih1-ftv2} from \cref{eq:dictapp-ftv2},
    \cref{eq:dictapp-ih1-tevQ2} from \cref{eq:dictapp-tevQ2},
    \cref{eq:dictapp-ih1-theta2} from \cref{eq:dictapp-theta2},
    \cref{eq:dictapp-ih1-ambi} from \cref{eq:dictapp-hty-unambi11},
    \cref{eq:dictapp-ih1-tau-eq} from \cref{eq:dictapp-tau-eq}.

    \cref{eq:dictapp-ih1-tevC1} follows from \cref{eq:dictapp-tevC1} and
    \cref{eq:dictapp-hty12}, the presence of $\eta_1$ makes sure $t_2'$ no
    longer refers to dictionary variables from $Q_1$ so that it is well-typed
    in $\Gamma_{\upsilon}^1$. %
    Similarly, \cref{eq:dictapp-ih1-tevC2} follows from \cref{eq:dictapp-tevC2} and \cref{eq:dictapp-hty22}.

    To prove \cref{eq:dictapp-ih1-ev-canon}, observe that compared to
    \cref{eq:dictapp-ev-canon}, $[a \mapsto \tau_2']$ and
    $[a \mapsto \tau_2'']$ have simply moved places. Furthermore, the $\TC~a$
    constraint has been added with as evidence $\eta_1{}(t_2')$ and
    $\eta_2{}(t_2'')$.
    To account for this change, we introduce a new group of constraints ($i = k + 1$) with $\TC~a$ as the principal and only constraint.
    The required condition about the principal constraint then follows from \cref{lem:coherence-asmpt}.
    The required condition about all members of the group follows trivially.

    By using \textsc{EraseTyApp}, we can rewrite
    \cref{eq:dictapp-ih1-conclusion} to our goal \cref{eq:dictapp-goal'''}:
    \begin{align*}
      &\tred{(\eta_1{}(t_1')~\overline{\tev[C_{11}]}~\eta_1{}(t_2')~\overline{\tev[C_{12}]})}
        {(\eta_2{}(t_1'')~\overline{\tev[C_{21}]}~\eta_2{}(t_2'')~\overline{\tev[C_{22}]})}
    \end{align*}



  \item[\pcase{var-var}.] The \textsc{Var} rule was last used in both derivations.
    We have that:
    \begin{align}
      &\decelab{Q_1}{\Gamma}{x}{\forall\,\overline{a}.\,\overline{C_1} \Rightarrow \tau_1}{x} \notag \\
      &(x\ty{}\forall\,\overline{a}.\,\overline{C_1} \Rightarrow \tau_1) \in \Gamma \label{eq:var-in-context} \\
      &\overline{a_f} = \ftv{Q_1, \Gamma} \notag \\
      &\Gamma_\upsilon^1 \vdash \overline{\tev[Q_1]} \ty [\overline{a_f \mapsto \tau_{a_f}}]Q_1 \notag \\
      &\theta_1 = [\overline{a \mapsto \tau_a}, \overline{a_f \mapsto \tau_{a_f}}] \notag \\
      &\Gamma_\upsilon^1 \vdash \overline{\tev[C_1] \ty \theta_1{}(C_1)} \notag \\
      &\eta_1 = [\overline{d_1 \mapsto \tev[Q_1]}] \text{ for each } (d_1 \ty C_1) \in Q_1 \notag \\
      &\decelab{Q_2}{\Gamma}{x}{\forall\,\overline{b}.\,\overline{C_2} \Rightarrow \tau_2}{x} \notag \\
      &(x\ty{}\forall\,\overline{b}.\,\overline{C_2} \Rightarrow \tau_2) \in \Gamma \label{eq:var-var-in-context} \\
      &\overline{b_f} = \ftv{Q_2, \Gamma} \notag \\
      &\Gamma_\upsilon^2 \vdash \overline{\tev[Q_2]} \ty [\overline{b_f \mapsto \tau_{b_f}}]Q_2 \notag \\
      &\theta_2 = [\overline{b \mapsto \tau_b}, \overline{b_f \mapsto \tau_{b_f}}] \notag \\
      &\Gamma_\upsilon^2 \vdash \overline{\tev[C_2] \ty \theta_2{}(C_2)} \notag \\
      &\eta_2 = [\overline{d_2 \mapsto \tev[Q_2]}] \text{ for each } (d_2 \ty C_2) \in Q_2 \notag \\
      &\text{the principal type } \sigma \text{ is context-unambiguous where } \dec{\epsilon}{\Gamma}{x}{\sigma} \label{eq:var-var-ambi} \\
      &\taueq{\sigma}{\theta_1{}(\tau_1)}{\theta_2{}(\tau_2)} \label{eq:var-var-tau-eq} \\
      &\evidenceCanonicity{}(\Gamma_\upsilon^1, \Gamma_\upsilon^2, \theta_1, \theta_2, (\overline{\tev[Q_1]} \wedge \overline{\tev[C_1]}), (\overline{\tev[Q_2]} \wedge \overline{\tev[C_2]}), Q_1 \wedge \overline{C_1}, Q_2 \wedge \overline{C_2}) \label{eq:var-var-ev-canon}
    \end{align}
    And must prove that:
    \begin{align}
      &\tred{(\eta_1{}(x)~\overline{\upsilon_a}~\overline{\tev[C_1]})}{(\eta_2{}(x)~\overline{\upsilon_b}~\overline{\tev[C_2]})} \notag
    \end{align}
    Since $x \notin \dom{\eta_1}$ and $x \notin \dom{\eta_2}$, we know that:
    \begin{align}
      &\eta_1{}(x) = x \label{eq:var-domain-eta1} \\
      &\eta_2{}(x) = x \label{eq:var-var-domain-eta2}
    \end{align}
    Thus we can rewrite our goal using \cref{eq:var-domain-eta1} and \cref{eq:var-var-domain-eta2} to:
    \begin{align}
      &\tred{(x~\overline{\upsilon_a}~\overline{\tev[C_1]})}
        {(x~\overline{\upsilon_b}~\overline{\tev[C_2]})} \label{eq:var-var-goal}
      \end{align}

      As the context $\Gamma$ does not contain duplicate elements (we follow
      the Barendregt convention), we can deduce from \cref{eq:var-in-context}
      and \cref{eq:var-var-in-context} that
      $\forall\,\overline{a}.\,\overline{C_1} \Rightarrow \tau_1 =
      \forall\,\overline{b}.\,\overline{C_2} \Rightarrow \tau_2$ and thus
      that:
      \begin{align}
        &\overline{a} = \overline{b} \label{eq:var-var-a-b-eq} \\
        &\overline{C_1} = \overline{C_2} \label{eq:var-var-C1-C2-eq} \\
        &\tau_1 = \tau_2 \label{eq:var-var-tau1-tau2-eq}
      \end{align}
      We distinguish two cases based on the size of $\overline{C_1}$ and
      $\overline{C_2}$ \cref{eq:var-var-C1-C2-eq}:
      \begin{enumerate}
      \item Either $\overline{C_1} = \overline{C_2} = \epsilon$.
        Consequently:
        \begin{align*}
          &\overline{\tev[C_1]} \text{ is empty} \\
          &\overline{\tev[C_2]} \text{ is empty}
        \end{align*}
        Using this, we can rewrite our goal to:
        \begin{align}
          &\tred{(x~\overline{\upsilon_a})}{(x~\overline{\upsilon_b})} \label{eq:var-var-nil-goal}
        \end{align}
        To prove \cref{eq:var-var-nil-goal}, we start out by applying
        \textsc{Var} with $x$:
        \begin{align*}
          &\tred{x}{x}
        \end{align*}
        After which we can repeatedly apply \textsc{EraseTyApp} to add the
        type applications of $\overline{\upsilon_a}$. By using
        \textsc{Symmetry} we can do the same for $\overline{\upsilon_b}$ to
        obtain our goal:
        \begin{align*}
          &\tred{(x~\overline{\upsilon_a})}{(x~\overline{\upsilon_b})}
        \end{align*}

      \item Or $\overline{C_1}$ and $\overline{C_2}$ are not empty.

        From \cref{lemma:var-type-is-principal}, \cref{eq:var-var-a-b-eq}, and
        \cref{eq:var-var-C1-C2-eq} we know that
        $\overline{C_1} = \overline{C_2}$ are the constraints from a principal
        type of $x$. This, combined with \cref{eq:var-var-tau-eq} and
        \cref{eq:var-var-tau1-tau2-eq}, gives us that:
        \begin{align}
          &\theta_1{}(\overline{C_1}) = \theta_2{}(\overline{C_2}) \label{eq:var-var-cons-C1-C2-subst-eq}
        \end{align}
        If we combine this with \cref{eq:var-var-C1-C2-eq} and
        \cref{eq:var-var-ev-canon}, and Lemma~\ref{lem:canonical-ev-equiv} we get that for each $\tev[C_1]$  and  $\tev[C_2]$,
        \begin{equation}
          \tred{\tev[C_1]}{\tev[C_2]}
          \label{eq:var-var-cons-tred-tev1-tev2}
        \end{equation}

        To prove \cref{eq:var-var-goal}, we start out by applying
        \textsc{Var} with $x$:
        \begin{align*}
          &\tred{x}{x}
        \end{align*}
        After which we can repeatedly apply \textsc{EraseTyApp} to add the
        type applications of $\overline{\upsilon_a}$. By using
        \textsc{Symmetry} we can do the same for $\overline{\upsilon_b}$:
        \begin{align*}
          &\tred{(x~\overline{\upsilon_a})}{(x~\overline{\upsilon_b})}
        \end{align*}
        Finally, we repeatedly apply \textsc{App} using
        \cref{eq:var-var-cons-tred-tev1-tev2} to add the evidence applications of
        $\overline{\tev[C_1]}$. By using \textsc{Symmetry} we can do the same for
        $\overline{\tev[C_2]}$ to obtain our goal \cref{eq:var-var-goal}:
        \begin{align*}
          &\tred{(x~\overline{\upsilon_a}~\overline{\tev[C_1]})}{(x~\overline{\upsilon_b}~\overline{\tev[C_2]})}
        \end{align*}
      \end{enumerate}



  \item[\pcase{abs-abs}.] The \textsc{Abs} rule was last used in both derivations.
    We have that:
    \begin{align}
      &\decelab{Q_1}{\Gamma}{\lambda{}x.\,e'}{\tau_1' \rightarrow \tau_2'}{\lambda{}(x\ty{}\upsilon_1).\,t_1'} \notag \\
      &\decelab{Q_1}{(x\ty{}\tau_1'), \Gamma}{e'}{\tau_2'}{t_1'} \label{eq:abs-hty1} \\
      &\elab{\tau}{\tau_1'}{\upsilon_1} \notag \\
      &\overline{a_f} = \ftv{Q_1, \Gamma} \label{eq:abs-ftv1} \\
      &\Gamma_\upsilon^1 \vdash \overline{\tev[Q_1]} \ty [\overline{a_f \mapsto \tau_{a_f}}]Q_1 \label{eq:abs-tevQ1} \\
      &\theta_1 = [\overline{a \mapsto \tau_a}, \overline{a_f \mapsto \tau_{a_f}}] \label{eq:abs-theta1} \\
      &\Gamma_\upsilon^1 \vdash \overline{\tev[C_1] \ty \theta_1{}(C_1)} \label{eq:abs-tevC1} \\
      &\eta_1 = [\overline{d_1 \mapsto \tev[Q_1]}] \text{ for each } (d_1 \ty C_1) \in Q_1 \label{eq:abs-eta1} \\
      &\notag \\
      &\decelab{Q_2}{\Gamma}{\lambda{}x.\,e'}{\tau_1'' \rightarrow \tau_2''}{\lambda{}(x\ty{}\upsilon_2).\,t_2'} \notag \\
      &\decelab{Q_2}{(x\ty{}\tau_1''), \Gamma}{e'}{\tau_2''}{t_2'} \label{eq:abs-hty2} \\
      &\elab{\tau}{\tau_1''}{\upsilon_2} \notag \\
      &\overline{b_f} = \ftv{Q_2, \Gamma} \label{eq:abs-ftv2} \\
      &\Gamma_\upsilon^2 \vdash \overline{\tev[Q_2]} \ty [\overline{b_f \mapsto \tau_{b_f}}]Q_2 \label{eq:abs-tevQ2} \\
      &\theta_2 = [\overline{b \mapsto \tau_b}, \overline{b_f \mapsto \tau_{b_f}}] \label{eq:abs-theta2} \\
      &\Gamma_\upsilon^2 \vdash \overline{\tev[C_2] \ty \theta_2{}(C_2)} \label{eq:abs-tevC2} \\
      &\eta_2 = [\overline{d_2 \mapsto \tev[Q_2]}] \text{ for each } (d_2 \ty C_2) \in Q_2 \label{eq:abs-eta2} \\
      &\text{the principal type } \sigma \text{ is context-unambiguous where } \dec{\epsilon}{\Gamma}{\lambda{}x.\,e'}{\sigma} \label{eq:abs-ambi} \\
      &\taueq{\sigma}{\theta_1{}(\tau_1' \rightarrow \tau_2')}{\theta_2{}(\tau_1'' \rightarrow \tau_2'')} \label{eq:abs-tau-eq} \\
      &\evidenceCanonicity{}(\Gamma_\upsilon^1, \Gamma_\upsilon^2, \theta_1, \theta_2, (\overline{\tev[Q_1]} \wedge \overline{\tev[C_1]}), (\overline{\tev[Q_2]} \wedge \overline{\tev[C_2]}), Q_1 \wedge \overline{C_1}, Q_2 \wedge \overline{C_2}) \label{eq:abs-ev-canon}
    \end{align}

    As
    $\forall\,\overline{a}.\,\overline{C_1} \Rightarrow \tau_1 = \tau_1'
    \rightarrow \tau_2'$ and
    $\forall\,\overline{b}.\,\overline{C_2} \Rightarrow \tau_2 = \tau_1''
    \rightarrow \tau_2''$ because of the \textsc{Abs} rule, we have that:
    \begin{align}
      &\overline{a} \text{ is empty} \label{eq:abs-empty-a} \\
      &\overline{b} \text{ is empty} \label{eq:abs-empty-b} \\
      &\overline{C_1} \text{ is empty} \label{eq:abs-empty-C1} \\
      &\overline{C_2} \text{ is empty} \label{eq:abs-empty-C2} \\
      &\tau_1 = \tau_1' \rightarrow \tau_2' \label{eq:abs-tau-1-type} \\
      &\tau_2 = \tau_1'' \rightarrow \tau_2'' \label{eq:abs-tau-2-type}
    \end{align}
    From \cref{eq:abs-empty-a} and \cref{eq:abs-empty-b} follows that:
    \begin{align}
      &\overline{\tau_a} \text{ is empty} \label{eq:abs-empty-tau-a} \\
      &\overline{\tau_b} \text{ is empty} \label{eq:abs-empty-tau-b}
    \end{align}
    This means we can rewrite \cref{eq:abs-theta1} and \cref{eq:abs-theta2}
    to:
    \begin{align}
      &\theta_1 = [\overline{a_f \mapsto \tau_{a_f}}] \label{eq:abs-theta1'} \\
      &\theta_2 = [\overline{b_f \mapsto \tau_{b_f}}] \label{eq:abs-theta2'}
    \end{align}
    And from \cref{eq:abs-empty-C1} and \cref{eq:abs-empty-C2} follows that:
    \begin{align}
      &\overline{\tev[C_1]} \text{ is empty} \label{eq:abs-empty-tevs1} \\
      &\overline{\tev[C_2]} \text{ is empty} \label{eq:abs-empty-tevs2}
    \end{align}
    This means we can omit \cref{eq:abs-tevC1} and \cref{eq:abs-tevC2}.

    We can also rewrite \cref{eq:abs-ev-canon} to:
    \begin{align}
      &\evidenceCanonicity{}(\Gamma_\upsilon^1, \Gamma_\upsilon^2, \theta_1, \theta_2, \overline{\tev[Q_1]}, \overline{\tev[Q_2]}, Q_1, Q_2) \label{eq:abs-ev-canon'}
    \end{align}

    Because of the Barendregt convention, \cref{eq:abs-hty1}, and
    \cref{eq:abs-hty2}, we know:
    \begin{align}
      &\tau_1' = \tau_1'' \label{eq:abs-tau1}
    \end{align}

    We must now prove that:
    \begin{align*}
      &\tred{\eta_1{}(\lambda{}(x\ty{}\upsilon_1).\,t_1')}
        {\eta_2{}(\lambda{}(x\ty{}\upsilon_2).\,t_2')}
    \end{align*}
    We can rewrite this to:
    \begin{align}
      &\tred{(\lambda{}(x\ty{}\upsilon_1).\,\eta_1{}(t_1'))}
        {(\lambda{}(x\ty{}\upsilon_2).\,\eta_2{}(t_2'))} \label{eq:abs-goal}
    \end{align}

    Our induction hypothesis, is as follows:

    \begin{align}
      &\forall~\overline{a_f'}~\overline{\tau_{a_f}'}~\overline{\tev[Q_1]'}~\theta_1'~\eta_1'~\sigma', \notag \\
      &\quad\overline{a_f'} = \ftv{Q_1, (x\ty{}\tau_1'), \Gamma} \rightarrow \notag \\
      &\quad\Gamma_\upsilon^1 \vdash \overline{\tev[Q_1]'} \ty [\overline{a_f' \mapsto \tau_{a_f}'}]Q_1 \rightarrow \notag \\
      &\quad\theta_1' = [\overline{a_f' \mapsto \tau_{a_f}'}] \rightarrow \notag \\
      &\quad\eta_1' = [\overline{d_1 \mapsto \tev[Q_1]'}] \text{ for each } (d_1 \ty C_1) \in Q_1 \rightarrow \notag \\
      &\quad\forall~Q_2~\overline{b}~\overline{C_2}~\tau_2~t_2, \notag \\
      &\quad\quad\decelab{Q_2}{(x\ty{}\tau_1'), \Gamma}{e'}{\forall\,\overline{b}.\,\overline{C_2} \Rightarrow \tau_2}{t_2} \rightarrow \notag \\
      &\quad\quad\forall~\overline{b_f'}~\overline{\tau_{b_f}'}~\overline{\tev[Q_2]'}~\overline{\tev[C_2]'}~\overline{\tau_b'}~\theta_2'~\eta_2', \notag \\
      &\quad\quad\quad\overline{b_f'} = \ftv{Q_2, (x\ty{}\tau_1'), \Gamma} \rightarrow \notag \\
      &\quad\quad\quad\Gamma_\upsilon^2 \vdash \overline{\tev[Q_2]'} \ty [\overline{b_f' \mapsto \tau_{b_f}'}]Q_2 \rightarrow \notag \\
      &\quad\quad\quad\theta_2' = [\overline{b \mapsto \tau_b'}, \overline{b_f' \mapsto \tau_{b_f}'}] \rightarrow \notag \\
      &\quad\quad\quad\Gamma_\upsilon^2 \vdash \overline{\tev[C_2]' \ty \theta_2'{}(C_2)} \rightarrow \notag \\
      &\quad\quad\quad\eta_2' = [\overline{d_2 \mapsto \tev[Q_2]'}] \text{ for each } (d_2 \ty C_2) \in Q_2 \rightarrow \notag \\
      &\quad\quad\quad\text{the principal type } \sigma' \text{ is context-unambiguous where } \dec{\epsilon}{(x \ty \tau_1'), \Gamma}{e'}{\sigma} \rightarrow \notag \\
      &\quad\quad\quad\taueq{\sigma}{\theta_1'{}(\tau_2')}{\theta_2{}(\tau_2)} \rightarrow \notag \\
      &\quad\quad\quad\evidenceCanonicity{}(\Gamma_\upsilon^1, \Gamma_\upsilon^2, \theta_1', \theta_2', \overline{\tev[Q_1]'}, (\overline{\tev[Q_2]'} \wedge \overline{\tev[C_2]'}), Q_1, Q_2 \wedge \overline{C_2}) \rightarrow \notag \\
      &\quad\quad\quad\tred{(\eta_1{}'(t_1'))}
        {(\eta_2'{}(t_2)~\overline{\upsilon_b'}~\overline{\tev[C_2]'})} \label{eq:abs-ih}
    \end{align}

    If we instantiate \cref{eq:abs-ih} with the following variables:
    $\overline{a_f'} = \overline{a_f}$,
    $\overline{\tau_{a_f}'} = \overline{\tau_{a_f}}$,
    $\overline{\tev[Q_1]'} = \overline{\tev[Q_1]'}$, $\theta_1' = \theta_1$,
    $\eta_1' = \eta_1$, $Q_2 = Q_2$, $\overline{b}$ empty, $\overline{C_2}$
    empty, $\tau_2 = \tau_2''$, $t_2 = t_2'$,
    $\overline{b_f'} = \overline{b_f}$,
    $\overline{\tau_{b_f}'} = \overline{\tau_{b_f}}$,
    $\overline{\tev[Q_2]'} = \overline{\tev[Q_2]}$, $\overline{\tev[C_2]}$
    empty, $\overline{\tau_b'}$ empty, $\theta_2' = \theta_2$,
    $\eta_2 = \eta_2$, and $\sigma'$ with the principal type $\sigma'$
    inferred by the algorithm (see \cref{thm:principal-types}), we
    must prove the following:

    \begin{align}
      &\overline{a_f} = \ftv{Q_1, (x\ty{}\tau_1'), \Gamma} \label{eq:abs-ih-ftv1} \\
      &\Gamma_\upsilon^1 \vdash \overline{\tev[Q_1]} \ty [\overline{a_f \mapsto \tau_{a_f}}]Q_1 \label{eq:abs-ih-tevQ1} \\
      &\theta_1 = [\overline{a_f \mapsto \tau_{a_f}}] \label{eq:abs-ih-theta1} \\
      &\eta_1 = [\overline{d_1 \mapsto \tev[Q_1]}] \text{ for each } (d_1 \ty C_1) \in Q_1 \label{eq:abs-ih-eta1} \\
      &\decelab{Q_2}{(x\ty{}\tau_1'), \Gamma}{e'}{\tau_2''}{t_2'} \label{eq:abs-ih-hty2} \\
      &\overline{b_f} = \ftv{Q_2, (x\ty{}\tau_1'), \Gamma} \label{eq:abs-ih-ftv2} \\
      &\Gamma_\upsilon^2 \vdash \overline{\tev[Q_2]} \ty [\overline{b_f \mapsto \tau_{b_f}}]Q_2 \label{eq:abs-ih-tevQ2} \\
      &\theta_2 = [\overline{b_f \mapsto \tau_{b_f}}] \label{eq:abs-ih-theta2} \\
      &\eta_2 = [\overline{d_2 \mapsto \tev[Q_2]}] \text{ for each } (d_2 \ty C_2) \in Q_2 \label{eq:abs-ih-eta2} \\
      &\text{the principal type } \sigma' \text{ is context-unambiguous where } \dec{\epsilon}{(x \ty \tau_1'), \Gamma}{e'}{\sigma'} \label{eq:abs-ih-ambi} \\
      &\taueq{\sigma'}{\theta_1{}(\tau_2')}{\theta_2{}(\tau_2'')} \label{eq:abs-ih-tau-eq} \\
      &\evidenceCanonicity{}(\Gamma_\upsilon^1, \Gamma_\upsilon^2, \theta_1, \theta_2, \overline{\tev[Q_1]}, \overline{\tev[Q_2]}, Q_1, Q_2) \label{eq:abs-ih-ev-canon}
    \end{align}
    To obtain:
    \begin{align}
      &\tred{(\eta_1{}(t_1'))}{(\eta_2{}(t_2'))} \label{eq:abs-ih-conclusion}
    \end{align}
    \cref{eq:abs-ih-ftv1} follows from \cref{eq:abs-ftv1} and the fact that
    any free type variables in $\tau'$ are also free type variables in
    $\Gamma$, \cref{eq:abs-ih-tevQ1} follows from \cref{eq:abs-tevQ1},
    \cref{eq:abs-ih-theta1} from \cref{eq:abs-theta1'}, \cref{eq:abs-ih-eta1}
    from \cref{eq:abs-eta1}, \cref{eq:abs-ih-hty2} from \cref{eq:abs-hty2} and
    \cref{eq:abs-tau1}, \cref{eq:abs-ih-ftv2} from \cref{eq:abs-ftv2} and
    \cref{eq:abs-tau1}, \cref{eq:abs-ih-tevQ2} from \cref{eq:abs-tevQ2},
    \cref{eq:abs-ih-theta2} from \cref{eq:abs-theta2'}, \cref{eq:abs-ih-eta2}
    from \cref{eq:abs-eta2}, \cref{eq:abs-ih-ambi} from \cref{eq:abs-ambi}
    using \cref{lemma:abs-ih-ambi}. \cref{eq:abs-ih-tau-eq} from
    \cref{eq:abs-tau-eq}, and \cref{eq:abs-ih-ev-canon} from
    \cref{eq:abs-ev-canon'}.

    If we apply \textsc{Abs} to \cref{eq:abs-ih-conclusion}, we get our goal \cref{eq:abs-goal}:
    \begin{align*}
      &\tred{(\lambda{}(x\ty{}\upsilon_1).\,\eta_1{}(t_1'))}
        {(\lambda{}(x\ty{}\upsilon_2).\,\eta_2{}(t_2'))}
    \end{align*}



  \item[\pcase{app-app}.] The \textsc{App} rule was last used in both derivations.
    We have that:
    \begin{align}
      &\decelab{Q_1}{\Gamma}{e_1'~e_2'}{\tau_2'}{t_1'~t_2'} \notag \\
      &\decelab{Q_1}{\Gamma}{e_1'}{\tau_1' \rightarrow \tau_2'}{t_1'} \label{eq:app-hty11} \\
      &\decelab{Q_1}{\Gamma}{e_2'}{\tau_1'}{t_2'} \label{eq:app-hty12} \\
      &\overline{a_f} = \ftv{Q_1, \Gamma} \label{eq:app-ftv1} \\
      &\Gamma_\upsilon^1 \vdash \overline{\tev[Q_1]} \ty [\overline{a_f \mapsto \tau_{a_f}}]Q_1 \label{eq:app-tevQ1} \\
      &\theta_1 = [\overline{a \mapsto \tau_a}, \overline{a_f \mapsto \tau_{a_f}}] \label{eq:app-theta1} \\
      &\Gamma_\upsilon^1 \vdash \overline{\tev[C_1] \ty \theta_1{}(C_1)} \label{eq:app-tevC1} \\
      &\eta_1 = [\overline{d_1 \mapsto \tev[Q_1]}] \text{ for each } (d_1 \ty C_1) \in Q_1 \label{eq:app-eta1} \\
      &\notag \\
      &\decelab{Q_2}{\Gamma}{e_1'~e_2'}{\tau_2''}{t_1''~t_2''} \notag \\
      &\decelab{Q_2}{\Gamma}{e_1'}{\tau_1'' \rightarrow \tau_2''}{t_1''} \label{eq:app-hty21} \\
      &\decelab{Q_2}{\Gamma}{e_2'}{\tau_1''}{t_2''} \label{eq:app-hty22} \\
      &\overline{b_f} = \ftv{Q_2, \Gamma} \label{eq:app-ftv2} \\
      &\Gamma_\upsilon^2 \vdash \overline{\tev[Q_2]} \ty [\overline{b_f \mapsto \tau_{b_f}}]Q_2 \label{eq:app-tevQ2} \\
      &\theta_2 = [\overline{b \mapsto \tau_b}, \overline{b_f \mapsto \tau_{b_f}}] \label{eq:app-theta2} \\
      &\Gamma_\upsilon^2 \vdash \overline{\tev[C_2] \ty \theta_2{}(C_2)} \label{eq:app-tevC2} \\
      &\eta_2 = [\overline{d_2 \mapsto \tev[Q_2]}] \text{ for each } (d_2 \ty C_2) \in Q_2 \label{eq:app-eta2} \\
      &\text{the principal type } \sigma \text{ is context-unambiguous where } \dec{\epsilon}{\Gamma}{e_1'~e_2'}{\sigma} \label{eq:app-ambi} \\
      &\taueq{\sigma}{\theta_1{}(\tau_2')}{\theta_2{}(\tau_2'')} \label{eq:app-tau-eq} \\
      &\evidenceCanonicity{}(\Gamma_\upsilon^1, \Gamma_\upsilon^2, \theta_1, \theta_2, (\overline{\tev[Q_1]} \wedge \overline{\tev[C_1]}), (\overline{\tev[Q_2]} \wedge \overline{\tev[C_2]}), Q_1 \wedge \overline{C_1}, Q_2 \wedge \overline{C_2}) \label{eq:app-ev-canon}
    \end{align}
    As $\forall\,\overline{a}.\,\overline{C_1} \Rightarrow \tau_1 = \tau_2'$
    and $\forall\,\overline{b}.\,\overline{C_2} \Rightarrow \tau_2 = \tau_2''$
    because of the \textsc{App} rule, we have that:
    \begin{align}
      &\overline{a} \text{ is empty} \label{eq:app-empty-a} \\
      &\overline{b} \text{ is empty} \label{eq:app-empty-b} \\
      &\overline{C_1} \text{ is empty} \label{eq:app-empty-C1} \\
      &\overline{C_2} \text{ is empty} \label{eq:app-empty-C2} \\
      &\tau_1 = \tau_2' \label{eq:app-tau-1-type} \\
      &\tau_2 = \tau_2'' \label{eq:app-tau-2-type}
    \end{align}
    From \cref{eq:app-empty-a} and \cref{eq:app-empty-b} follows that:
    \begin{align}
      &\overline{\tau_a} \text{ is empty} \label{eq:app-empty-tau-a} \\
      &\overline{\tau_b} \text{ is empty} \label{eq:app-empty-tau-b}
    \end{align}
    This means we can rewrite \cref{eq:app-theta1} and \cref{eq:app-theta2} to:
    \begin{align}
      &\theta_1 = [\overline{a_f \mapsto \tau_{a_f}}] \label{eq:app-theta1'} \\
      &\theta_2 = [\overline{b_f \mapsto \tau_{b_f}}] \label{eq:app-theta2'}
    \end{align}
    From \cref{eq:app-empty-C1} and \cref{eq:app-empty-C2} it that:
    \begin{align}
      &\overline{\tev[C_1]} \text{ is empty} \label{eq:app-empty-tevC1} \\
      &\overline{\tev[C_2]} \text{ is empty} \label{eq:app-empty-tevC2}
    \end{align}
    This means we can omit \cref{eq:app-tevC1} and \cref{eq:app-tevC2}.

    We can also rewrite \cref{eq:app-ev-canon} to:
    \begin{align}
      &\evidenceCanonicity{}(\Gamma_\upsilon^1, \Gamma_\upsilon^2, \theta_1, \theta_2, \overline{\tev[Q_1]}, \overline{\tev[Q_2]}, Q_1, Q_2) \label{eq:app-ev-canon'}
    \end{align}
    From \cref{lemma:ambi-app1}, \cref{lemma:ambi-app2}, and
    \cref{eq:app-ambi}, we have that:
    \begin{align}
      &\text{the principal type } \sigma_1 \text{ is context-unambiguous where } \dec{\epsilon}{\Gamma}{e_1'}{\sigma_1} \label{eq:app-ambi-1} \\
      &\text{the principal type } \sigma_2 \text{ is context-unambiguous where } \dec{\epsilon}{\Gamma}{e_2'}{\sigma_2} \label{eq:app-ambi-2}
    \end{align}
    As we will be needing them, we first prove that:
    \begin{align}
      &\taueq{\sigma_1}{\theta_1{}(\tau_1' \rightarrow \tau_2')}{\theta_2{}(\tau_1'' \rightarrow \tau_2'')} \label{eq:app-tau-12-eq} \\
      &\taueq{\sigma_2}{\theta_1{}(\tau_1')}{\theta_2{}(\tau_1'')} \label{eq:app-tau-1-eq}
    \end{align}
    These two statements follow from applying \cref{lemma:tau-eq-app12} and
    \cref{lemma:tau-eq-app1} with \cref{eq:app-hty11}, \cref{eq:app-hty21},
    \cref{eq:app-hty12}, \cref{eq:app-hty22}, \cref{eq:app-tau-eq}, and
    \cref{eq:app-ambi-1}.

    We must now prove that:
    \begin{align}
      &\tred{\eta_1{}(t_1'~t_2')}{\eta_2{}(t_1''~t_2'')} \label{eq:app-goal}
    \end{align}
    Which we can rewrite to:
    \begin{align}
      &\tred{(\eta_1{}(t_1')~\eta_1{}(t_2'))}{(\eta_2{}(t_1'')~\eta_2{}(t_2''))} \label{eq:app-goal'}
    \end{align}

    Our induction hypotheses, simplified using \cref{eq:app-empty-a},
    \cref{eq:app-empty-C1}, \cref{eq:app-tau-1-type},
    \cref{eq:app-empty-tau-a}, and \cref{eq:app-empty-tevC1}, are as follows:

    \begin{align}
      &\forall~\overline{a_f'}~\overline{\tau_{a_f}'}~\overline{\tev[Q_1]'}~\theta_1'~\eta_1'~\sigma', \notag \\
      &\quad\overline{a_f'} = \ftv{Q_1, \Gamma} \rightarrow \notag \\
      &\quad\Gamma_\upsilon^1 \vdash \overline{\tev[Q_1]'} \ty [\overline{a_f' \mapsto \tau_{a_f}'}]Q_1 \rightarrow \notag \\
      &\quad\theta_1' = [\overline{a_f' \mapsto \tau_{a_f}'}] \rightarrow \notag \\
      &\quad\eta_1' = [\overline{d_1 \mapsto \tev[Q_1]'}] \text{ for each } (d_1 \ty C_1) \in Q_1 \rightarrow \notag \\
      &\quad\forall~Q_2~\overline{b}~\overline{C_2}~\tau_2~t_2, \notag \\
      &\quad\quad\decelab{Q_2}{\Gamma}{e_1'}{\forall\,\overline{b}.\,\overline{C_2} \Rightarrow \tau_2}{t_2} \rightarrow \notag \\
      &\quad\quad\forall~\overline{b_f'}~\overline{\tau_{b_f}'}~\overline{\tev[Q_2]'}~\overline{\tev[C_2]'}~\overline{\tau_b'}~\theta_2'~\eta_2', \notag \\
      &\quad\quad\quad\overline{b_f'} = \ftv{Q_2, \Gamma} \rightarrow \notag \\
      &\quad\quad\quad\Gamma_\upsilon^2 \vdash \overline{\tev[Q_2]'} \ty [\overline{b_f' \mapsto \tau_{b_f}'}]Q_2 \rightarrow \notag \\
      &\quad\quad\quad\theta_2' = [\overline{b \mapsto \tau_b'}, \overline{b_f' \mapsto \tau_{b_f}'}] \rightarrow \notag \\
      &\quad\quad\quad\Gamma_\upsilon^2 \vdash \overline{\tev[C_2]' \ty \theta_2'{}(C_2)} \rightarrow \notag \\
      &\quad\quad\quad\eta_2' = [\overline{d_2 \mapsto \tev[Q_2]'}] \text{ for each } (d_2 \ty C_2) \in Q_2 \rightarrow \notag \\
      &\quad\quad\quad\text{the principal type } \sigma' \text{ is context-unambiguous where } \dec{\epsilon}{\Gamma}{e_1'}{\sigma'} \rightarrow \notag \\
      &\quad\quad\quad\taueq{\sigma'}{\theta_1'{}(\tau_1' \rightarrow \tau_2')}{\theta_2{}(\tau_2)} \rightarrow \notag \\
      &\quad\quad\quad\evidenceCanonicity{}(\Gamma_\upsilon^1, \Gamma_\upsilon^2, \theta_1', \theta_2', \overline{\tev[Q_1]'}, (\overline{\tev[Q_2]'} \wedge \overline{\tev[C_2]'}), Q_1, Q_2 \wedge \overline{C_2}) \rightarrow \notag \\
      &\quad\quad\quad\tred{(\eta_1{}'(t_1'))}
        {(\eta_2'{}(t_2)~\overline{\upsilon_b'}~\overline{\tev[C_2]'})} \label{eq:app-ih1}
    \end{align}

    \begin{align}
      &\forall~\overline{a_f'}~\overline{\tau_{a_f}'}~\overline{\tev[Q_1]'}~\theta_1'~\eta_1'~\sigma', \notag \\
      &\quad\overline{a_f'} = \ftv{Q_1, \Gamma} \rightarrow \notag \\
      &\quad\Gamma_\upsilon^1 \vdash \overline{\tev[Q_1]'} \ty [\overline{a_f' \mapsto \tau_{a_f}'}]Q_1 \rightarrow \notag \\
      &\quad\theta_1' = [\overline{a_f' \mapsto \tau_{a_f}'}] \rightarrow \notag \\
      &\quad\eta_1' = [\overline{d_1 \mapsto \tev[Q_1]'}] \text{ for each } (d_1 \ty C_1) \in Q_1 \rightarrow \notag \\
      &\quad\forall~Q_2~\overline{b}~\overline{C_2}~\tau_2~t_2, \notag \\
      &\quad\quad\decelab{Q_2}{\Gamma}{e_2'}{\forall\,\overline{b}.\,\overline{C_2} \Rightarrow \tau_2}{t_2} \rightarrow \notag \\
      &\quad\quad\forall~\overline{b_f'}~\overline{\tau_{b_f}'}~\overline{\tev[Q_2]'}~\overline{\tev[C_2]'}~\overline{\tau_b'}~\theta_2'~\eta_2', \notag \\
      &\quad\quad\quad\overline{b_f'} = \ftv{Q_2, \Gamma} \rightarrow \notag \\
      &\quad\quad\quad\Gamma_\upsilon^2 \vdash \overline{\tev[Q_2]'} \ty [\overline{b_f' \mapsto \tau_{b_f}'}]Q_2 \rightarrow \notag \\
      &\quad\quad\quad\theta_2' = [\overline{b \mapsto \tau_b'}, \overline{b_f' \mapsto \tau_{b_f}'}] \rightarrow \notag \\
      &\quad\quad\quad\Gamma_\upsilon^2 \vdash \overline{\tev[C_2]' \ty \theta_2'{}(C_2)} \rightarrow \notag \\
      &\quad\quad\quad\eta_2' = [\overline{d_2 \mapsto \tev[Q_2]'}] \text{ for each } (d_2 \ty C_2) \in Q_2 \rightarrow \notag \\
      &\quad\quad\quad\text{the principal type } \sigma' \text{ is context-unambiguous where } \dec{\epsilon}{\Gamma}{e_2'}{\sigma'} \rightarrow \notag \\
      &\quad\quad\quad\taueq{\sigma'}{\theta_1'{}(\tau_1')}{\theta_2{}(\tau_2)} \rightarrow \notag \\
      &\quad\quad\quad\evidenceCanonicity{}(\Gamma_\upsilon^1, \Gamma_\upsilon^2, \theta_1', \theta_2', \overline{\tev[Q_1]'}, (\overline{\tev[Q_2]'} \wedge \overline{\tev[C_2]'}), Q_1, Q_2 \wedge \overline{C_2}) \rightarrow \notag \\
      &\quad\quad\quad\tred{(\eta_1{}'(t_2'))}
        {(\eta_2'{}(t_2)~\overline{\upsilon_b'}~\overline{\tev[C_2]'})} \label{eq:app-ih2}
    \end{align}

    If we instantiate our first induction hypothesis \cref{eq:app-ih1} with
    the following variables: $\overline{a_f'} = \overline{a_f}$,
    $\overline{\tau_{a_f}'} = \overline{\tau_{a_f}}$,
    $\overline{\tev[Q_1]'} = \overline{\tev[Q_1]'}$, $\theta_1' = \theta_1$,
    $\eta_1' = \eta_1$, $Q_2 = Q_2$, $\overline{b}$ empty, $\overline{C_2}$
    empty, $\tau_2 = (\tau_1'' \rightarrow \tau_2'')$, $t_2 = t_1''$,
    $\overline{b_f'} = \overline{b_f}$,
    $\overline{\tau_{b_f}'} = \overline{\tau_{b_f}}$,
    $\overline{\tev[Q_2]'} = \overline{\tev[Q_2]}$, $\overline{\tev[C_2]}$
    empty, $\overline{\tau_b'}$ empty, $\theta_2' = \theta_2$,
    $\eta_2 = \eta_2$, and $\sigma' = \sigma_1$, we must prove the following:

    \begin{align}
      &\overline{a_f} = \ftv{Q_1, \Gamma} \label{eq:app-ih1-ftv1} \\
      &\Gamma_\upsilon^1 \vdash \overline{\tev[Q_1]} \ty [\overline{a_f \mapsto \tau_{a_f}}]Q_1 \label{eq:app-ih1-tevQ1} \\
      &\theta_1 = [\overline{a_f \mapsto \tau_{a_f}}] \label{eq:app-ih1-theta1} \\
      &\eta_1 = [\overline{d_1 \mapsto \tev[Q_1]}] \text{ for each } (d_1 \ty C_1) \in Q_1 \label{eq:app-ih1-eta1} \\
      &\decelab{Q_2}{\Gamma}{e_1'}{\tau_1'' \rightarrow \tau_2''}{t_1''} \label{eq:app-ih1-hty2} \\
      &\overline{b_f} = \ftv{Q_2, \Gamma} \label{eq:app-ih1-ftv2} \\
      &\Gamma_\upsilon^2 \vdash \overline{\tev[Q_2]} \ty [\overline{b_f \mapsto \tau_{b_f}}]Q_2 \label{eq:app-ih1-tevQ2} \\
      &\theta_2 = [\overline{b_f \mapsto \tau_{b_f}}] \label{eq:app-ih1-theta2} \\
      &\eta_2 = [\overline{d_2 \mapsto \tev[Q_2]}] \text{ for each } (d_2 \ty C_2) \in Q_2 \label{eq:app-ih1-eta2} \\
      &\text{the principal type } \sigma_1 \text{ is context-unambiguous where } \dec{\epsilon}{\Gamma}{e_1'}{\sigma_1} \label{eq:app-ih1-ambi} \\
      &\taueq{\sigma_1}{\theta_1{}(\tau_1' \rightarrow \tau_2')}{\theta_2{}(\tau_1'' \rightarrow \tau_2'')} \label{eq:app-ih1-tau-eq} \\
      &\evidenceCanonicity{}(\Gamma_\upsilon^1, \Gamma_\upsilon^2, \theta_1, \theta_2, \overline{\tev[Q_1]}, \overline{\tev[Q_2]}, Q_1, Q_2) \label{eq:app-ih1-ev-canon}
    \end{align}
    To obtain:
    \begin{align}
      &\tred{(\eta_1{}(t_1'))}{(\eta_2{}(t_1''))} \label{eq:app-ih1-conclusion}
    \end{align}
    \cref{eq:app-ih1-ftv1} follows from \cref{eq:app-ftv1},
    \cref{eq:app-ih1-tevQ1} follows from \cref{eq:app-tevQ1},
    \cref{eq:app-ih1-theta1} from \cref{eq:app-theta1'},
    \cref{eq:app-ih1-eta1} from \cref{eq:app-eta1}, \cref{eq:app-ih1-hty2}
    from \cref{eq:app-hty21}, \cref{eq:app-ih1-ftv2} from \cref{eq:app-ftv2},
    \cref{eq:app-ih1-tevQ2} from \cref{eq:app-tevQ2}, \cref{eq:app-ih1-theta2}
    from \cref{eq:app-theta2'}, \cref{eq:app-ih1-eta2} from
    \cref{eq:app-eta2}, \cref{eq:app-ih1-ambi} from \cref{eq:app-ambi-1},
    \cref{eq:app-ih1-tau-eq} from \cref{eq:app-tau-12-eq}, and
    \cref{eq:app-ih1-ev-canon} from \cref{eq:app-ev-canon'}.

    If we instantiate \cref{eq:app-ih2} with the following variables:
    $\overline{a_f'} = \overline{a_f}$,
    $\overline{\tau_{a_f}'} = \overline{\tau_{a_f}}$,
    $\overline{\tev[Q_1]'} = \overline{\tev[Q_1]'}$, $\theta_1' = \theta_1$,
    $\eta_1' = \eta_1$, $Q_2 = Q_2$, $\overline{b}$ empty, $\overline{C_2}$
    empty, $\tau_2 = \tau_1''$, $t_2 = t_2''$,
    $\overline{b_f'} = \overline{b_f}$,
    $\overline{\tau_{b_f}'} = \overline{\tau_{b_f}}$,
    $\overline{\tev[Q_2]'} = \overline{\tev[Q_2]}$, $\overline{\tev[C_2]}$
    empty, $\overline{\tau_b'}$ empty, $\theta_2' = \theta_2$,
    $\eta_2 = \eta_2$, and $\sigma' = \sigma_2$, we must prove the following:

    \begin{align}
      &\overline{a_f} = \ftv{Q_1, \Gamma} \label{eq:app-ih2-ftv1} \\
      &\Gamma_\upsilon^1 \vdash \overline{\tev[Q_1]} \ty [\overline{a_f \mapsto \tau_{a_f}}]Q_1 \label{eq:app-ih2-tevQ1} \\
      &\theta_1 = [\overline{a_f \mapsto \tau_{a_f}}] \label{eq:app-ih2-theta1} \\
      &\eta_1 = [\overline{d_1 \mapsto \tev[Q_1]}] \text{ for each } (d_1 \ty C_1) \in Q_1 \label{eq:app-ih2-eta1} \\
      &\decelab{Q_2}{\Gamma}{e_2'}{\tau_1''}{t_2''} \label{eq:app-ih2-hty2} \\
      &\overline{b_f} = \ftv{Q_2, \Gamma} \label{eq:app-ih2-ftv2} \\
      &\Gamma_\upsilon^2 \vdash \overline{\tev[Q_2]} \ty [\overline{b_f \mapsto \tau_{b_f}}]Q_2 \label{eq:app-ih2-tevQ2} \\
      &\theta_2 = [\overline{b_f \mapsto \tau_{b_f}}] \label{eq:app-ih2-theta2} \\
      &\eta_2 = [\overline{d_2 \mapsto \tev[Q_2]}] \text{ for each } (d_2 \ty C_2) \in Q_2 \label{eq:app-ih2-eta2} \\
      &\text{the principal type } \sigma_2 \text{ is context-unambiguous where } \dec{\epsilon}{\Gamma}{e_2}{\sigma_2} \label{eq:app-ih2-ambi} \\
      &\taueq{\sigma_2}{\theta_1{}(\tau_1')}{\theta_2{}(\tau_1'')} \label{eq:app-ih2-tau-eq} \\
      &\evidenceCanonicity{}(\Gamma_\upsilon^1, \Gamma_\upsilon^2, \theta_1, \theta_2, \overline{\tev[Q_1]}, \overline{\tev[Q_2]}, Q_1, Q_2) \label{eq:app-ih2-ev-canon}
    \end{align}
    To obtain:
    \begin{align}
      &\tred{(\eta_1{}(t_2'))}{(\eta_2{}(t_2''))} \label{eq:app-ih2-conclusion}
    \end{align}
    \cref{eq:app-ih2-ftv1} follows from \cref{eq:app-ftv1},
    \cref{eq:app-ih2-tevQ1} follows from \cref{eq:app-tevQ1},
    \cref{eq:app-ih2-theta1} from \cref{eq:app-theta1'},
    \cref{eq:app-ih2-eta1} from \cref{eq:app-eta1}, \cref{eq:app-ih2-hty2}
    from \cref{eq:app-hty22}, \cref{eq:app-ih2-ftv2} from \cref{eq:app-ftv2},
    \cref{eq:app-ih2-tevQ2} from \cref{eq:app-tevQ2}, \cref{eq:app-ih2-theta2}
    from \cref{eq:app-theta2'}, \cref{eq:app-ih2-eta2} from
    \cref{eq:app-eta2}, \cref{eq:app-ih2-ambi} from \cref{eq:app-ambi-2},
    \cref{eq:app-ih2-tau-eq} from \cref{eq:app-tau-1-eq}, and
    \cref{eq:app-ih2-ev-canon} from \cref{eq:app-ev-canon'}.

    If we now apply the \textsc{App} rule with
    \cref{eq:app-ih1-conclusion} and \cref{eq:app-ih2-conclusion}, we get our goal \cref{eq:app-goal'}
    \begin{align*}
      &\tred{(\eta_1{}(t_1')~\eta_1{}(t_2'))}{(\eta_2{}(t_1'')~\eta_2{}(t_2''))}
    \end{align*}



    \end{enumerate}

\end{proof}

\section*{Acknowledgements}
This research is partially funded by the Research Fund KU Leuven and the Agency for Innovation by Science and Technology in Flanders (IWT).
Dominique Devriese holds a postdoctoral fellowship of the Research Foundation - Flanders (FWO).

\bibliographystyle{ACM-Reference-Format}\bibliography{../paper/refs} %

\end{document}